\pdfoutput=1
\documentclass[letterpaper,12pt,titlepage,oneside,final]{book}

\usepackage{ifthen} 
\newboolean{PrintVersion}
\setboolean{PrintVersion}{false} 

\usepackage{amsmath,amssymb,amstext,amsthm,bbm}
\usepackage{multirow} 
\usepackage[pdftex]{graphicx}
\usepackage{minitoc} 
\usepackage[backend=bibtex,hyperref=true,url=false,backref=true,maxbibnames=100,style=phys,sorting=none]{biblatex}
\AtEveryBibitem{\clearfield{month}}
\DefineBibliographyStrings{english}{backrefpage  = {cited on p.},backrefpages = {cited on pp.}}
\usepackage[font=footnotesize,labelfont=bf, format=default]{caption}
\usepackage{booktabs,colortbl} 
\newcolumntype{M}[1]{>{\centering\arraybackslash}m{#1}} 
\usepackage{microtype} 
\DisableLigatures{encoding = *, family = *} 

\usepackage[usenames,dvipsnames]{xcolor} 
\definecolor{newgreen}{rgb}{0,0.5,0}
\definecolor{newblue}{rgb}{0,0,0.6}

\usepackage[letterpaper=true,pagebackref=false]{hyperref} 
\usepackage[all]{hypcap} 
\hypersetup{
    plainpages=false,       
    pdfpagelabels=true,     
    bookmarks=true,         
    unicode=false,          
    pdftoolbar=true,        
    pdfmenubar=true,        
    pdffitwindow=false,     
    pdfstartview={Fit},    
    pdftitle={Studies of Symmetries that Give Special Quantum States Right to Exist}, 
    pdfauthor={Hoan Bui Dang},   
    pdfsubject={Physics - Quantum Information and Foundations},  
    pdfnewwindow=true,      
    colorlinks=true,        
    linkcolor=newblue,         
    citecolor=newgreen,        
    filecolor=magenta,      
    urlcolor=newgreen,           
}
\ifthenelse{\boolean{PrintVersion}}{   
\hypersetup{	
    citecolor=black,%
    filecolor=black,%
    linkcolor=black,%
    urlcolor=black}
}{} 

\usepackage[capitalize,nameinlink, noabbrev]{cleveref}
\crefformat{equation}{#2(#1)#3}

\usepackage[automake,acronym,nogroupskip,nohypertypes={notation}]{glossaries}
\usepackage{enumitem} 

\let\origdoublepage\cleardoublepage
\newcommand{\clearemptydoublepage}{%
  \clearpage{\pagestyle{empty}\origdoublepage}}
\let\cleardoublepage\clearemptydoublepage

\theoremstyle{plain}
\newtheorem{theorem}{Theorem}[chapter]
\newtheorem{lemma}[theorem]{Lemma}

\theoremstyle{definition}
\newtheorem*{definition}{Definition}
\newtheorem*{example}{Example}
\newtheorem*{conjecture}{Conjecture}

\theoremstyle{remark}
\newtheorem*{remark}{Remark}
\newtheorem*{note}{Note}

\DeclareMathOperator{\GL}{GL}
\DeclareMathOperator{\SL}{SL}

\DeclareMathOperator{\ESL}{ESL}
\DeclareMathOperator{\PSL}{PSL}
\DeclareMathOperator{\F}{\mathbb{F}}
\DeclareMathOperator{\PGL}{PGL}
\DeclareMathOperator{\tr}{tr}
\DeclareMathOperator{\Tr}{Tr}
\DeclareMathOperator{\Gal}{Gal}

\newcommand{\be}{\begin{equation}}
\newcommand{\ee}{\end{equation}}

\newcommand{\bmp}{\begin{pmatrix}}
\newcommand{\emp}{\end{pmatrix}}

\newcommand{\fd}[1]{\mathbb{#1}}
\newcommand{\bra}[1]{\left\langle{#1}\right\vert}
\newcommand{\ket}[1]{\left\vert{#1}\right\rangle}
\newcommand{\abs}[1]{\left\vert{#1}\right\vert}

\newcommand{\braket}[2]{\langle#1|#2\rangle}
\newcommand{\ketbra}[1]{\vert#1\rangle\langle#1\vert}

\newcommand{\eye}{\mathbbm{1}}
\newcommand{\qs}{\ket{\psi}}
\newcommand{\qsp}{\ket{\phi}}
\newcommand{\Dp}{D_\mathbf{p}}
\newcommand{\p}{\mathbf{p}}
\newcommand{\bu}{\mathbf{u}}
\newcommand{\bv}{\mathbf{v}}
\newcommand{\bp}{\mathbf{p}}
\newcommand{\Fd}{\fd{F}_d}
\newcommand{\Fdinf}{\fd{F}_d \cup \{\infty\}}
\newcommand{\Fp}{\fd{F}_p}
\newcommand{\etal}{{\it et al }}
\newcommand{\SLbar}{SL(2,$\fd{Z}_{\bar{d}}$)~}
\newcommand{\Zdbar}{$\fd{Z}_{\bar{d}}$~}
\newcommand{\mcH}{\mathcal{H}}
\newcommand{\mcZ}{\mathcal{Z}}
\newcommand{\mcM}{\mathcal{M}}
\newcommand{\mcW}{\mathcal{W}}
\newcommand{\UZ}{U_{\mathcal{Z}}}

\newcommand{\UW}{U_{\mathcal{W}}}

\makeatletter

\newenvironment{chapquote}[2][4em]
  {\setlength{\@tempdima}{#1}%
   \def\chapquote@author{#2}%
   \parshape 1 \@tempdima \dimexpr\textwidth-2\@tempdima\relax%
   \itshape}
  {\par\normalfont\hfill--\ \chapquote@author\hspace*{\@tempdima}\par\bigskip}
\makeatother

\newglossary[nlg]{notation}{not}{ntn}{Notation} 
\makeglossaries

\newglossaryentry{d}{type=notation,name={$d$}, description={the dimensionality of the Hilbert space under consideration}, sort={1a}}
\newglossaryentry{omega}{type=notation,name={$\omega$}, 
description={$d$-th root of unity: $\omega = e^{2\pi i /d}$}, sort={1b}} 
\newglossaryentry{Zd}{type=notation,name={$\fd{Z}_d$}, 
description={the set of integers modulo $d$}, sort={1c}} 
\newglossaryentry{Fd}{type=notation,name={$\Fd$}, 
description={the finite field of order $d$ (only when $d$ is a prime power)}, sort={1d}} 

\newglossaryentry{X}{type=notation,name={$X$}, 
description={the shift operator}, sort={2a}} 
\newglossaryentry{Z}{type=notation,name={$Z$}, 
description={the phase operator},sort={2b}} 
\newglossaryentry{Du}{type=notation,name={$D_{\bu}$}, 
description={the displacement operator indexed by $\bu=(u_1,u_2) \in \fd{Z}_d^2$ or $\Fd^2$ },sort={2c}}
\newglossaryentry{Omega}{type=notation,name={$\Omega(\bu,\bv)$}, 
description={the symplectic form, or symplectic area},sort={2d}} 

\newglossaryentry{SL}{type=notation,name={$\SL(2,\F_d)$}, 
description={the special linear group (aka the symplectic group) consisting of $2 \times 2$ matrices over the field $\Fd$, whose determinants are 1}, sort={4a}} 
\newglossaryentry{GL}{type=notation,name={$\GL(2,\F_d)$}, 
description={the general linear group consisting of $2 \times 2$ matrices over the field $\Fd$, whose determinants are non-zero},sort={4b}}
\newglossaryentry{GLp}{type=notation,name={$\GL_p(2,\F_d)$}, 
description={the subgroup of $\GL(2,\F_d)$ consisting of matrices whose determinants belong to the subfield $\F_p \subset \F_d$},sort={4c}}
\newglossaryentry{ESL}{type=notation,name={$\ESL(2,\Fd)$}, 
description={the extended special linear group consisting of $2 \times 2$ matrices over the field $\Fd$, whose determinants are $\pm1$},sort={$\SL2$},sort={4d}}  	
\newglossaryentry{PSL}{type=notation,name={$\PSL(2,\F_d)$}, 
description={the projective special linear group: quotient group of SL obtained by setting SL elements $S$ and $S'$ equivalent if $S=cS'$ for some constant $c \in \Fd$},sort={4e}} 	
\newglossaryentry{PGL}{type=notation,name={$\PGL(2,\F_d)$}, 
description={the projective linear group: quotient group of GL obtained by setting GL elements $G$ and $G'$ equivalent if $G=cG'$ for some constant $c \in \Fd$},sort={4f}} 	

\newglossaryentry{Zauner}{type=notation,name={$\mathcal{Z}$}, 
description={the Zauner symplectic matrix $\mathcal{Z} = \bmp 0 & -1 \\ 1 & -1\emp$}, sort={5a}} 

\newglossaryentry{UZ}{type=notation,name={$U_{\mathcal{Z}}$}, 
description={the Zauner unitary}, sort={5b}}

\newglossaryentry{S}{type=notation,name={$S$}, 
description={generally denotes a symplectic matrix in $\SL(2,\F_d)$}, sort={5c}}  
\newglossaryentry{US}{type=notation,name={$U_S$}, 
description={generally denotes a symplectic unitary}, sort={5d}}  
\newglossaryentry{G}{type=notation,name={$G$}, 
description={generally denotes an element of $\GL(2,\F_d)$}, sort={5e}}  
\newglossaryentry{UG}{type=notation,name={$U_G$}, 
description={generally denotes a Galois-unitary}, sort={5f}}  
\newglossaryentry{Delta}{type=notation,name={$\Delta$}, 
description={the determinant of a GL element}, sort={5g}}  
\newglossaryentry{K}{type=notation,name={$K$}, 
description={complex conjugation}, sort={5h}}  
\newglossaryentry{gk}{type=notation,name={$g_k$}, 
description={the Galois automorphism mapping $\omega \mapsto \omega^k$}, sort={5i}}  

\newglossaryentry{Trace}{type=notation,name={$\Tr$}, 
description={trace of a matrix or a linear operator}, sort={6a}} 
\newglossaryentry{trace}{type=notation,name={$\tr$}, 
description={field theoretic trace}, sort={6b}} 

\newglossaryentry{real}{type=notation,name={$\fd{R}$}, 
description={the field of real numbers}, sort={7a}}
\newglossaryentry{rational}{type=notation,name={$\fd{Q}$}, 
description={the field of rational numbers}, sort={7b}} 
\newglossaryentry{complex}{type=notation,name={$\fd{C}$}, 
description={the field of complex numbers}, sort={7c}} 

\newglossaryentry{Qw}{type=notation,name={$\fd{Q}(\omega)$}, 
description={the cyclotomic field generated from $\omega$ and the rationals}, sort={7e}} 

\newglossaryentry{Qwd}{type=notation,name={$\fd{Q}(\omega)^d$}, 
description={the $d$-dimensional vector space over the cyclotomic field}, sort={7f}}

\newglossaryentry{Legendre}{type=notation,name={$l(x)$}, 
description={the Legendre symbol}, sort={9a}} 
\newglossaryentry{Q}{type=notation,name={${\bf Q}$}, 
description={the set of quadratic residues consisting of elements of a field that can be written as the square of another non-zero element},sort={9b}} 
\newglossaryentry{N}{type=notation,name={${\bf N}$}, 
description={the set of quadratic non-residues consisting of elements of a field that can not be written as the square of any other element},sort={9c}} 

\setacronymstyle{long-short}
\newacronym{mub}{MUB}{Mutually Unbiased Bases}
\newacronym[shortplural=MUS]{mus}{MUS}{Minimum Uncertainty State}
\newacronym{povm}{POVM}{Positive Operator-Valued Measure}
\newglossaryentry{sicpovm}{type = acronym, name = {SIC-POVM}, description = {Symmetric Informationally-Complete Positive Operator-Valued Measure}, long={Symmetric Informationally-Complete Positive Operator-Valued Measure}, short={SIC-POVM}, longplural = {Symmetric Informationally-Complete Positive Operator-Valued Measures}, shortplural = {SIC-POVMs}, first = {Symmetric Informationally-Complete Positive Operator-Valued Measure (SIC-POVM)} }
\newacronym{wh}{WH}{Weyl-Heisenberg}
\newcommand{\MUB}{\gls{mub}}
\newcommand{\MUBs}{\glspl{mub}}
\newcommand{\MUS}{\gls{mus}}
\newcommand{\MUSs}{\glspl{mus}}

\newcommand{\WH}{\gls{wh}}

\newcommand{\off}{\phantom{u}}
\newcommand{\twotofour}[2][]{\makebox[0pt]{\smash{\raisebox{.5\normalbaselineskip}{#2}\hspace*{4\arraycolsep}}}}

\begin{document}

\dominitoc 
\setcounter{mtc}{4} 

\pagestyle{empty}
\pagenumbering{roman}

\begin{titlepage} 
        \begin{center}
        \vspace*{1.0cm}
	\Huge
	{\bf Preface to the arXiv version}
	\normalsize
        \vspace*{1.0cm}
	 \end{center}
This thesis was accepted by the University of Waterloo (Ontario, Canada) in 2015. It concerns studies on symmetric quantum structures in finite dimensional Hilbert spaces such as symmetric informationally-complete states (SIC-POVMs) and mutually unbiased bases (MUBs). In one of the studies on MUBs, we made use of Galois-unitary operators (g-unitaries for short), which are a generalized notion of anti-unitary operators. The concept of g-unitaries was first introduced in 2013 by Appleby, Yadsan-Appleby, and Zauner in the context of the SIC-POVM existence problem. Despite certain bizarre behaviors of theirs, g-unitaries are deeply interesting and useful for solving certain problems in quantum information and they deserve to be further investigated. In making this thesis available online, I hope that it will become a resourceful introduction for those who are interested in this subject.\\

\vspace*{1cm}

\hspace*{10.4cm} Hoan Bui Dang\\
\hspace*{11cm} Waterloo, Canada, 2015\\

\end{titlepage}

\begin{titlepage}
        \begin{center}
        \vspace*{1.0cm}

        \Huge
        {\bf Studies of symmetries that give special quantum states \\the ``right to exist''}

        \vspace*{1.0cm}

        \normalsize

        \vspace*{1.0cm}

        \Large
        Hoan Bui Dang \\

        \vspace*{3cm}

        \normalsize
        A thesis \\
        presented to the University of Waterloo \\ 
        in fulfillment of the \\
        thesis requirement for the degree of \\
        Doctor of Philosophy \\
        in \\
        Physics \\

        \vspace*{2.0cm}

        Waterloo, Ontario, Canada, 2015 \\

        \vspace*{1.0cm}

        \copyright\ Hoan Bui Dang 2015 \\
        \end{center}
\end{titlepage}

\pagestyle{plain}
\setcounter{page}{2}

\cleardoublepage 

  \noindent
I hereby declare that I am the sole author of this thesis. This is a true copy of the thesis, including any required final revisions, as accepted by my examiners.

  \bigskip
  
  \noindent
I understand that my thesis may be made electronically available to the public.

\cleardoublepage


\begin{center}\textbf{Abstract}\end{center}

In this thesis we study symmetric structures in Hilbert spaces known as symmetric informationally complete positive operator-valued measures (SIC-POVMs), mutually unbiased bases (MUBs), and MUB-balanced states \cite{Appleby2014S, Dang2013, Appleby2014G}. Our tools include symmetries such as the Weyl-Heisenberg (WH) group symmetry, Clifford unitaries, Zauner symmetry, and Galois-unitaries (g-unitaries). In the study of SIC-POVMs, we found their geometric significance as the ``most orthogonal'' bases on the cone of non-negative operators. While investigating SICs, we discovered a linear dependency property of the orbit of an arbitrary vector with the Zauner symmetry under the WH group. In dimension $d=3$, the linear dependency structures arising from certain special SIC states are identified with the Hesse configuration known from the study of elliptic curves in mathematics. We provide an analytical explanation for linear dependencies in every dimension, and a numerical analysis based on exhaustive numerical searches in dimensions $d=4$ to 9. We also study the relations among normal vectors of the hyperplanes spanned by the linearly dependent sets, and found 2-dimensional SICs embedded in the Hilbert space of dimension $d=6$, and 3-dimensional SICs for $d=9$. A full explanation is given for the case $d=6$. Another study in the thesis focuses on the roles of g-unitaries in the theory of mutually unbiased bases. G-unitaries are, in general, non-linear operators defined to generalize the notion of anti-unitaries. Due to Wigner's theorem \cite{Wigner1931}, their action has to be restricted to a smaller region of the Hilbert space, which consists of vectors whose components belong to a specific number field. G-unitaries are relevant to MUBs when this number field is the cyclotomic field. In this case, we found that g-unitaries simply permuted the bases in the standard set of MUBs in odd prime-power dimensions. With their action further restricted only to MUB vectors, g-unitaries can be represented by rotations in the Bloch space, just as ordinary unitary operators can. We identify g-unitaries that cycle through all $d+1$ bases in prime power dimensions $d=p^n$ where $n$ is odd (the problem in even prime power dimensions has been solved using ordinary unitaries). Each of these MUB-cycling g-unitaries always leaves one state in the Hilbert space invariant. We provide a method for calculating these eigenvectors. Furthermore, we prove that when $d=3$ mod 4, they are MUB-balanced states in the sense of Wootters and Sussman \cite{Wootters2007} and Amburg \etal \cite{Amburg2014}.

\cleardoublepage


\begin{center}\textbf{Acknowledgements}\end{center}

First and foremost, I would like to thank my co-supervisor, Prof. Christopher Fuchs, who inspired me to study quantum information and  to investigate the SIC problem particularly. Working with Chris, whether at Bell Labs or in the QBism group at Perimeter Institute, has always been an enjoyable and motivating experience: I cannot recall a time walking out of his office without fresh new ideas or inspired thoughts. Despite the circumstance that only allowed us to work remotely with each other in the last years of my PhD program, Chris has always been behind to support me whenever I was in need.

I am deeply grateful to my collaborators, Prof. Ingemar Bengtsson and Dr. Marcus Appleby, for their mentoring and tremendous research support. Being in three different continents means that in order for us to have a discussion via Skype, Marcus would need to get ready at 6am, while Ingemar would have to stay up until midnight. They have meant much more to me than just collaborators, and I cannot thank them enough for the kindness, availability and knowledge they have been providing.

I am very thankful to my supervisor, Prof. Joseph Emerson, for being willing to take me as a student into his research group during the final year of my program and for his support to my research. I would also like to thank all the past and current members of my committee, Prof. Norbert Lutkenhaus, Prof. Debbie Leung, Prof. Lucien Hardy, and Prof. Daniel Gottesman, for their attention and many research comments and suggestions that they have been providing me throughout the course of my PhD. I am very grateful to Prof. William Wootters for his kindness in offering to come a long way from Williamstown (Massachusetts) to be an external examiner in my thesis defense, and for many detailed and helpful comments on this thesis.

My groupmates Matthew Graydon and Gelo Tabia have been my closest friends during my time at Waterloo. We have taken courses and traveled to conferences together. We have discussed physics uncountably many times (not in a set theoretic sense). They have always been available whenever I need help or just someone to talk to. I also want to thank my colleagues {\r A}sa Ericsson,  Huangjun Zhu, Kate Blanchfield, and David Andersson for many fruitful research discussions and collaborations.

Finally, nothing that I have accomplished would have been possible without my family. Whether living with me or being thousands of miles away, they always give me their full love, trust, and support. I dedicate this thesis to my family, although I know that nothing I do will be enough to thank them.

My research was financially supported in part by the U.S. Office of Naval Research (Grant No. N00014-09-1-0247), and by the Natural Sciences and Engineering Research Council of Canada via the Vanier Canada Graduate Scholarship. Research at Perimeter Institute is supported by the Government of Canada through Industry Canada and by the Province of Ontario through the Ministry of Research and Innovation.

\cleardoublepage

\phantom{a}
\vskip 4cm
{\it
In memory of my grandfather, Dr. Hoan Trong Bui (1929-2014), who named me after himself, wishing that whenever I succeed our name will be praised, and whenever I fail we will take the blame together.}
\cleardoublepage

\renewcommand\contentsname{Table of Contents}
\addtocontents{toc}{\protect{\pdfbookmark[0]{\contentsname}{toc}}}
\tableofcontents
\cleardoublepage
\phantomsection

\addcontentsline{toc}{chapter}{List of Tables}
\listoftables
\cleardoublepage
\phantomsection		

\addcontentsline{toc}{chapter}{List of Figures}
\listoffigures
\cleardoublepage
\phantomsection		

\addcontentsline{toc}{chapter}{\textbf{List of Acronyms}} 
\setlist[description]{leftmargin=!, labelwidth=6.5em} 
\printglossary[type = acronym, title={List of Acronyms}, nopostdot]
\setlist[description]{style=standard} 
\cleardoublepage
 \phantomsection 

\addcontentsline{toc}{chapter}{\textbf{List of Notations}} 
\setlist[description]{leftmargin=!, labelwidth=6.5em} 
\printglossary[type = notation, title={List of Notations}, nopostdot]
\setlist[description]{style=standard} 
\cleardoublepage
 \phantomsection 

\pagenumbering{arabic}

\chapter{Introduction}\label{chap:introduction}

\begin{chapquote}{Philip Anderson (1972)}
``It is only slightly overstating the case to say that physics is the study of symmetry.''
\end{chapquote}

\section{Overview}\label{sec:motivation}
A symmetry is a property of an object that remains the same under certain transformations. Although this might sound like a purely mathematical concept, symmetries can be found almost everywhere in the world surrounding us. For example, a bicycle has a left-right reflection symmetry: its left half is (mostly) the mirror image of its right half. Without this property, an unbalanced bicycle might be unpleasant to ride. A circular shape has a full rotational symmetry: if someone rotates your round dinner plate about its center by an arbitrary angle while you are away, you will not be able to tell the difference when you come back. Not only have we all made use of this property when we learned how to use a compass in elementary school, or when we played ball games as kids, we are all now living in a world of modern machinery that is largely based on inventions with rotational symmetry, such as wheels and gears. Nature is full of symmetries as well. The bodies of most animals have a bilateral (left-right) symmetry. Flowers often have radially or bilaterally symmetric shapes, which have been found to aid bumble bees in their foraging process \cite{West1998}. In general, one can find symmetries from a minuscule scale such as in atomic or molecular lattices, all the way to the cosmic scale such as in galaxies that are hundreds of thousands light years in diameter across.

In physics, symmetry plays a fundamental role. Studying physical phenomena can be broken down into two components \cite{Gross1996}. The first component is the given initial conditions, which might be very complicated and unpredictable, and therefore they have to be ``given,'' i.e. there is not much we can do but to accept them as they are. The second component consists of rules that capture all the patterns and regularities that are independent of the initial conditions. This is where the physics lies. In other words, when we say we understand the physics of a phenomenon, it means we have figured out what some of these rules are, and the more rules we have found, the more deeply have we understood. In this sense, one major goal of physics is to discover rules of regularities that can be applied to a broad range of phenomena. However, this is a very difficult task, as regularities are often buried under a vast amount of irregularity from the initial conditions. For example, who would have thought that there is a similarity between the elliptic orbits of planets in the Solar system and the falling of an apple? Or who would have thought that electric fields and magnetic fields can be transformed into each other, and how this could start a train of thoughts leading to the explanation of the perihelion precession of Mercury's orbit? This is where symmetry comes to help.

Symmetries in physics, often coming in the form of invariance or equivalence principles, help filter out the irrelevant complications to reveal the regularities at the heart of physical phenomena. For example, the weak equivalence principle (also known as the Galilean equivalence principle, which lays the foundation for theories of gravity) states that the trajectory of a point mass in a gravitational field depends only on its initial position and velocity, and not on its mass or composition. In this example, the principle was deduced from experimental observations by Galileo in the late 16th century. However, as symmetries have become an increasingly powerful tool, it often is the case that symmetries dictate the laws in modern physics and even provide predictions that predate experimental discoveries. For example, Lorentz symmetry formed the backbone of relativity and led to the derivation of Dirac's equation and the prediction of anti-particles. The gauge symmetries underlay the development of electromagnetism, quantum electrodynamics, quantum chromodynamics, and the Standard Model. The symmetry of exchanging identical particles in quantum mechanics classified all elementary particles into bosons or fermions whose behaviors are very distinctive. The list can go on. However, the intimate connection between symmetry and physics is not merely based on historical evidence. It has been rigorously proved that every continuous symmetry of the action of a physical system implies a conserved physical quantity, a result known as Noether's theorem \cite{Noether1918}. As Philip Anderson has put it, ``it is only slightly overstating the case to say that physics is the study of symmetry.''

In this thesis, we are interested in the study of symmetries in quantum physics. Quantum theory is considered one of the most successful theories in physics in many different ways: 1) no experiment has ever violated its predictions, 2) the theory has provided the most accurate experimental tests to date, for example the determination of the fine structure constant $\alpha$ with an agreement to one part per billion \cite{Odom2006}, 3) the theory is applied in most areas of modern physics including condensed matter physics, atomic, molecular and optical physics, particle physics, astrophysics etc., and finally 4) it has a huge impact on today's world, with a wide range of applications such as transistors for computing devices, lasers, light emitting diodes, liquid crystal displays, nuclear magnetic resonance, and magnetic resonance imaging, just to name a few. On the other hand, quantum theory is also considered one of the strangest theories. It has been developed for over a century now, but many questions since its birth are still under debate: what is the nature of the wave function? do wave functions collapse? is the theory non-local because of ``spooky action at a distance''? and many more. Indeed, the field of quantum foundations is still an active research area, and many questions have to be answered before quantum theory can be fully comprehended.

During the last few decades, research in quantum foundations has received a boost from developments in the new field of quantum information. Quantum information makes use of special features in quantum theory to help accomplish information-related tasks that are impossible using classical physics. For example, entanglement is used in superdense coding \cite{Bennett1992} and quantum teleportation \cite{Bennett1993}. Another example is quantum key distribution \cite{Bennett2014}, which relies on the quantum information-disturbance trade-off to help generate and distribute secure encryption keys. At the same time, quantum information brings tools from information theory into quantum physics, and helps provide us with an information theoretic framework to study quantum theory. An example of this is the quantification of quantum information using von Neumann entropy, which is an analogue of the Shannon entropy used in classical information theory \cite{Nielsen2010}.

The research presented in this thesis arises from problems in quantum information involving various symmetric structures in the space of quantum states such as SICs, MUBs, and MUB-balanced states (their definitions will be provided later). These structures display such a high degree of symmetry that makes it seem as though they have no ``right to exist,'' as Amburg \etal \cite{Amburg2014} have described MUB-balanced states. Our hopes in investigating these problems are not only to make use of their symmetries to discover new properties and new applications in quantum information, but also to gain a deeper understanding about symmetries in quantum state spaces and quantum theory. The content of the thesis is organized as follows.

\section{Organization of the thesis}\label{sec:organization}

The thesis consists of various studies, from symmetric informationally complete states and linear dependency structures in Weyl-Heisenberg orbits, to the study of Galois-unitaries with applications to the theory of mutually unbiased bases. The results of these studies are organized into two main chapters according the relevant symmetry: \cref{chap:WH} contains results related to the Weyl-Heisenberg symmetry, as well as Clifford unitaries and Zauner symmetry, and \cref{chap:gunitary} provides the results from a study of Galois-unitary symmetry.

\cref{chap:WH} starts with a historical introduction of the Weyl-Heisenberg group in \cref{sec:historical}. We then focus on a class of symmetric structures in the Hilbert space known as SICs in \cref{sec:SICs}. We provide a brief history of the development of the SIC problem, and an extensive list of SICs' applications and major known results. We discuss group symmetries that are intimately related to SICs such as Weyl-Heisenberg covariance, Clifford group, and Zauner symmetry. We show that SICs form the most orthogonal bases on the cone of non-negative operators. Then, in \cref{sec:linde}, we present our results from studies of linear dependencies in Weyl-Heisenberg orbits, which include an analysis in dimension $d=3$ and the connection to elliptic curves via Hesse configuration, an analytical explanation of linear dependencies in all dimensions where the initial vector is an eigenvector of the Zauner unitary, a detailed numerical report in low dimensions which shows extra linearly dependent relations that cannot be accounted for by our theorem, and a robust construction of ``small SICs'' resulted from the linear dependency structure.

\cref{chap:gunitary} contains our results from the study of a novel symmetry called Galois-unitary, applied to the theory of mutually unbiased bases. \cref{sec:motivations} describes the motivations for our study, including the context in which g-unitaries were first constructed. We first introduce mutually unbiased bases and describe their Clifford-based construction in \cref{sec:MUBs}. We then provide a representation for the general linear group, using Clifford group extended by g-unitary operators in \cref{sec:gunitaries}. The treatment is divided into cases, when the dimension $d=p$ is an odd prime, and when $d=p^n$ is an odd prime power. We proved that the representation is faithful if $n$ is odd, and is ``almost'' faithful if $n$ is even. We provide some basic arithmetic of g-unitaries in \cref{sec:arithmetic}. In \cref{sec:geometric} we describe a type of  geometric object called complementarity polytopes and use their symmetry groups to provide a geometric interpretation of g-unitaries. \cref{sec:simulation} proposes a scheme to simulate g-unitaries using unitary operators in a larger Hilbert space. The MUB-cycling problem is discuss in \cref{sec:MUBcyclers}, in which we prove that MUB-cyclers exist in every odd prime power dimensions $d=p^n$ where $n$ is odd, and they do not exist when $n$ is even. We also provide a characterization of all MUB-cyclers when they do exist. In \cref{sec:eigenvectors}, we prove that every MUB-cycler has a unique (up to a phase) eigenvector, and provide a way to calculate this eigenvector. In \cref{sec:MUBbalanced}, we show that when $d=3$ mod 4, the eigenvectors of MUB-cyclers are MUB-balanced states, that they form a single orbit under the extended Clifford group, and that they are identical to those constructed in Amburg \etal \cite{Amburg2014}.

\cref{chap:summary} provides a summary of our main results, and suggests a list of open problems and ideas for future investigation.

\cref{sec:fieldtheory} provides an introduction to field theory, which covers the basic concepts used in the thesis such as field extensions, Galois automorphisms, finite fields, cyclotomic fields, etc. \cref{sec:Clifford-app} describes faithful representations of the Clifford group specifically for the case of prime and prime power dimensions.

\section{List of specific contributions}\label{sec:contributions}

Results in this thesis that represent my own specific contributions include:
\begin{enumerate}
\item \cref{sub:orthogonal}, published in \cite{Appleby2014S}.
\item \cref{sub:analytical,sub:numerical,sub:smallSICs}, published in \cite{Dang2013}.
\item \cref{sec:arithmetic,sec:simulation}, from my own research notes.
\item Numerical analysis of g-unitary rotations leading to \cref{sec:geometric}, published in \cite{Appleby2014G}.
\item \cref{lem-Am,thm-GLtypes} in \cref{sec:MUBcyclers}, published in \cite{Appleby2014G}.
\item Contribution to the proof of \cref{thm-eigenvectors,lem-UG2m0,lem-dimS1,lem-UGeigenvector} in \cref{sec:eigenvectors}, published in \cite{Appleby2014G}.
\item Contribution to the proof of \cref{thm-MUBbalanced} in \cref{sec:MUBbalanced}, published in \cite{Appleby2014G}.
\end{enumerate}

\chapter{Weyl-Heisenberg group symmetry}\label{chap:WH}
\minitoc
\section{Historical background}\label{sec:historical}

The use of group theory in quantum mechanics dates back to the very early days of the theory. In 1925, Hermann Weyl learned from Born the recent developments in quantum mechanics made by Born, Jordan, and Heisenberg, and he immediately tried his own approach  from the perspective of the representation theory of groups \cite{Scholz2007}. This work was published in 1927 \cite{Weyl1927}, and further developed in his book \cite{Weyl1931}, in which Weyl used the ray representations of the Abelian group of rotations to develop a quantum formalism that is applicable to both finite and infinite dimensions.

Weyl realized that the canonical commutation relation between the position operator $\hat{q}$ and momentum operator $\hat{p}$ (this is sometimes called the Heisenberg commutation rule, although it first appeared in a paper by Born and Jordan \cite{Born1925,Fedak2009})
\be\label{eq-Heisenberg}
[\hat{q},\hat{p}] = i
\ee
does not admit finite-dimensional representations (we have set the unit to $\hbar$). In other words, for any dimension $d$ that is finite, there do not exist $d \times d$ matrices $Q$ and $P$ that can satisfy \cref{eq-Heisenberg}. One can see that by taking the trace of \cref{eq-Heisenberg} and observing that the left hand side is zero due to the cyclic property of the trace function, while the right hand side is non-zero. Moreover, $\hat{p}$ and $\hat{q}$ are unbounded operators and they are not defined on the whole Hilbert space. 

For a finite dimension $d$, Weyl instead introduced Hermitian matrices $P$ and $Q$, which are defined by
\be
P = \frac{1}{i\alpha} \log X \hskip 15mm
Q = \frac{1}{i\beta} \log Z
\ee
so that 
\be
X = e^{i\alpha P} \hskip 15mm
Z = e^{i\beta Q},
\ee
where $X$ and $Z$ are $d \times d$ unitary matrices satisfying Weyl's commutation relation
\be\label{eq-Weyl}
XZ = \omega^{-1}ZX \hskip 15mm \omega = e^{2\pi i /d}.
\ee
The commutator $[Q,P]$, as one takes the limit $d \rightarrow \infty$ while keeping $\alpha\beta = 2\pi/d$, can be calculated to be \cite{Santhanam1977}
\be
[Q,P]_{r,s} = i\delta(r-s),
\ee
which recovers Heisenberg commutation relation. For a detailed construction of quantum mechanics in finite dimensions based on Weyl's commutation relation \cref{eq-Weyl}, we refer the readers to a series of papers by Jagannathan and Santhanam {\it et al} \cite{Santhanam1976, Santhanam1977, Jagannathan1981, Jagannathan1982, Santhanam1982}. Here, let us focus on how Weyl constructed unitary $X$ and $Z$ in finite dimensions that satisfy such a commutation relation. The following argument should not be taken as a rigorous mathematical derivation, but should rather be considered as a train of thoughts leading to the construction of the Weyl-Heisenberg group.

Let $X$ and $Z$ be two elements of the group of unitary rotations in a $(d-1)$-dimensional ray space, meaning that they are $d \times d$ unitary matrices. We further assume that they satisfy the following commutation relation
\be\label{eq-XZ} XZ = \omega^{-1} ZX,
\ee
where $\omega=e^{2\pi i/d}$ is a primitive $d$-th root of unity. It follows that
\be
X^j Z^k = \omega^{-jk}Z^k X^j
\ee
 for all integers $j$ and $k$. If either $j$ or $k$ is equal to $d$, then $\omega^{-jk}=1$ and it follows that $X^d$ commutes with $Z$ and $Z^d$ commutes with $X$. Under the extra assumption that the representation is irreducible, by Schur's lemma we conclude that
\be
X^d = Z^d = \eye.
\ee
We can choose a basis in which $Z$ is diagonal and write it in the form
\be \label{def-Zmatrix}
Z = \bmp 1& 0& 0&\cdots& 0 \\ 0& \omega& 0&\cdots& 0 \\ 0& 0& \omega^2& \cdots & 0\\ \vdots& \vdots& \vdots& \ddots& \vdots \\ 0& 0& 0&\cdots& \omega^{d-1} \emp.
\ee
In this basis, $X$ takes the form of a cyclic permutation matrix
\be \label{def-Xmatrix}
X = \bmp 0& 0& \cdots& 0& 1 \\ 1& 0& \cdots& 0& 0 \\ 0& 1& \cdots& 0& 0 \\ \vdots& \vdots & \ddots &\vdots & \vdots\\0& 0& \cdots& 1& 0 \emp.
\ee
The two operators $X$ and $Z$ we arrive at are known as the shift and the clock (or phase) operators, respectively. They were introduced by Sylvester in 1882 in the very early days of matrix theory \cite{Sylvester2012}. Here they naturally arise from the construction of a group-based quantum theory built upon Weyl's commutation relation.

The set of operators of the form $\omega^i X^j Z^k$, with $i,j$ and $k$ taking integer values in the range $[0,d-1]$, forms a group under ordinary matrix multiplication. This group is called the discrete \WH ~group (to be distinguished from the continuous WH group in infinite dimensions, although we will drop the label ``discrete'' from now on, as it should be clear from the context of the thesis that we are working in finite dimensions only). The group elements
\be 
D_{j,k}=X^j Z^k
\ee 
are called displacement operators (please note that displacement operators may be defined with different phases to suit different situations, and we will make the definition precise when it comes to each situation). 

Besides its structural role in the foundations of quantum mechanics, the \WH~group has found applications and connections to many other fields of science. For example, in modern mathematics, it naturally appears in the theory of elliptic curves and theta functions \cite{Mumford2007}. In classical signal processing, WH group is used in the development of adaptive radar and error-correcting codes in communications \cite{Howard2006}. In quantum information, this group is also known as the generalized Pauli group, and it has numerous applications, for example in superdense coding \cite{Bennett1992}, quantum error correction \cite{Gottesman1997,Gottesman1998} and the theory of mutually unbiased bases ~\cite{Appleby2009, Thas2009, Appleby2014S, Blanchfield2014}. Its intimate relation to SIC-POVMs will be discussed in the next section.

\section{Symmetric informationally-complete POVMs}\label{sec:SICs}
The Weyl-Heisenberg group symmetry plays an indisputable role in the studies of a special class of symmetric structures in quantum state space known as \glspl{sicpovm}, or SICs for short. We will prove a number of results about SICs in this thesis. Moreover, the SIC problem was part of the motivation for the study of linear dependencies in \cref{sec:linde} and g-unitaries in \cref{chap:gunitary}. We devote this section to give an introduction to SIC-POVMs and their relation to the WH group, and to present our result on SICs being the closest to orthonormal bases with respect to a class of orthogonality measures.

\subsection{Definitions, significances, and the existence problem}\label{sub:SICexistence}
There is more than one way to define a SIC-POVM. We start with the one that explains the meaning of its name.
\begin{definition}
A set of $n$ Hermitian operators $\{E_i\}_{i=1}^n$ on a $d$-dimensional Hilbert space is called a \gls{povm} if they satisfy
\be \label{def-POVM1} E_i \ge 0 \ee for all $i = 1,2,...,n$ and
\be \label{def-POVM2} \sum_i E_i = \eye. \ee

\end{definition}

\begin{example}
The projection operators $P_i = \ketbra{i}$ of a projective (Von Neumann) measurement, where $\ket{i}$ are states constituting an orthonormal basis in a $d$-dimensional Hilbert space, form a POVM of $d$ elements.
\end{example}

A POVM $\{E_i\}_{i=1}^n$ can be thought of as a generalized quantum measurement, with $n$ outcomes labeled by $i$, whose probabilities are given by the Born rule
\be \label{eq-Bornrule}
p(i) = \Tr(\rho E_i).
\ee
One can see that condition \cref{def-POVM1} in the definition is to enforce that all the probabilities are non-negative, while condition \cref{def-POVM2}, often called the completeness condition, guarantees that they add up to one, as should be the case for a normalized probability distribution.
\begin{definition}
A POVM is said to be informationally complete if the unknown measured quantum state $\rho$ is completely specified by the measurement outcome probabilities $p(i)$.
\end{definition}

Informational completeness is equivalent to saying that the POVM elements $E_i$ span the space of Hermitian operators regarded as a $d^2$-dimensional real vector space equipped with the Hilbert-Schmidt inner product
\be \langle H_1,H_2 \rangle = \Tr(H_1 H_2).\ee
An informationally complete POVM therefore must have a minimum of $d^2$ elements.

\begin{definition}
A SIC-POVM is a POVM with $d^2$ elements $\{\Pi_i/d\}_{i=1}^{d^2}$, where $\Pi_i$ are rank-1 projection operators satisfying the symmetric property
\be \label{def-SIC}
\Tr(\Pi_i \Pi_j) = \alpha \hskip 5mm \forall i \ne j
\ee
for some constant $\alpha$.
\end{definition}
\begin{note}
We loosely call $\{\Pi_i\}$ a SIC, even though the POVM elements are technically $\Pi_i/d$.
\end{note}
There are a few things one can quickly deduce from the definition of SIC-POVMs. First, the value of the constant $\alpha$ can be determined from the dimension of the Hilbert space. From the POVM completeness condition, we have
\be
\sum_{i,j=1}^{d^2} \Pi_i \Pi_j = \left(\sum_{i=1}^{d^2} \Pi_i\right)^2 = d^2 \eye.
\ee
Taking the trace of both sides and making use of the symmetric property, one finds
\be \alpha = \frac{1}{d+1} ~.\ee

Secondly, although it is not explicit, the definition of a SIC-POVM implies that it is informationally complete. To see this, we will first show that the operators $\Pi_i$ are linearly independent. Suppose
\be \label{eq-biPii}\sum_i b_i \Pi_i = 0
\ee
for some set of numbers $b_i$. Multiplying both sides of the equation by $\Pi_k$ for some $k$ and then taking the trace, we obtain
\be b_k+ \alpha\sum_{i\ne k} b_i = 0 .\ee
On the other hand, $\Pi_i$ have unit trace as they are rank-1 projectors, so just taking the trace of \cref{eq-biPii} yields
\be \sum_i b_i = 0 .\ee
Given that $\alpha \ne 1$, it follows that 
\be \label{eq-lininde}
b_k=0 \hskip 5mm \forall k
\ee
and $\Pi_i$ are indeed linearly independent. There are $d^2$ of them, so they span the $d^2$-dimensional space of Hermitian operators. Thus, the POVM is informationally complete.

If one prefers to work with quantum states rather than with quantum measurements, there is an alternative definition of SIC-POVMs.

\begin{definition}
A set of $d^2$ normalized quantum states $\{\ket{\psi_i}\}_{i=1}^{d^2}$ is called a SIC set if it has a constant overlap between any two distinct states:
\be \label{def-SIC2}
\abs{\braket{\psi_i}{\psi_j}} = \frac{1}{\sqrt{d+1}} \hskip 5mm \forall i \ne j.
\ee
\end{definition}
If we define the projection operators $\Pi_i = \ketbra{\psi_i}$, then they are linearly independent and they span the space of Hermitian operators (following the same argument as before). Therefore the identity matrix can be written as a linear combination
\be \eye = \sum_i c_i \Pi_i.
\ee
Using the previous trick of taking the trace of the equation above and taking its trace after multiplying both sides by some $\Pi_k$, one can show that $c_k = 1/d$ for all $k$. So the set $\{\Pi_i/d\}$ is a POVM, and the two definitions of SIC-POVMs are indeed equivalent.

The second definition has a geometrical interpretation. Any vector $\qs \in \fd{C}^d$ spans a one-dimensional subspace of $\fd{C}^d$ called a line, which consists of all vectors of the form $a\qs$ for any scalar $a$. If two lines are represented by normalized vectors $\ket{\psi_1}$ and $\ket{\psi_2}$, then the angle $\theta$ between the two lines is given by
\be
\cos \theta = \abs{\braket{\psi_1}{\psi_2}}.
\ee
This means that the lines represented by vectors in a SIC set have a constant pairwise angle. Such lines are called equiangular lines, and a SIC set therefore is a set of equiangular lines, not just any set but a maximal one (there cannot be more than $d^2$ elements in the set because of the linear independence and the dimensionality of the space of Hermitian operators). The question is: do SICs exist in every dimension $d$?

Mathematicians have long been interested in figuring out the largest number of equiangular lines a vector space can admit (let us denote that number by $N(V)$, where $V$ refers to the vector space), and in constructing these maximal sets of equiangular lines. 

We start with the simplest type of vector spaces: the real ones $\fd{R}^d$. In dimension $d=2$, one can draw a maximum of 3 equiangular lines on a 2-dimensional plane (imagine the three hands of a watch at 20 minutes past 8 o'clock), so $N(\fd{R}^2) = 3$. In three dimensions, $N(\fd{R}^3) = 6$ and the 6 equiangular lines can be constructed by connecting antipodal vertices of a regular icosahedron, one of the five Platonic solids. This result, together with $N(\fd{R}^4) = 6$, has been known since 1948 in work by Haantjes \cite{Haantjes1948}. Lint and Seidel further investigated the problem, and obtained results in a number of higher dimensions \cite{Lint1966}:
\be
N(\fd{R}^5) = 10, \hskip 5mm
N(\fd{R}^6) = 16, \hskip 5mm
N(\fd{R}^7) = 28.
\ee
Lemmens and Seidel's paper in 1973 \cite{Lemmens1973} contains some important results about real equiangular lines. One result is Gerzon's theorem, which provides an upper bound for $N(\fd{R}^d)$
\be \label{eq-NRdbound}
N(\fd{R}^d) \le d(d+1)/2.
\ee
This actually can be seen from the linear independence argument in \cref{eq-lininde} together with the fact that a real $d \times d$ orthogonal matrix is specified by $d(d+1)/2$ real parameters. When the bound is saturated, one can calculate the angle to be
\be \label{eq-angle-real}
\cos\theta = \frac{1}{\sqrt{d+2}}.
\ee
Another key result, mentioned in \cite{Lemmens1973} as P. Neumann's theorem, states that if there exist $m$ equiangular lines in $\fd{R}^d$, where $m>2d$, and if the pairwise angle among them is $\theta$, then $(\cos \theta)^{-1}$ is an odd integer. Together with \cref{eq-angle-real}, this implies that for $d>3$, a necessary condition for the bound in \cref{eq-NRdbound} to be achieved is that $d=a^2 - 2$ for some odd integer $a$. The converse is not true, for example in dimension $d= 47 = 7^2-2$, where it has been proved that the bound cannot be achieved \cite{Makhnev2002}.

We will skip the detailed developments of this problem during the last 40 years or so, but we want to note that many dimensions have been investigated, many sets (not necessarily maximal) of real equiangular lines have been constructed, and many improvements have been made to the bounds of $N(\fd{R}^d)$ since 1948. However, the exact value for $N(\fd{R}^d)$ largely remains unknown even in small dimensions, and as of now in 2015, this is still an on-going line of research \cite{Tremain2008,Greaves2014}.

One would have thought that if the problem has been so difficult for the case of real vector spaces, its counterpart in complex vector spaces would be hopeless. Perhaps this is why the problem of complex equiangular lines did not get a lot of attention until much later. Surprisingly, the complex version of the problem seems to be more tractable than the real one. This is just to say that $N(\fd{C}^d)$ seems to be of a nice and simple form. Proving so, on the other hand, is a totally different matter and is in fact one of the most challenging open problems in quantum information and algebraic combinatorics. 

The upper bound for the number of complex equiangular lines
\be
N(\fd{R}^d) \le d^2
\ee
was proved by Delsarte \etal in their 1975 paper \cite{Delsarte1975}, which also mentioned that the bound can be saturated for $d=2$ and 3 without giving further details. Hoggar later provided solutions to the complex equiangular lines problem for $d=2,$ 3 and 8 in 1982 \cite{Hoggar1982}. In 1999, Zauner introduced the problem in the context of quantum designs in his PhD thesis \cite{Zauner1999}, in which he conjectured that
\be
N(\fd{C}^d) = d^2
\ee
for every dimension $d$, and gave concrete constructions for dimensions $d=2$ to 5. 

The term SIC-POVM was coined in a paper by Renes \etal in 2003 \cite{Renes2004}, in which they constructed SICs numerically all the way up to dimension $d=45$, thereby adding considerable weight to Zauner conjecture. In addition to its geometric name ``complex equiangular lines'', SICs are also known as ``minimal spherical 2-designs'' in studies of quantum t-designs for quantum information theory, and as ``equiangular tight frames'' in the theory of signal processing for engineering. They have practical applications in quantum tomography \cite{Rehacek2004, Grassl2005, Scott2006, Ballester2007, Zhu2011, Petz2012, Kalev2012}, quantum cryptography and communication \cite{Englert2008, Renes2005, Du2006, Durt2008, Fuchs2003, Fuchs2007, Kim2007, Bodmann2007, Oreshkov2011}, and radar and classical signal processing \cite{Howard2006, Bodmann2008, Herman2009, Oreshkov2011, Balan2009}. In addition, SIC-POVMs play important roles in foundational studies in quantum physics such as the QBist interpretation of quantum mechanics \cite{Fuchs2010, Appleby2011P, Fuchs2011, Fuchs2013, Gelo2013}. They also have deep mathematical connections to Lie algebras \cite{Appleby2011, Appleby2013G}, elliptic curves \cite{Hughston2007, Bengtsson2010}, and Galois theory \cite{Appleby2013}. Despite an enormous amount of research on SIC-POVMs in the recent years \cite{Zauner1999, Renes2004, 
Rehacek2004, Grassl2005, Scott2006, Ballester2007, Zhu2011, Petz2012, Kalev2012,
Englert2008, Renes2005, Du2006, Durt2008, Fuchs2003, Fuchs2007, Kim2007, Bodmann2007, Oreshkov2011,
Howard2006, Bodmann2008, Herman2009, Oreshkov2011, Balan2009, 
Fuchs2010, Appleby2011P, Fuchs2011, Fuchs2013, Gelo2013, Bengtsson2012,
Appleby2011, Appleby2013G, Hughston2007, Bengtsson2010, Appleby2013,
Scott2010, Appleby2012, Appleby2014Systems,
Appleby2005, Klappenecker2005, Colin2005, Flammia2006, Khatirinejad2007, Bos2007, Appleby2007, Appleby2014S, Grassl2008, Grassl2009,  Appleby2009S, Bengtsson2009, Fickus2009, Godsil2009, Zhu2010, Filippov2010,  Medendorp2011, Tabia2012, Rastegin2014, Zhu2010T, Gour2014}
and strong numerical evidences of their existence (published for every dimension up to $d=67$ \cite{Scott2010}) as well as analytical solutions in many small dimensions ($d=2$ to 15, 16, 19, 24, 28, 35, and 48 \cite{Scott2010, Appleby2012, Appleby2014Systems}), a general analytical construction or an existence proof for all dimensions is still missing.

\subsection{Weyl-Heisenberg group covariance}

One may attempt to find a SIC set by solving the defining system of equations \cref{def-SIC2}. But one would quickly realize that this set of non-linear equations is highly over-constrained: including normalization there are a total of $d^4$ equations, while the number of real parameters needed to describe $d^2$ vectors in $\fd{C}^d$ is $2d^3$. Taking the conjugate symmetry of the inner product into account, the number of equations reduces to $(d^4+d^2)/2$, but that is still one order in $d$ higher than the number of variables. It seems that the existence of a solution must be a miracle. This is what it means when we say SIC states have no ``right to exist'' (the expression is borrowed from Amburg \etal \cite{Amburg2014}, in which they actually talk about MUB-balanced states). However, given strong evidences of SIC states' existence, there is another viewpoint one could take: their existence is not a miracle, but rather an indication of deep underlying symmetries. In fact, one such symmetry has been observed, namely the Weyl-Heisenberg symmetry.\\

\begin{definition}We define the shift operator \gls{X} and the phase operator \gls{Z} by their action on the basis states $\{\ket{x}\}_{x=0}^{d-1}$ in a $d$-dimensional Hilbert space:
\be X\ket{x} = \ket{x+1} \hskip 15mm Z\ket{x} = \omega^x\ket{x}, \ee
where $\omega = e^{2 \pi i/d}$ is a primitive $d$-th root of unity, and the arithmetic inside Dirac's kets is modulo $d$. $X$ and $Z$ in matrix form are given in \cref{def-Xmatrix} and \cref{def-Zmatrix}.\\
\end{definition}
\begin{definition}
The Weyl-Heisenberg displacement operators \gls{Du}, labeled by a two-component vector $\bu$, are defined to be
\be \label{def-Du}
D_{\bu} \equiv \tau^{u_1 u_2}X^{u_1}Z^{u_2} \hskip 15mm \bu=\bmp u_1 \\ u_2 \emp, \ee
where $\tau = -e^{\pi i /d}$ and the two components $u_1$ and $u_2$ are integers modulo $\bar{d}$, which is defined to be
\be \label{def-dbar}
\bar{d} \equiv \begin{cases}
d & \text{if $d$ is odd} \\
2d & \text{if $d$ is even}. \end{cases}
\ee
to allow us to conveniently deal with both cases of odd $d$ and even $d$ at the same time \cite{Appleby2005}. Note that $\tau^{\bar{d}}=1$, $D_{\bu} = D_{\bv}$ if and only if $\bu = \bv \mod \bar{d}$, and $D_\bu$ are all traceless except when $\bu = \mathbf{0}$ mod $d$.
\end{definition}
The particular choice of the phase factors $\tau^{u_1u_2}$ in the definition above is so that
\be D_{\bu}^{\dagger} = D_{-\bu}, \ee
\be \label{eq-grouplaw}
D_{\bu} D_{\bv} = \tau^{\Omega(\bu,\bv)} D_{\bu+\bv}, \ee
where 
\be \Omega(\bu,\bv) \equiv u_2 v_1 - u_1 v_2 \ee 
is the symplectic form of $\bu$ and $\bv$.\\

\begin{definition}  The \gls{wh} group is defined to be the set of operators
\be \label{def-WH}
\mathcal{W}_d = \{\tau^s D_{\bu}: s \in \fd{Z}_{\bar{d}}, \bu \in \fd{Z}_{\bar{d}}^2\}.
\ee
\end{definition}
This is a group under matrix multiplication, with the group law given by \cref{eq-grouplaw}. Although $\mathcal{W}_d$ technically has $d^3$ elements if $d$ is odd or $8d^3$ elements if $d$ is even, if we ignore overall phases of its elements this number reduces to $d^2$. For example, the Weyl-Heisenberg (WH) orbit of a given quantum state, i.e. a set of states obtained by applying all the displacement operators to the initial state, can only have at most $d^2$ distinct states because the overall phases of quantum states carry no physical meaning. For this reason, from now on when we consider a WH orbit, we will use  $\fd{Z}_d^2$ instead of $\fd{Z}_{\bar{d}}^2$ to index the displacement operators.\\

\begin{definition}
A SIC set is said to be Weyl-Heisenberg covariant if it is an orbit under the WH group. In other words, it can be written as $\{D_{\bu} \qs : \bu \in \fd{Z}_d^2 \}$ for some $\qs$. Such a state $\qs$ is called a Weyl-Heisenberg SIC fiducial.
\end{definition}

If we assume WH covariance, the problem of finding a SIC set turns into the problem of finding a single normalized fiducial state $\qs$ such that
\be \label{def-SICWH}
\abs{\bra{\psi} D_{\bu} \qs}^2 = \frac{1}{d+1} \hskip 10mm \forall \bu \in \fd{Z}_d^2 \backslash \{\mathbf{0}\}.
\ee
One can see that the Weyl-Heisenberg symmetry helps reduce the number of equations to the order of $d^2$. Although the number of (real) variables is now $2d$, and this system of equations is still over-constrained, it significantly simplifies the problem. Almost all known analytical solutions to the SIC problem, as well as all numerical solutions to date \cite{Renes2004,Scott2010}, are Weyl-Heisenberg covariant, with one exception being the construction in $d=8$ by Hoggar \cite{Hoggar1982}, which is covariant with respect to a 3-fold tensor product of WH groups for $d=2$. It has even been proved that in prime dimensions, if a SIC set with group covariance exists, the group must be the WH group \cite{Zhu2010}. It therefore looks as though Weyl-Heisenberg covariance is an intrinsic symmetry of SIC-POVMs.\\

\subsection{Clifford unitaries and Zauner symmetry} \label{sub:Clifford}
On top of the Weyl-Heisenberg covariance, another order-3 symmetry on SIC-POVMs was observed by Zauner \cite{Zauner1999} and later explicitly worked out by Appleby \cite{Appleby2005}. Before we get there, we first need to define the Clifford group and provide a unitary representation.
\begin{definition}
The Clifford group $\mathcal{C}_d$ is defined to be the normalizer of the Weyl-Heisenberg group $\mathcal{W}_d$ within the unitary group $U(d)$. In other words, a unitary operator $U$ belongs to $\mathcal{C}_d$ if and only if
\be
U \mathcal{W}_d U^{\dagger} = \mathcal{W}_d.
\ee
\end{definition}
\begin{remark}
It should be mentioned that there are several different versions of the Weyl-Heisenberg and Clifford groups. The version that we have just defined is the ordinary one, which can be defined for any dimension $d$. This ordinary version is relevant to the current chapter when we discuss SICs and linear dependencies in WH orbits. In \cref{chap:gunitary}, we will use another version which is applicable only to odd prime power dimensions, where one takes advantage of finite fields to define a Galoisian variant of the WH group. Corresponding to this Galoisian WH group are two Galoisian variants of the Clifford group: the full and the restricted one \cite{Appleby2009}, of which we will only use the latter. To avoid confusion, the definitions of the Galoisian WH and Clifford groups are put in \cref{sec:Clifford-app}. For now, in this chapter, we use the name WH and Clifford groups to refer to the ordinary version.
\end{remark}
Clifford unitaries can be constructed from symplectic matrices \cite{Appleby2005}. We define the symplectic group \SLbar to be the set of all $2 \times 2$ matrices 
\be S = \bmp \alpha & \beta \\ \gamma & \delta \emp 
\hskip 10mm \alpha, \beta,\gamma,\delta \in \fd{Z}_{\bar{d}}\ee
such that $\det(G) = 1$ mod $\bar{d}$. If $\beta$ has a multiplicative inverse $\beta^{-1}$ in \Zdbar, we can associate $S$ with a unitary $U_S$ defined explicitly by
\be \label{def-US}
U_S = \frac{e^{i\phi}}{\sqrt{d}}\sum_{x,y=0}^{d-1} \tau^{\beta^{-1}(\alpha y^2 - 2xy + \delta x^2)} \ket{x}\bra{y}, \ee
where $e^{i\phi}$ is an arbitrary phase. If $\beta$ does not admit an inverse, we can always decompose $S$ into a product of two symplectic matrices \cite{Appleby2005}
\be
S = S_1 S_2 =  \bmp \alpha_1 & \beta_1 \\ \gamma_1 & \delta_1 \emp 
 \bmp \alpha_2 & \beta_2 \\ \gamma_2 & \delta_2 \emp 
\ee
such that $\beta_1$ and $\beta_2$ have inverses, and then define $U_S = U_{S_1} U_{S_2}$. Such unitaries $U_S$ with arbitrary overall phases are called symplectic unitaries. They are particularly constructed to satisfy
\be
U_S D_{\bu} U_S^{\dagger} = D_{S\bu}
\ee
and
\be
U_S U_{S'} \doteq U_{SS'}
\ee
for any $S,S' \in$ \SLbar and $u \in$ \Zdbar, where ``$\doteq$'' means equal up to an overall phase. Clifford unitaries (modulo overall phases) are then products of symplectic unitaries and displacement operators $U_S D_{\bu}$.

Every known WH covariant SIC fiducial vector is invariant (ignoring a global phase) under an order 3 Clifford unitary, and conversely, every canonical order 3 Clifford unitary (corresponding to a symplectic matrix of trace -1) has a SIC fiducial as one of its eigenvectors in all dimensions where an exhaustive search has been done \cite{Scott2010}. There is a particular choice for an order 3 Clifford unitary that can be conveniently written in the same form in all dimensions, which we call the Zauner unitary.

\begin{definition} The Zauner unitary is defined to be the symplectic unitary \gls{UZ} corresponding to the Zauner symplectic matrix \gls{Zauner}
\be \label{def-Zauner}
\mcZ \equiv \bmp 0 & -1 \\ 1 & -1 \emp.\ee
\end{definition}
One can easily verify that $\mcZ^3=\eye$ so that $\mcZ$ is indeed of order 3.\\
\begin{conjecture}[Zauner-Appleby \cite{Zauner1999,Appleby2005}]\label{con-ZaunerAppleby}
In every $d$-dimensional Hilbert space, there exists a Weyl-Heisenberg SIC fiducial which is an eigenvector in the largest eigen subspace of the Zauner unitary $U_\mcZ$.
\end{conjecture}

We want to note that in addition to being an extra symmetry for SIC-POVMs on top of the WH covariance, the Zauner symmetry also plays a pivotal role in the study of linear dependencies in Weyl-Heisenberg orbits in \cref{sec:linde}.

\subsection{Analogues to orthonormal bases}\label{sub:orthogonal}

While the discussion is still on SIC-POVMs, there is one nice property of them that we would like to introduce, namely they are as close as possible to being an orthonormal basis on the cone of non-negative operators \cite{Appleby2014S}.

The set operators acting on $d$-dimensional vectors in $\fd{C}^d$ can be considered as a $d^2$-dimensional Hilbert space with the Hilbert Schmidt inner product given by
\be \label{eq-inner-SIC}
\langle A,B \rangle = \Tr(A^\dagger B).
\ee
Let $\{B_i\}_{i=1}^{d^2}$ be an orthogonal basis for this space of operators. One might wonder if it is possible to put some restrictions on $B_i$. For example, can they all be Hermitian? Or unitary? The answers for both are yes. The Hermitian operators themselves form a $d^2$-dimensional real vector space with the same inner product as defined in \cref{eq-inner-SIC}, so one can have a set of $d^2$ Hermitian operators $H_i$ that forms an orthogonal basis for the space of Hermitian operators. One can show that this basis spans the whole space of operators by noticing that any operators (not necessarily Hermitian) can be written as
\be A = H_+ -iH_-,\ee
where 
\be H_+ = (A + A^\dagger)/2, \hskip 15mm H_- = (iA-iA^\dagger)/2 \ee
 are clearly Hermitian. As for an orthogonal unitary basis, one example is the set of Weyl-Heisenberg displacements operators $\{D_\bu: \bu \in \fd{Z}_d^2\}$ defined in \cref{def-Du}. These displacement operators are orthogonal to each other because for any $\bu \ne \bv$ 
\be \langle D_\bu, D_\bv \rangle = \Tr( D_\bu^\dagger D_\bv) =  \tau^{-\Omega(\bu,\bv)}\Tr(D_{\bv-\bu}) = 0 .\ee

However, imposing positive semi-definiteness on an orthogonal basis for the space of operators is impossible, as we will show. Let $\{A_i\}$ be a set of $d^2$ positive semi-definite operators and assume that $A_i$ are normalized, meaning that $\Tr(A_i^2) = 1$. We would like to quantify the extent to which this set is orthogonal. A natural class of ``orthogonality measures'' is defined by
\be \label{def-Kt}
K_t \equiv \sum_{i \ne j} \abs{\langle A_i, A_j \rangle}^t = \sum_{i\ne j} \left(\Tr (A_i A_j)\right)^t
\ee
for any real number $t \ge 1$. This sum consists of $d^4-d^2$ terms and it vanishes if and only if $A_i$ are orthogonal to each other. However, as we will see from the following theorem, this can never happen, since $K_t$ is bounded below by a positive number.

\begin{theorem}\label{thm-Kt}
Let $\{A_i\}_{i=1}^{d^2}$ be a set of $d^2$ normalized positive semi-definite operators on a Hilbert space of dimension $d$, and let $K_t$ be defined as in \cref{def-Kt}, then $K_t$ is lower bounded by
\be K_t \ge \frac{d^2(d-1)}{(d+1)^{t-1}} \ee
When $t=1$, the bound is saturated if and only if $A_i$ are rank-1 projectors and $\sum A_i = d\eye$. When $t>1$, the bound is saturated if and only if $\{A_i/d\}$ forms a SIC-POVM.
\end{theorem}
\begin{proof}
We will first prove the inequality for the $t=1$ case, by making use of a version of the Cauchy-Schwarz inequality (also known as Bouniakowsky inequality \cite{Bouniakowsky1859}):
\be \label{ineq-CS}
\left(\sum_{i=1}^N x_i^2\right)\left(\sum_{i=1}^N y_i^2\right) \ge \left(\sum_{i=1}^N x_i y_i \right)^2 \ee
for any $2N$ real numbers $x_i$ and $y_i$. Particularly, setting $y_i = 1$ leads to
\be \label{ineq-CS2}
\left(\sum_{i=1}^N x_i^2\right) \ge \frac{1}{N}\left(\sum_{i=1}^N x_i \right)^2, \ee
with equality if and only if $x_1 = x_2 =...=x_N$. 

Since $\Tr(A_i^2) = 1$ by the normalization assumption, the (real and positive) eigenvalues of $A_i$ are no larger than 1, and therefore
\be \Tr(A_i) \ge \Tr(A_i^2) = 1, \ee
with equality if and only if exactly one eigenvalue of $A_i$ is 1 and the rest are 0, meaning that $A_i$ are rank-1 projectors.
Let $G$ be a positive semi-definite operator defined by
\be
G = \sum_{i=1}^{d^2} A_i
\ee
It follows that $\Tr(G) \ge d^2$. Applying the inequality \cref{ineq-CS2} to the eigenvalues of $G$ we get
\be \Tr(G^2) \ge \frac{1}{d}(\Tr G)^2 \ge d^3,
\ee
which implies
\be \label{ineq-K1}
K_1 \ge d^3-d^2.
\ee
Equality is obtained if and only if $A_i$ are all rank-1 projectors, and $G=d\eye$.\\

For $t>1$, if we define a function $f(x) = x^t$, then this is a strictly convex function. We can rewrite $K_t$ as
\be
K_t = \sum_{i \ne j} f\left(\Tr(A_i A_j )\right).
\ee
Applying a particular instance of Jensen inequality \cite{Jensen1906}, namely
\be
\sum_{i=1}^N f(x_i) \ge N f\left(\frac{\sum x_i}{N}\right)
\ee
for any convex function $f(x)$ and any $x_i$ in the domain, with equality (in the case of strict convexity) if and only if $x_i$ are all constant, we obtain
\be \begin{split}
K_t &\ge (d^4-d^2) f\left(\frac{\sum_{i \ne j} \Tr(A_i A_j)}{d^4-d^2}  \right) \\ 
& \ge (d^4-d^2) f\left(\frac{d^3-d^2}{d^4-d^2}  \right) \\
& = \frac{d^2(d-1)}{(d+1)^{t-1}}.
\end{split}\ee
For this bound to be saturated, equality must take place in both the Jensen inequality and \cref{ineq-K1}. This means the $d^2$ operators $A_i$ must be all rank-1 projectors satisfying $\sum_i A_i = d\eye$ and $\Tr(A_i A_j)$ are constant for all $i\ne j$. This is precisely the definition of a SIC-POVM.
\end{proof}

\begin{remark}
What we have shown is that on the cone of non-negative operators, there does not exist an orthonormal basis. Furthermore, using $K_t$ as a natural class of ``orthogonality measures'', we have shown that SIC-POVMs stand out as the ``most orthogonal'' bases on this cone. We want to note that one of our orthogonality measures, namely $K_2$, is closely related to the frame potential 
\be \Phi = \sum_{i,j} \abs{\braket{\psi_i}{\psi_j}}^4 \ee
introduced by Renes \etal \cite{Renes2004} via the simple relation
\be K_2 = \Phi -d^2. \ee
In \cite{Renes2004}, the minimization of the frame potential was used to aid the numerical search for SICs, and the bound was proved in the context of frames and spherical $t$-designs. Here, our proof of the bound relies only on a few well-known elementary inequalities.
\end{remark}

\section{Linear dependencies in Weyl-Heisenberg orbits}\label{sec:linde}
The study of linear dependencies in WH orbits \cite{Dang2013} stems from an observation that among 9 vectors in any known 3-dimensional SIC set, one can find some sets of 3 vectors that are linearly dependent. This led to our investigation in higher dimensions, where the question we asked was: among $d^2$ SIC vectors in dimension $d$, could one find a set of $d$ of them that are linearly dependent? Going through the numerical SICs provided in \cite{Scott2010}, we found a striking pattern: it seems as though whenever $d$ is divisible by 3, the answer is yes. 

Since these are Weyl-Heisenberg covariant SICs, the SIC vectors can be expressed as $D_\bp \qs$, where $D_\bp$ are the WH displacement operators indexed by $\bp = (p_1,p_2) \in \fd{Z}_d^2$, and $\qs$ is a SIC fiducial vector. If $\{\bp_i\}_{i=1}^d$ is a set of $d$ indices $\bp_i$ such that the $d$ vectors $D_{\bp_i} \qs$ are linearly dependent, we call it a ``good'' p-set. It follows that for any good p-set, there exists a set of coefficients $\lambda_i$, which are not simultaneously zero, such that
\be \sum_i \lambda_i D_{\bp_i} \qs = 0 \hskip 10mm \text{or} \hskip 10mm L\qs = 0,\ee
where $L$ is defined to be $L = \sum_i \lambda_i D_{\bp_i}$. In dimensions $d$ that are divisible by 3, not only have we found many good p-sets, but we have also noticed numerically that in many cases their corresponding $L$ matrices are of rank $d-1$. This means the null space of $L$ is 1-dimensional, i.e. the matrix equation $L x = 0$ has a unique solution (up to a phase), which is the SIC fiducial $\qs$. If we know what the $L$ matrices are, we could simply solve this matrix equation to obtain the SIC fiducial!

Finding the $L$ matrices requires us to identify good p-sets $\{\bp_i\}$ as well as the coefficients $\lambda_i$. As it will be shown later in this section, we succeeded in the first task. However, finding $\lambda_i$ is non-trivial, despite a 3-fold symmetry of theirs that we observe. In fact, it turns out that this approach to the SIC problem cannot work, because the linear dependence property is not unique to SIC fiducials, but is generic to a class of eigenvectors of certain Clifford unitaries, one of which is the Zauner unitary $U_\mcZ$ defined in \cref{def-Zauner}.

In this section, we first examine a special SIC set in dimension $d=3$, which has a connection to elliptic curves via Hesse configuration \cite{Hughston2007}. We then provide an analytical proof for linear dependencies in the WH orbits of the eigenvectors of the Zauner unitary. We give a detailed report on our numerical study, in which the number of observed linear dependencies is often higher than what can be accounted for from the analytical prediction. And lastly, we show a robust construction of ``small SICs'' in dimension $d=2$ and 3 that resulted from this study of linear dependencies.\\

\begin{remark}
Before going into the details of this study, we want to note that the opposite problem, i.e. to find fiducial vectors whose WH orbits contain no linear dependencies, is useful for classical signal processing and it has been solved \cite{Pfander2013, Lawrence2005, Malikiosis2015}.
\end{remark}

\subsection{Dimension $d=3$ and Hesse configuration}\label{sub:d3}
In dimension $d=3$, there is a continuous family of SICs that can be parameterized by a single parameter \cite{Appleby2005}. All other known SICs in dimension 3 have been shown to be unitarily equivalent to this family \cite{Zhu2010}. Explicitly, the 9 vectors of a SIC in this family can be written in the following form (ignoring normalization factors):
\be \label{def-SIC3}
\begin{split}
\bmp 0 \\ 1 \\ -e^{i\theta} \emp,
\bmp 0 \\ 1 \\ -e^{i\theta}\eta \emp,
\bmp 0 \\ 1 \\ -e^{i\theta}\eta^2 \emp,\\
\bmp -e^{i\theta} \\0 \\1 \emp,
\bmp -e^{i\theta}\eta \\0 \\1 \emp,
\bmp -e^{i\theta}\eta^2 \\0 \\1 \emp,\\
\bmp 1\\ -e^{i\theta} \\0 \emp,
\bmp 1\\ -e^{i\theta}\eta \\0 \emp,
\bmp 1\\ -e^{i\theta}\eta^2 \\0 \emp,
\end{split}
\ee 
where $\eta = e^{2\pi i/3}$ is a third root of unity (we reserve $\omega$ for the $d$-th root of unity in general) and the parameter $\theta$ is in the interval $[0,\pi/3]$. One observes that the 3 vectors on each line in \cref{def-SIC3} span a 2-dimensional subspace. Hence, any such SIC contains 3 sets of 3 linearly dependent vectors. However, for certain values of $\theta$, there are additional linear dependencies. One can find these values by putting any 3 vectors from 3 different lines in \cref{def-SIC3} together as a $3\times 3$ matrix and set the determinant of this matrix to zero to obtain
\be
e^{3i\theta} = \eta^k \hskip 15mm k=0,1,2.
\ee
Given the range of $\theta$ in consideration $[0,\pi/6]$, there are two choices $\theta = 0$ or $2\pi/9$, giving rise to two special SICs that both contain 12 sets of 3 linearly dependent vectors. The SIC corresponding to $\theta=0$ is ``extra special'' because its fiducial vector
\be
\qsp = \bmp 0 \\1 \\ -1\emp 
\ee
is an eigenvector of symplectic unitaries $U_S$ for all $S \in \SL(2,\Fd)$. One can see this by noticing that the density operator can be written as
\be
\ketbra{\phi} = \eye - U_P,
\ee
where the parity operator $U_P$ is the symplectic unitary corresponding to an SL element
\be
P = \bmp -1 & 0 \\ 0 & -1 \emp,
\ee
which in turn is the only element (besides the identity element) that commutes with all other elements in SL. If we label the 9 SIC vectors in \cref{def-SIC3} by $\ket{00},\ket{01},...,\ket{22}$, then the linearly dependent relations for the $\theta=0$ SIC are as follows:
\be \label{eq-linde1-SIC3}\begin{split}
\ket{00} + \ket{10} + \ket{20} &= 0\\
\ket{01} + \ket{11} + \ket{21} &= 0\\
\ket{02} + \ket{12} + \ket{22} &= 0
\end{split} \ee
\be \label{eq-linde2-SIC3}\begin{split}
\ket{00} + \eta\ket{01} + \eta^2\ket{02} &= 0\\
\ket{10} + \eta\ket{11} + \eta^2\ket{12} &= 0\\
\ket{20} + \eta\ket{21} + \eta^2\ket{22} &= 0
\end{split} \ee
\be \label{eq-linde3-SIC3}\begin{split}
\phantom{\eta}\ket{00} + \phantom{\eta}\ket{11} + \eta\ket{22} &= 0\\
\phantom{\eta}\ket{01} + \eta\ket{12} + \phantom{\eta}\ket{20} &= 0\\
\eta\ket{02} + \phantom{\eta}\ket{10} + \phantom{\eta}\ket{21} &= 0
\end{split} \ee
\be \label{eq-linde4-SIC3}\begin{split}
\phantom{\eta^2}\ket{00} + \phantom{\eta^2}\ket{12} + \eta^2\ket{21} &= 0\\
\eta^2\ket{01} + \phantom{\eta^2}\ket{10} + \phantom{\eta^2}\ket{22} &= 0\\
\phantom{\eta^2}\ket{02} + \eta^2\ket{11} + \phantom{\eta^2}\ket{20} &= 0.
\end{split} \ee

We note that in \cref{eq-linde1-SIC3,eq-linde2-SIC3,eq-linde2-SIC3,eq-linde3-SIC3,eq-linde4-SIC3}, each linearly dependent relation involves 3 SIC vectors, and each SIC vector appears in 4 relations. If we represent the SIC vectors by 9 points, and draw a ``line'' connecting 3 points if their SIC vectors are linearly dependent, then we obtain a set of 9 points and 12 lines, each line containing 3 points, and each point is contained in 4 lines, as illustrated in \cref{fig-Hesse}. It was pointed out by Lane Hughston \cite{Hughston2007} that this was precisely the Hesse configuration \cite{Hesse1844}, often denoted by the configuration $(9_4,12_3)$ in the language of configurations in geometry.
\begin{figure}[h]
\centering
\includegraphics[scale = 0.3]{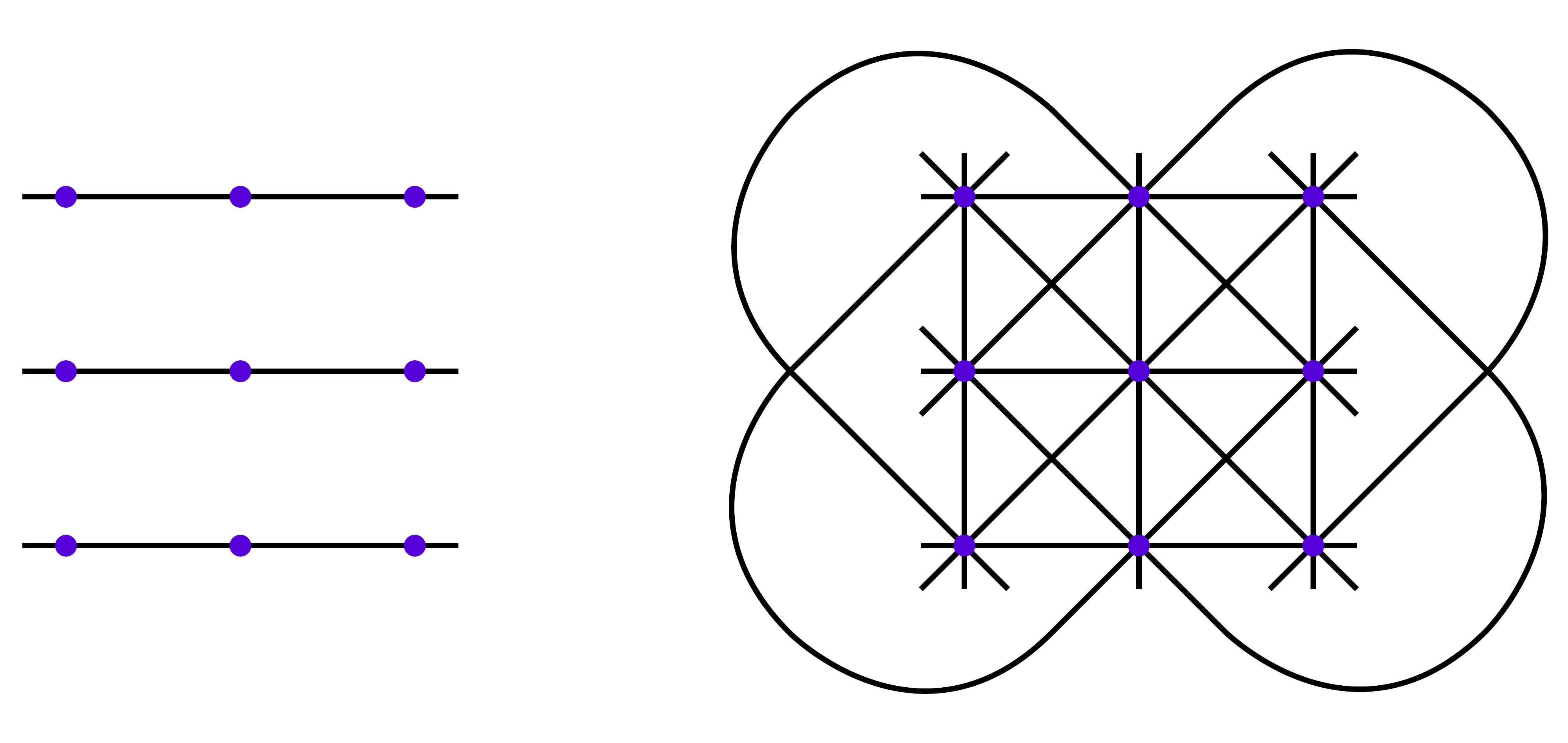}
\vskip 2mm
\parbox{12cm}{\caption[Linear dependency structures of SICs in $d=3$]{The linear dependency structures of a generic 3-dimensional SIC (on the left) and of the special SICs with $\theta = 0$ or $2\pi/9$ (on the right). Each point represents a SIC vector, and points are connected by a line if the corresponding SIC vectors are linearly dependent. The structure to the right, with 9 points and 12 lines, is known as the Hesse configuration.}
\label{fig-Hesse}}
\end{figure}

The Hesse configuration is not realizable in the Euclidean plane. However, it can be realized in the complex projective plane, which was discovered by Hesse in the study of elliptic curves. Particularly, let us consider the following family of cubic curves described by the polynomial equations
\be \label{eq-cubic}
P(x_1,x_2,x_3) = x_1^3 + x_2^3 + x_3^3 + \lambda x_1 x_2 x_3 = 0 \hskip 15mm \lambda \in \fd{C}.
\ee
The $3\times 3$ Hessian matrix $H$ consists of second derivatives of $P$, with its entries given by
\be H_{ij} = \frac{\partial^2 P}{\partial x_i \partial x_j}.
\ee
Points on the curve where the determinant of the Hessian matrix vanishes are called inflection points. This determinant is also a cubic polynomial in the family, and by B{\'e}zout's theorem \cite{Bezout1779}, these two cubic curves intersect at 9 points. These 9 inflection points are the same for all values of $\lambda$, and they coincide exactly with the 9 vectors of the $\theta = 0$ SIC. Furthermore, there are 4 special (singular) curves in the family corresponding to the cases when $\lambda = \infty$ and $\lambda^3 = -27$:
\be \begin{split}
P_0&: x_1 x_2 x_3 = 0\\
P_1&: (x_1 + x_2 + x_3)(x_1 + \eta x_2 + \eta^2 x_3)(x_1 + \eta^2 x_2 + \eta x_3) = 0\\
P_2&: (x_1 + \eta^2 x_2 + \eta^2 x_3)(x_1 + \eta x_2 + x_3)(x_1 + x_2 + \eta x_3) = 0\\
P_3&: (x_1 + \eta x_2 + \eta x_3)(x_1 + \eta^2 x_2 + x_3)(x_1 + x_2 + \eta^2 x_3) = 0\\
\end{split}\ee
One can see that each of these curves $P_i$ degenerates into 3 projective lines, giving a total of 12 lines. Each of the lines passes through 3 inflection points, and each inflection point belongs to 4 lines. This results in the Hesse configuration. More details on its connection to SICs can be found in \cite{Bengtsson2010}.

\begin{remark}
For each of the 12 linear dependencies in the $\theta=0$ SIC, the corresponding set of 3 linearly dependent vectors spans a 2-dimensional plane in the Hilbert space, whose normal vector is unique up to a scalar multiplication. These 12 normal vectors form a complete set of MUBs in dimension $d=3$ and this observation has been used for a Kochen-Specker inequality \cite{Bengtsson2012}.
\end{remark}

\subsection{Linear dependencies from Zauner eigenvectors}\label{sub:analytical}

As we mentioned earlier, the linear dependence property is not a unique feature of SICs. We will now show that linear dependencies can arise in all dimensions in Weyl-Heisenberg orbits of vectors that lie in certain eigenspaces of the Zauner unitary $U_\mcZ$. Known SIC fiducials just happen to be among those vectors.

Recall that by the definition in \cref{def-Zauner}, the Zauner symplectic matrix $\mcZ$ is of order 3. Hence, its symplectic unitary $U_\mcZ$ is also of order 3, for a suitable choice of the phase $e^{i\phi}$ in \cref{def-US}. This means that its eigenvalues must be 1, $\eta$, or $\eta^2$, where $\eta=e^{2\pi i/3}$ is a third root of unity, and $U_\mcZ$ has 3 eigenspaces corresponding to these 3 eigenvalues. If the phase in \cref{def-US} is particularly chosen to be
\be e^{i\phi} = e^{i\pi(d-1)/12}, \ee
then one finds that \cite{Zauner1999} the eigenspaces $\mcH_1$, $\mcH_{\eta}$, and $\mcH_{\eta^2}$ corresponding to the eigenvalues $1$, $\eta$, and $\eta^2$ have dimensions as shown in \cref{tab-UZeigenspaces}. In all dimensions $d$, SIC fiducials are found in $\mcH_1$, which will be specifically referred to as the Zauner subspace. We want to note that when $d$ is equal to 8 mod 9, additional SIC fiducials are found in the other two eigenspaces as well \cite{Scott2010}.
\begin{table}[h]
\centering
\begin{tabular}{M{1.85cm} M{2.5cm} M{2.5cm} M{2.5cm}}
 & \cellcolor{gray!25}$d=3k$ & \cellcolor{gray!25}$d=3k+1$ & \cellcolor{gray!25}$d=3k+2$\\ 
\cellcolor{gray!25}dim($\mcH_1$) & $k+1$ & $k+1$ & $k+1$ \\ 
\cellcolor{gray!25}dim($\mcH_\eta$) & $k$ & $k$ & $k+1$  \\ 
\cellcolor{gray!25}dim($\mcH_{\eta^2}$) & $k-1$ & $k$ & $k$  \\ 
\end{tabular}
\vskip 4mm
\parbox{12cm}{\caption[The dimensions of the eigenspaces of the Zauner unitary $U_\mcZ$]{Dimensions of the three eigenspaces $\mcH_1$, $\mcH_\eta$, and $\mcH_{\eta^2}$ of the Zauner unitary $U_\mcZ$ for different dimensions $d$.}
\label{tab-UZeigenspaces}}
\end{table}

Since Zauner symplectic $\mcZ$ is of order 3, it generally moves points $\bp = (p_1,p_2)$ on the discrete phase space $\Fd^2$ in orbits $\{\bp, \mcZ \bp, \mcZ^2 \bp\}$ of size 3, which will be referred to as triplets. The exception is when $\bp$ is a fixed point of $\mcZ$, i.e. $\mcZ \bp = \bp$, in which case we will call it a singlet. One can easily solve for the fixed points of $\mcZ$ and find that there are 3 singlets when $d$ is divisible by 3, and only 1 trivial singlet (the zero vector) otherwise. The singlets are given in \cref{tab-singlets}. In the Hilbert space, we will use the same terminology to call $D_\bp \qs$ a singlet if $\bp$ is a singlet, and to call $\{D_\bp\qs, D_{\mcZ\bp}\qs, D_{\mcZ^2 \bp}\qs\}$ a triplet otherwise.
\begin{table}[h]
\centering
\begin{tabular}{M{4.5cm} M{2.5cm}}
\cellcolor{gray!25}$d=3k$ & \cellcolor{gray!25}$d\ne3k$ \\ 
$\bmp 0\\0 \emp, \bmp k\\2k \emp, \bmp 2k\\k \emp$ & $\bmp 0\\0 \emp$ \vphantom{$\bmp 0\\0\\0 \emp$}
\end{tabular}
\vskip 2mm
\parbox{12cm}{\caption[Fixed points of the Zauner symplectic $\mcZ$]{Fixed points of the Zauner symplectic $\mcZ$.}
\label{tab-singlets}}
\end{table}
\begin{theorem}\label{thm-linde}
Let $\qsp$ be an eigenvector of the Zauner unitary $U_\mcZ$ with eigenvalue $\lambda$, and let $V$ be a set of $d$ vectors in the Weyl-Heisenberg orbit of $\qsp$, i.e. $\{D_\bp \qsp:\bp \in \Fd^2\}$. Then the vectors in $V$ are linearly dependent if:
\begin{enumerate}
\item $V$ contains $k$ triplets, or $k-1$ triplets and 3 singlets for the case $d=3k$,
\item $V$ contains $k$ triplets and 1 singlet and $\qsp \in \mcH_\eta \cup \mcH_{\eta^2}$ for the case $d=3k+1$,
\item $V$ contains $k$ triplets and 1 singlet and $\qsp \in \mcH_{\eta^2}$ for the case $d=3k+2$.
\end{enumerate}
\end{theorem}
\begin{proof}
We start with the case $d=3k$. If $\bp$ is a singlet, then $\qsp$ and $D_\bp \qsp$ lie in the same eigenspace of $U_\mcZ$ because
\be
U_{\mcZ} D_\bp \qsp = U_{\mcZ} D_\bp U_{\mcZ}^\dagger U_\mcZ \qsp = \lambda D_{\mcZ \bp} \qsp = \lambda D_\bp \qsp.
\ee
If $\bp$ is in a triplet, we construct 3 new vectors from linear combinations of vectors in the triplet  $\{D_\bp\qsp, D_{\mcZ\bp}\qsp, D_{\mcZ^2 \bp}\qsp\}$ as follows:
\be \label{def-rst}\begin{split}
\ket{r} & = D_\bp \qsp + \phantom{\eta^2}U_\mcZ D_\bp \qsp + \phantom{\eta^2}U_\mcZ^2 D_p \qsp \\
\ket{s} & = D_\bp \qsp + \eta^2 U_\mcZ D_\bp \qsp + \eta^{\phantom{2}} U_\mcZ^2 D_p \qsp \\
\ket{t} & = D_\bp \qsp + \eta^{\phantom{2}} U_\mcZ D_\bp \qsp + \eta^2 U_\mcZ^2 D_p \qsp.
\end{split}\ee
Given a choice of $\bp$, we will refer to vectors constructed in way as $r$-type, $s$-type, and $t$-type respectively. One can straightforwardly verify that 
\be U_\mcZ \ket{r} = \ket{r} \ee
so $r$-type vectors belong to the eigenspace $\mcH_1$ of $U_\mcZ$. Similarly, $s$-type and $t$-type vectors belong to the other two eigenspace $\mcH_\eta$ and $\mcH_{\eta^2}$, respectively. Moreover, \cref{def-rst} can be inverted so that any vector in the triplet $\{D_\bp\qsp, D_{\mcZ\bp}\qsp, D_{\mcZ^2 \bp}\qsp\}$ can be written as a linear combination of $\ket{r},\ket{s}$ and $\ket{t}$:
\be \begin{split}
D_\bp \qsp  &= \left(\ket{r} + \phantom{\eta^2}\ket{s} + \phantom{\eta^2}\ket{t}\right)/3 \\
U_\mcZ D_\bp \qsp  &= \left(\ket{r} + \eta^{\phantom{2}} \ket{s} + \eta^2 \ket{t}\right)/3 \\
U_\mcZ^2 D_\bp \qsp  &= \left(\ket{r} + \eta^2 \ket{s} + \eta^{\phantom{2}} \ket{t}\right)/3.
\end{split}\ee
Therefore the two sets $\{D_\bp\qsp, D_{\mcZ\bp}\qsp, D_{\mcZ^2 \bp}\qsp\}$ and $\{\ket{r},\ket{s},\ket{t}\}$ have exactly the same linear span.

If $V$ contains $k$ triplets, this gives $k$-many of each $r$-,$s$-, and $t$-type vector. From \cref{tab-UZeigenspaces} we know that the $r$-type vectors belong to an eigenspace of dimension $k+1$, the $s$-type vectors belong to an eigenspace of dimension $k$, and the $t$-type vectors belong to an eigenspace of dimension $k-1$. It clearly follows that the $k$ $r$-type vectors cannot fully span their eigenspace, while the $k$ $t$-type vectors are overcomplete and therefore linearly dependent.

If $V$ contains $k-1$ triplets and 3 singlets, this gives $(k-1)$-many of each $r$-,$s$-, and $t$-type vector, plus the 3 singlets that belong to the same eigenspace as  $\qsp$. This means there will be $k+2$ vectors among these that belong to the same eigenspace. Since the largest eigenspace of $U_\mcZ$ has dimensionality $k+1$, we obtain linear dependency.

Still sticking to the case $d=3k$, we want to note that if we further assume $\qsp \in \mcH_\eta$, we also obtain linear dependency when $V$ contains $k-1$ triplets and 2 singlets, or $k-2$ triplets and 3 singlets, using the same argument. Assuming $\qsp \in \mcH_{\eta^2}$ allows us to extend this linear dependency condition even further, to include cases when $V$ contains $k-3$ triplets and 3 singlets.

In the case $d=3k+1$, the 3 eigenspaces $\mcH_1,\mcH_\eta$ and $\mcH_{\eta^2}$ have dimensionality $k+1,k$, and $k$ respectively. If $\qsp \in \mcH_\eta$, then the singlet together with $k$ $s$-type vectors give $k+1$ vectors in $\mcH_\eta$, resulting in linear dependency. If $\qsp \in \mcH_{\eta^2}$, the singlet together with $k$ $t$-type vectors give $k+1$ vectors in $\mcH_{\eta^2}$, also resulting in linear dependency.

In the case $d=3k+2$, the eigenspace $\mcH_{\eta^2}$ has dimensionality $k$. The $k$ $t$-type vectors, together with the singlet, form a set of $k+1$ vectors in $\mcH_{\eta^2}$. Therefore they are linearly dependent.\\
\end{proof}

\subsection{Numerical linear dependencies}\label{sub:numerical}

Although the results in the previous section provide a significant understanding of how linear dependencies can arise in a WH orbit of an initial vector that is an eigenvector of $U_\mcZ$, they do not account for all the linear dependencies that we observe numerically. In this section, we provide a numerical analysis for linear dependencies in dimensions $d=4$ to 8, with partial results in $d=9$. We pay more attention and provide more details for the cases $d=6$ and 9, as we are interested in dimensions that are divisible by 3 (we know from \cref{thm-linde} that in dimensions $d=3k$ one can obtain linear dependencies from an initial vector in the Zauner subspace $\mcH_1$, where SIC fiducials are expected to be). Although the additional (cannot be accounted for by \cref{thm-linde}) linear dependencies in $d=6$ and 9 do not depend on whether the initial vector is a SIC fiducial or not, there are some interesting orthogonality relationships in the dependency structure that seem to be unique to SICs. Our investigation in $d=6$ and 9 also led to ``small SICs'', which will be the focus of the next section. In dimension $d=8$, there are SIC fiducials in $\mcH_{\eta^2}$, and we observe that these SIC fiducials yield more linearly dependencies than an arbitrary initial vector in $\mcH_{\eta^2}$. This indicates some similarity to the situation in $d=3$ that led to the Hesse configuration that we discussed in \cref{sub:d3}.

In each dimension $d$ that we analyzed, our computer program started with an arbitrary initial vector of each of the eigenspaces of the Zauner unitary $U_\mcZ$, generated the full orbit under the action of the WH group, and then performed an exhaustive search for all subsets of $d$ vectors that are linearly dependent by calculating the determinants of the $d \times d$ matrices formed by these $d$ vectors to a numerical precision of $10^{-15}$. For each eigenspace, the procedure was repeated for a small number of arbitrarily chosen initial vectors, to make sure that the results are the same. As for SIC fiducials, we used those given in \cite{Scott2010}, and when there are more than one Clifford orbit we repeat the calculation with fiducials from each orbit.

We found no distinction between choices of the initial vector, except in dimension $d=8$ where the SIC fiducial gives rise to 24,935,160 sets of linearly dependent vectors, slightly higher than the generic result shown in \cref{tab-lindenumerical}. In dimension $d=7$, the numerical results match our prediction from \cref{thm-linde}. But in all other cases, there are more numerically observed linear dependencies than what \cref{thm-linde} can account for. \cref{tab-lindenumerical} shows the number of linearly dependent sets found in WH orbits with generic initial vectors in different eigenspaces of $U_Z$ for dimensions $d=4$ to 8. We were not able to perform an exhaustive search for dimension 9 or higher.
\begin{table}[h]
\centering
\begin{tabular}{M{0.6cm} M{1.5cm} M{1.5cm} M{2cm} M{1.5cm} M{2cm}}
& \cellcolor{gray!25}$d=4$ & \cellcolor{gray!25}$d=5$ & \cellcolor{gray!25}$d=6$ & \cellcolor{gray!25}$d=7$
& \cellcolor{gray!25}$d=8$ \\ 
\cellcolor{gray!25}$\mcH_1$ & 0 & 0 & 984 & 0 & 0 \\
\cellcolor{gray!25}& 0 & 0 & (768) & 0 & 0 \\
\cellcolor{gray!25}$\mcH_\eta$ & 116 & 0 & 635,052 & 5,796 & 0 \\
\cellcolor{gray!25} & (68) & 0 & (75,342) & (5,796) & 0 \\
\cellcolor{gray!25}$\mcH_{\eta^2}$ & 116 & 6,600 & 17,903,28 & 5796 & 24,756,984 \\
\cellcolor{gray!25} & (68) & (4,200) & - & (5,796) & ($\le$766,080) \\
\end{tabular}
\vskip 4mm
\parbox{14cm}{\caption[Number of linear dependencies in WH orbits of eigenvectors of $U_Z$]{Number of linear dependencies in a WH orbit where the initial vector is arbitrarily taken from each eigenspace of $U_Z$. The numbers in brackets are the number of sets (or an upper bound in one case) predicted from \cref{thm-linde}.}
\label{tab-lindenumerical}}
\end{table}

In dimension $d=6$ we found 984 numerical linearly dependent sets starting from a generic initial vector in the Zauner subspace. Among these, only 768 sets have the property that they are invariant under the Zauner unitary, a condition for linear dependency in \cref{thm-linde}. This leaves 216 sets unaccounted for by our theorem. Interestingly, it is worth noting that these additional 216 sets are instead invariant under an order 6 symplectic unitary $U_{\mathcal{M}}$, where 
\be \mathcal{M} = \bmp 3 & 8 \\ 4 & 11 \emp.\ee
Each of the 36 vectors in the Weyl-Heisenberg orbit lies in 164 different linearly dependent sets, and each of the 984 sets clearly contains 6 vectors. In the language of geometry, we have 36 points and 984 hyperplanes in the complex projective space $\fd{CP}^5$, forming the configuration $(36_{164},984_{6})$.

The 984 linearly dependent sets in dimension $d=6$ themselves (for an initial vector in the Zauner subspace) can be grouped into orbits under the WH group (note that if a set is linearly dependent, then the set obtained by displacing it with the operator $D_{\bp}$ for any $\bp$ is also linearly dependent). We counted 27 orbits of length 36, and 1 orbit of length 12. The reason for the short orbit is because it contains sets that are invariant under the subgroup $\{\eye,D_{24},D_{42}\}$. Among these 28 orbits, 22 contain sets that are invariant under $U_\mcZ$, while the other 6 contain sets that are invariant under $U_{\mathcal{M}}$. The results are summarized in \cref{tab-lindeorbits}, where we have labeled the orbits from 1 to 28, with the first one being the short orbit, and the last 6 (23 to 28) are the ones invariant under $U_{\mathcal{M}}$.

\begin{table}[h]
\centering
\begin{tabular}{M{1.5cm} M{1.5cm} M{3cm} M{3cm} M{2.5cm}}
\cellcolor{gray!25}Orbit \phantom{aaaa} & \cellcolor{gray!25}Length \phantom{aaaa} & \cellcolor{gray!25} No. orbits under $U_\mcZ$ & \cellcolor{gray!25} No. orbits under $\{\eye,D_{24},D_{42}\}$ & \cellcolor{gray!25} No. ON quadruples \\ 
1 & 12 & 2 & 2 & 0 \\
2-10 & 36 & 2 & 1 & 0 \\
11 & 36 & 2 & 0 & 9 \\
12-13 & 36 & 2 & 0 & 0 \\
14-22 & 36 & 2 & 0 & 0 \\
23-28 & 36 & 1 & 1 & 0 \\
\end{tabular}
\vskip 4mm
\parbox{14cm}{\caption[Properties of WH orbits of linearly dependent sets in $d=6$]{Properties of WH orbits of linearly dependent sets with an initial vector in the Zauner subspace $\mcH_1$ in dimension $d=6$.}
\label{tab-lindeorbits}}
\end{table}

In hope of finding nice structures as in dimension $d=3$, we also studied the relationship among normal vectors of the 984 5-dimensional hyperplanes corresponding to the linearly dependent sets from Zauner subspace in $d=6$. We performed an exhaustive search for orthogonalities between these vectors. No basis was found, nor was a MUB. However we did find over 20,000 orthogonal triples of normal vectors, i.e. sets of 3 normal vectors that are orthogonal to each other. If we start from a SIC fiducial instead of an arbitrary vector in the Zauner subspace, the linear dependency remains identical. However, in this case we found 216 additional orthogonal triples, and we also found 9 orthogonal quadruples. They all belong to the same orbit (of length 36) under the WH group.

In dimension $d=9$, starting from an arbitrary vector $\qsp$ in the Zauner subspace, we found 79,767 sets of 9 linearly dependent vectors in the WH orbit of $\qsp$, 78,795 of which can be accounted for by \cref{thm-linde}. This number is too large for us to perform an exhaustive calculation of the scalar products between all pairs of normal vectors as in $d=6$, but we did find interesting relations among some normal vectors, which will be discussed in the next section. The 79,767 sets can be grouped into orbits under the WH group. We found a total of 987 orbits: 984 of length 81, 2 of length 27, and 1 of length 9. Like in $d=6$, they can be split into 2 groups: one group of 975 orbits that are exclusively invariant under $U_\mcZ$, and the other group of 12 orbits that are exclusively invariant under $U_\mcM$ (there are 186 orbits that are invariant under both), where $\mcM$ in odd dimensions $d=3k$ takes the form
\be
\mcM = \bmp k+1 & k \\ 2k & 2k+1 \emp. \ee

If we start from a SIC fiducial in $d=9$, we obtain an identical linear dependency structure. This suggests a distinction between SICs in dimension $d=3$ and SICs in higher dimensions divisible by 3. The ``special'' SICs in $d=3$ give rise to 12 linearly dependent sets, while others produce only 3. In this sense, no SICs are ``special''  in dimension $d=6$ and 9. However, we did find one other instance of a SIC fiducial vector giving more linear dependencies than other arbitrary vectors in the same eigenspace. This is the SIC fiducial in $\mcH_{\eta^2}$ in dimension $d=8$ (this additional SIC fiducial seems to only exist in dimensions that are equal to 8 mod 9 \cite{Scott2010}). In $d=8$, this particular SIC fiducial exhibit 24,935,160 linearly dependent sets, while a generic vector in $\mcH_{\eta^2}$ produces 24,756,984 sets. This may be connected to the fact that the SIC has a larger automorphism group than an arbitrary vector in the same eigenspace. Comparing to the case in $d=3$ where the Hesse configuration arises from the ``special'' SICs, one might ask whether there is a similar connection to elliptic curves in this family of SICs in dimensions $d=9k + 8$.

\subsection{Small SICs in dimensions $d=6$ and $9$}\label{sub:smallSICs}

In the numerical investigation presented in the previous section, we intentionally left out some interesting observations in dimensions $d=6$ and 9 for a separate discussion in this section. Among 984 vectors normal to the linear dependent sets generated from an arbitrary vector in the Zauner subspace (not necessarily a SIC fiducial) in $d=6$, we found 30 sets of 4 normal vectors that form 2-dimensional SICs, i.e. within each set, the overlaps between the vectors are $1/\sqrt{3}$ and the vectors lie in a 2-dimensional subspace. This phenomenon also happens in dimension $d=9$, where 3-dimensional SICs are found among 79,767 normal vectors. We refer to SICs of this kind as ``small SICs'', as their dimension is smaller than that of the embedding Hilbert space. Attempts to find small SICs in dimension $d=12$ yielded no positive result so far, but we have not been able to perform an exhaustive search. In this section, we provide an explanation for the small SICs in $d=6$. Small SICs in $d=9$ are not yet fully understood.

We observed that every 2-dimensional SIC set found in $d=6$ can be expressed as an orbit of a vector under the subgroup $\{\eye, D_{03}, D_{30}, D_{33} \}$. In other words, the SICs take the form $\{\qs, D_{03}\qs, D_{30}\qs, D_{33}\qs \}$, where $D_{ij}$ are the displacement operators defined in \cref{def-Du} for dimension $d=6$. In the following theorem we will prove that for such a set to form a 2-dimensional SIC, $\qs$ just needs to be any normalized eigenvector of $\UZ$ that is not in the Zauner subspace. Normal vectors of linearly dependent sets containing 3 singlets in the case the initial vector lies in the Zauner subspace happen to meet this condition, as seen from the proof of \cref{thm-linde}.

\begin{theorem} \label{thm-smallSICs}
In dimension $d=6$, if $\qs$ is an eigenvector of $\UZ$, then the 4 vectors $\qs,D_{03}\qs, D_{30} \qs$ and $D_{33} \qs$ are equiangular. Furthermore, if $\qs \in \mcH_\eta$ or $\qs \in \mcH_{\eta^2}$, those 4 vectors span a 2-dimensional subspace and therefore constitute a SIC.
\end{theorem}

\begin{proof}
Let $\lambda$ be the eigenvalue of $\UZ$ corresponding to $\qs$, so that 
\be \UZ \qs=\lambda \qs \hskip 15mm \ket{\psi}=\lambda \UZ^\dagger\qs. \ee
The first statement of the theorem is true because
\begin{equation}
\bra{\psi}D_{\mathbf{p}}\ket{\psi} = \bra{\psi}\lambda^* U_\mathcal{Z} D_{\mathbf{p}}  U_\mathcal{Z}^\dagger \lambda \ket{\psi}
= \bra{\psi}D_{\mathcal{Z} \mathbf{p}}\ket{\psi} \end{equation}
and because the Zauner symplectic matrix $\mathcal{Z}$ simply permutes the three points 03, 30 and 33 in the discrete phase space according to \cref{tab-pZp}.

\begin{table}[h]
\centering
\begin{tabular}{M{1.2cm} M{1.2cm} M{1.2cm} M{1.2cm}}
\cellcolor{gray!25} $\mathbf{p}$ & (0,3) & (3,0) & (3,3)\\
\cellcolor{gray!25} $\mathcal{Z}\mathbf{p}$ & (3,3) & (0,3) & (3,0)
\end{tabular}
\vskip 4mm
\parbox{14cm}{\caption[Action of $\mcZ$ on points (0,3), (3,0) and (3,3) in $d=6$]{Cyclic action of $\mcZ$ on points (0,3), (3,0) and (3,3).}
\label{tab-pZp}}
\end{table}

In proving the second statement of the theorem, we make use of a square root of the Zauner unitary. Let $\mathcal{W}$ be the symplectic matrix given by 
\begin{equation}
\mcW \equiv \begin{pmatrix}1 & -1\\1 & 0\end{pmatrix} \hspace{15mm}
\mcW^2 = \begin{pmatrix}0 & -1\\1 & -1\end{pmatrix} = \mathcal{Z}
\end{equation}
Let $\UW$ be the symplectic unitary corresponding to $\mcW$, with a phase chosen so that
\be \UW^2 = \UZ.\ee 
The structure of the eigenspaces of $\UZ$ and $\UW$ are described in \cref{tab-eigenUZUW}, where $\omega = e^{2\pi i/6}$ and $\eta = e^{2\pi i/3} =\omega^2$.
\begin{table}[ht!]
\begin{center}
\begin{tabular}{| l | c  | c |  c  | c | c |}
\hline
&\multicolumn{2}{c|}{}&\multicolumn{2}{c|}{}& \\[-3.5mm] 
Eigenspaces of $U_\mathcal{Z}$& \multicolumn{2}{c|}{$\mathcal{H}_1$ (Zauner)} &  \multicolumn{2}{c|}{$\mathcal{H}_{\eta}$} & $\mathcal{H}_{\eta^2}$  \\[0.5mm] \hline  
&\multicolumn{2}{c|}{}&\multicolumn{2}{c|}{}& \\[-3.5mm] 
Eigenvalue & \multicolumn{2}{c|}{1} & \multicolumn{2}{c|}{ $\eta$} & $\eta^2$\\ \hline
&\multicolumn{2}{c|}{}&\multicolumn{2}{c|}{}& \\[-3.5mm] 
Dimensionality & \multicolumn{2}{c|}{3}  & \multicolumn{2}{c|}{2}  & 1\\ \hline
& & & & & \\[-3.5mm] 
Eigenspaces of $U_\mathcal{W}$ & $\mathcal{K}_1$ & $\mathcal{K}_{\omega^3}$ & $\mathcal{K}_{\omega}$ & $\mathcal{K}_{\omega^4}$ & $\mathcal{K}_{\omega^2}$  \\[0.5mm] \hline  
& & & & & \\[-3.5mm] 
Eigenvalue & 1 & $\omega^3$ & $\omega$ &  $\omega^4$ & $\omega^2$\\ \hline
& & & & & \\[-3.5mm] 
Dimensionality & 2 & 1 & 1 &1 & 1\\ 
\hline
\end{tabular}
\vskip 4mm
\parbox{14cm}{\caption[Structure of the eigenspaces of $U_\mathcal{Z}$ and its square root $U_\mathcal{W}$ in $d=6$]{Structure of the eigenspaces of $U_\mathcal{Z}$ and its squareroot $U_\mathcal{W}$, with  $\omega = e^{2\pi i/6}$ and $\eta = e^{2\pi i/3} =\omega^2$. Note that $\UW$ is order 6, but it only has 5 eigenvalues (missing $\omega^5$) because the eigenspace corresponding to eigenvalue 1 is degenerate.}
\label{tab-eigenUZUW}}
\end{center}
\end{table}
Next, we define three operators $R, S$ and $T$ as follows:
\begin{equation}\label{RST}
\begin{split}
R &= (D_{03} + D_{30} + D_{33})/\sqrt{3}\\
S &= (D_{03} + \omega^2 D_{30} + \omega^4 D_{33})/\sqrt{3}\\
T &= (D_{03} + \omega^4 D_{30} + \omega^2 D_{33})/\sqrt{3}.
\end{split}
\end{equation}
Note that $\{\qs, D_{03}\qs, D_{30}\qs, D_{33}\qs\}$ and 
$\{\ket{\psi}, R\ket{\psi}, S\ket{\psi}, T\ket{\psi}\}$ have the same linear span. We will prove that $\ket{\psi} \in \mathcal{H}_{\eta}$ implies $S\ket{\psi} = 0$ and $R\ket{\psi} = \ket{\psi}$, and that $\ket{\psi} \in \mathcal{H}_{\eta^2}$ implies $T\ket{\psi} = 0$
and $R\ket{\psi} = -\ket{\psi}$. It will immediately follow that the linear span above is 2-dimensional.

From their definitions, one can verify the following properties of $R, S$ and $T$:
\begin{equation}S = T^\dagger, \hspace{10mm} S^2 = T^2 = 0, \hspace{10mm} R^2 = \eye,\end{equation}
\begin{equation}ST = \eye + R, \hspace{15mm}TS = \eye - R.\end{equation}
Moreover, $ST/2$ and $TS/2$ are rank-3 projection operators that are orthogonal to each other. One can also verify the following  commutation relations between $R,S,T$ and $\UW$:
\begin{equation}U_\mathcal{W}R = R U_\mathcal{W}, \hspace{10mm} U_\mathcal{W}S = \omega^4 S U_\mathcal{W}, \hspace{10mm} U_\mathcal{W}T = \omega^2 T U_\mathcal{W}.
\end{equation}
These equations tell us how $R,S$ and $T$ permute the eigenspaces of $U_\mathcal{W}$. For example, if $\qsp$ is an eigenvector of $U_\mathcal{W}$ with eigenvalue $\omega^2$, then $U_W S\qsp = \omega^4 S U_W \qsp = S \qsp$, so $S \qsp$ is an eigenvector of  $U_\mathcal{W}$ with eigenvalue 1. A full description of the action of $S$ and $T$ on the eigenspaces of $U_\mathcal{W}$ is shown in \cref{fig-RST} (the action of $R$ is not shown because $R$ commutes with $U_\mathcal{W}$ and simply leaves the eigenspaces invariant).
\begin{figure}[h]
\centering
\includegraphics[scale=0.4]{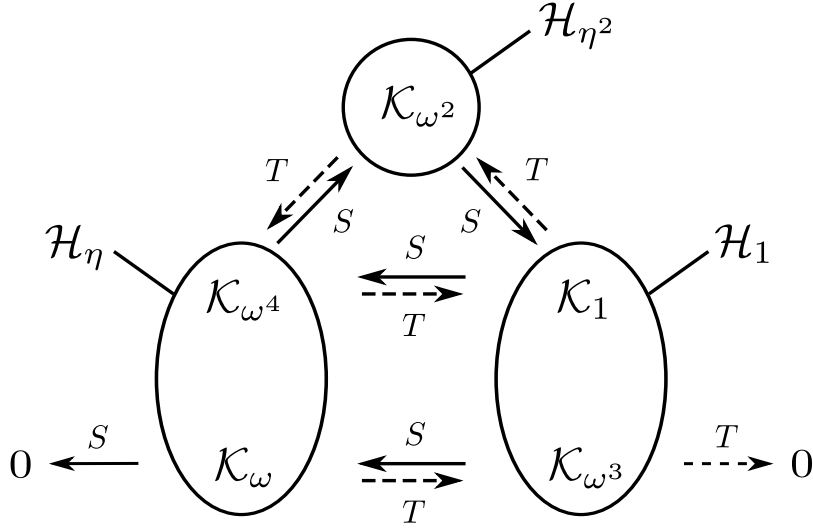}
\vskip 2mm
\parbox{14cm}{\caption[Action of $S$ and $T$ on eigenspaces of $\UW$ in $d=6$]{Action of $S$ (solid arrow) and $T$ (dashed arrow) on the eigenspaces of $U_W$. Note that $\mathcal{K}_1$ (2D) and $\mathcal{K}_{\omega^3}$ (1D) span the Zauner subspace $\mathcal{H}_{1}$ (3D), while $\mathcal{K}_{\omega^4}$ (1D) and $\mathcal{K}_{\omega}$ (1D) span $\mathcal{H}_{\eta}$ (2D), and $\mathcal{K}_{\omega^2}$ (1D) is identical to $\mathcal{H}_{\eta^2}$ (1D).}
\label{fig-RST}}
\end{figure}

Let $\ket{k_0}$, $\ket{k_1}$, $\ket{k_2}$, $\ket{k_3}$ and $\ket{k_4}$ be non-zero eigenvectors of $U_W$ belonging to the eigenspaces $\mathcal{K}_{1}$,  $\mathcal{K}_{\omega}$,  $\mathcal{K}_{\omega^2}$,  $\mathcal{K}_{\omega^3}$ and $\mathcal{K}_{\omega^4}$ respectively. Except for $\ket{k_0}$, the rest of them are unique up to a scalar, because the corresponding eigenspaces are 1-dimensional. As seen from the diagram, we have $S\ket{k_1} = 0$. We are going to prove that $S\ket{k_4}=0$, so that $S\ket{\psi} = 0$ for any $\ket{\psi} \in \mathcal{H}_{\eta}$.

Suppose otherwise, that $S\ket{k_4} \ne 0$. It has to be the case that $S\ket{k_0} = 0$, because otherwise $S\ket{k_0}$ will be a non-zero vector in the 1-dimensional eigenspace $\mathcal{K}_{\omega^4}$, which implies $S\ket{k_0} = \alpha \ket{k_4}$ for some $\alpha \ne 0$, and therefore $S\ket{k_4} = \alpha^{-1}S^2 \ket{k_0} = 0$, contradicting the assumption that $S\ket{k_4} \ne 0$. Since $S\ket{k_4}$ is a non-zero vector in the 1-dimensional eigenspace $\mathcal{K}_{\omega^2}$, we must also have $S\ket{k_4} = \beta \ket{k_2}$ for some $\beta \ne 0$, and therefore $S\ket{k_2} = \beta^{-1}S^2 \ket{k_4}=0.$ So from the assumption that $S\ket{k_4} \ne 0$, we deduce that $0 = S\ket{k_1} = S\ket{k_2} = S\ket{k_0}$, which implies $0 = TS\ket{k_1} = TS\ket{k_2} = TS\ket{k_0}$, which in turn means that $TS$ is orthogonal to the 4-dimensional subspace spanned by  $\mathcal{K}_1$,  $\mathcal{K}_{\omega}$ and  $\mathcal{K}_{\omega^2}$. This contradicts the fact that $TS/2$ is a rank-3 projection operator. 

Thus, we conclude that $S\ket{k_4} = 0$, and that $S\ket{\psi} = 0$ for any $\ket{\psi} \in \mathcal{H}_{\eta}$. The identity $R\ket{\psi} = \ket{\psi}$ immediately follows from $0 =TS\ket{\psi} = (\eye-R)\ket{\psi}$. Note that $T\ket{\psi} \ne 0$ (because $TS$ is orthogonal to $ST$) is a non-zero vector in $\mathcal{H}_1$. $T\ket{\psi}$ is not proportional to $\ket{\psi}$ because $T\ket{\psi}$ lies in $\mathcal{H}_{1}$ and $\ket{\psi}$ lies in $\mathcal{H}_{\eta}$. Therefore $\ket{\psi}$, $R\ket{\psi}$, $S\ket{\psi}$ and $T\ket{\psi}$ indeed span a 2-dimensional subspace. 

It is a similar reasoning that leads to $T\ket{k_2} = 0$, which implies that $T\ket{\psi}=0$ and $R\ket{\psi}=-\ket{\psi}$ for any $\ket{\psi} \in \mathcal{H}_{\eta^2}$. We conclude that $\{\ket{\psi}$, $R\ket{\psi}$, $S\ket{\psi},T\ket{\psi}\}$ spans a 2D subspace whenever $\qs \in \mcH_\eta \cup \mcH_{\eta^2}$, thus proving the second statement of the theorem.
\end{proof}

We have provided an analytical explanation for why $2$-dimensional small SICs occur in dimension $d=6$. Unfortunately, this argument cannot be applied to explain the 3-dimensional small SICs found in $d=9$. If one tries to use the construction described in the proof of \cref{thm-smallSICs} for the case $d=9$, one would obtain 9 vectors which do not span a $3$-dimensional subspace and whose overlaps are not constant (even though they only take two different values). The construction  also fails to produce SICs  in dimensions $12$ and $15$.  We summarize the situation in \cref{tab-d691215}, which shows the dimensionality of the subspace spanned by the $d^2/9$ vectors obtained using the construction in \cref{thm-smallSICs} for different eigenspaces of $\UZ$ from which $\qs$ is chosen.

\begin{table}[h]
\centering
\begin{tabular}{M{1cm} M{1.5cm} M{1.5cm} M{1.5cm} M{1.5cm}}
& \cellcolor{gray!25} $d=6$ & \cellcolor{gray!25}$d=9$ & \cellcolor{gray!25} $d=12$ & \cellcolor{gray!25} $d=15$\\
\cellcolor{gray!25}$\mathcal{H}_1$ & 4 & 8 & 12 & 15\\
\cellcolor{gray!25}$\mathcal{H}_{\eta}$ & 2 & 7 & 8 & 15\\
\cellcolor{gray!25}$\mathcal{H}_{\eta^2}$ & 2 & 6 & 8 & 10 \\
\end{tabular}
\vskip 4mm
\parbox{12cm}{\caption[Dimensions of the embedded subspaces in $d=6,9,12$ and $15$]{Dimensions of the spans of the orbits of $\qs$ under the subgroup generated by $D_{03}$ and $D_{30}$ in higher dimensions for three eigenspaces of $\UZ$ that $\qs$ is chosen from.}
\label{tab-d691215}}
\end{table}

Nevertheless, the fact remains that in dimension $d=9$, instead of choosing a vector in $\mathcal{H}_{\eta}$ or $\mathcal{H}_{\eta^2}$ and following the construction in the proof of \cref{thm-smallSICs}, if one starts from an initial vector $|\psi\rangle$ in the Zauner subspace (not necessarily a SIC fiducial), then one can always find $3$-dimensional SICs among the normal vectors to the linearly dependent sets generated by $\qs$. We have found 4 such small 3-dimensional SICs (there could possibly be more). Upon an inspection of the triple products, we noticed that all of these are unitarily equivalent to the ``most exceptional'' SIC, whose WH fiducial vector is $(0,1,-1)/\sqrt{2}$ \cite{Gelo2013}. This construction is robust in the sense that the resulting small SICs are always the same, regardless of the choice of the initial vector $\qs \in \mcH_1$. If this phenomenon repeats in higher dimensions it might open up an intriguing possibility that one might be able to get a constructive proof of SIC existence in this  way. We did not succeed in finding small SICs in dimension $d=12$, but an exhaustive search was out of computational power's reach.

\chapter{Galois-unitary symmetry}\label{chap:gunitary}
\vskip -4mm
\minitoc

\section{Motivations}\label{sec:motivations}

In contrast to the previous chapter where our study was based on an almost 90-year-old group symmetry, in this chapter we focus on a novel symmetry, which is a generalization of anti-unitary symmetry. 

Quantum physicists are very familiar with unitary transformations. One of the postulates in quantum mechanics states that the evolutions of quantum states are described by unitary transformations. All physical transformations must therefore be unitary. Unitary operators are defined to be operators that preserve the inner product between any two vectors in a Hilbert space:
\be
\langle U x , U y \rangle = \langle x,y \rangle,
\ee
and therefore they also preserve the transition probabilities:
\be
\abs{\langle U x , U y \rangle} = \abs{\langle x,y \rangle}.
\ee
Anti-unitary operators $\bar{U}$ are defined to satisfy
\be
\langle \bar{U} x , \bar{U} y \rangle = \langle x,y \rangle ^*,
\ee
where $^*$ represents complex conjugation. Although anti-unitaries are unphysical and less frequently seen, they also play important roles in physics, such as in representing time-reversal symmetry \cite{Wigner1931} and in entanglement theory \cite{Uhlmann2000}. One can clearly see from the definition that they also preserve transition probabilities. Together with unitaries, they form the only transformations of quantum states that have this probability preserving property, a milestone result from 1931 known as Wigner's theorem \cite{Wigner1931}.

In a restricted region of a Hilbert space, however, it is possible to have symmetries beyond those of unitary or anti-unitary character. Such transformations were recently constructed by Appleby \etal \cite{Appleby2013} to aid the search of SIC-POVMs. They are named Galois-unitaries (or g-unitaries for short), for they are unitary operators composed with Galois automorphisms of a chosen number field extension. The motivation under the construction of g-unitaries stems from an observation by Scott and Grassl \cite{Scott2010} that every known analytical SIC fiducial (except the continuous family in dimension $d=3$) can be expressed in terms of radicals, implying that the corresponding Galois group is solvable. In Appleby \etal \cite{Appleby2013}, SIC fiducials were found to be eigenvectors of a family of g-unitaries. It was hoped that the additional g-unitary symmetry, on top of WH covariance and Zauner symmetry, would help reveal the solution for the SIC problem. Despite significant progresses, the SIC problem remains unsolved. 

On the other hand, one notices that the Galois groups of cyclotomic field extensions (see \cref{sec:fieldtheory} for a review of Galois groups and cyclotomic fields) are quite simple. This is relevant to Mutually Unbiased Bases (to be defined in the next section) since all the components of standard MUB vectors indeed belong to a cyclotomic field. As we will see in \cref{sec:MUBs},  Clifford unitaries simply move one MUB vector to another and permute the bases according to M{\"o}bius transformations.  However, not all permutations on the bases can be realized by Clifford unitaries. This is where g-unitaries come in to provide some of the missing symmetries. Started out as a toy model for SICs, our study of the roles of g-unitaries in the theory of mutually unbiased bases \cite{Appleby2014G} led to a number of new findings. By extending the Clifford group with g-unitaries, we were able to solve the MUB-cycling problem in odd prime-power dimensions (see \cref{sec:MUBcyclers}). We also provided a construction for a distinguished class of quantum states known as MUB-balanced states (see \cref{sec:MUBbalanced}). Although our construction relies on a different technique, namely the g-unitary symmetry, it yields identical results to the construction by Amburg \etal \cite{Amburg2014}.

We want to note that g-unitay operators are not unitary in general, and therefore cannot be physically realized. However, it is possible to simulate them using unitary operators in a larger Hilbert space. We propose such a simulation scheme in \cref{sec:simulation}.

\section{Mutually unbiased bases}\label{sec:MUBs}
In a $d$-dimensional Hilbert space, two orthonormal bases $\{e_i\}$ and $\{e_i'\}$ are called \MUBs ~if
\be
\abs{\braket{e_i}{e_j'}} = \frac{1}{\sqrt{d}}
\ee
for any $i,j = 0,1,...,d-1$. One can see that if a quantum state is completely specified from the measurement outcome probabilities in one basis (i.e. it is one of the basis vectors), then the outcome probabilities from the measurement in the other basis must be a uniform distribution. This is very much like measuring with the position and momentum operators, where learning the precise location of a particle erases all information about its momentum and vice versa. Mutually unbiased bases can therefore be considered as a finite-dimensional analogue of complementary observables. They play important roles not only in foundational studies such as complementarity in quantum mechanics \cite{Schwinger1960}, but also in areas in quantum information such as quantum state estimation \cite{Wootters1989} and quantum cryptography \cite{Bennett2014} (we have provided only a few milestone references).

If the Hilbert space has $d$ dimensions, there can be at most $d+1$ bases that are mutually unbiased to each other \cite{Wootters1989}. These are said to form a complete (or full) set of MUBs, and we will use the acronym MUBs to refer to such a complete set only. It has been shown by Ivonovic \cite{Ivonovic1981} that complete sets of MUBs can be constructed in prime dimensions. This construction was later generalized to prime power dimensions by Wootters and Fields \cite{Wootters1989}. The existence of complete sets of MUBs in general is still an open question. Even in the lowest non-prime-power dimension $d=6$, although it has been found unlikely that a complete set of MUBs could exist, a non-existence proof has not been reached despite numerous research efforts \cite{Bengtsson2007, Butterley2007, Brierley2009, Jaming2009, Paterek2009, Brierley2010, Raynal2011}.

There are many known techniques for constructing a complete set of MUBs, for example by finding maximal Abelian subgroups of the Weyl-Heisenberg group, by taking the eigenbases of unitary operators constructed from the shift operator $X$ and the clock operator $Z$, by using Hadamard matrices, and by taking the WH orbit of an Alltop fiducial \cite{Blanchfield2014}. Here we describe a Clifford-based construction \cite{Appleby2009}, which produces the same set of MUBs that originally appeared in Wootters and Fields' paper \cite{Wootters1989}. We will refer to this particular set as the standard set of MUBs. All MUBs in the thesis are implicitly understood as standard sets of MUBs, unless specifically noted otherwise.

We now restrict ourselves to the case $d=p^n$ is an odd prime power. Making use of the existence of the finite field $\Fd$ when $d$ is a prime power, one can provide a faithful unitary representation $U_S$ of symplectic matrices $S \in \SL(2,\Fd)$. We will appeal to the particular representation described by \cref{def-US-primepower} in \cref{sec:Clifford-app} (see \cite{Appleby2009} for more details). Note that this is the representation for the restricted Galoisian Clifford group, which is not the same as the ordinary Clifford group used in \cref{chap:WH}.

Consider the following $d+1$ matrices in $\SL(2,\Fd)$
\be \begin{split}
S_b &= \bmp 1 & b \\ 0 & 1 \emp \hskip 6mm \text{for } \hskip 3mm b \in \Fd \\
S_\infty & = \bmp 0 & 1 \\ -1 & 0\emp \hskip 3mm \text{for } \hskip 3mm b = \infty.
\end{split}\ee
If we transform the standard basis by the symplectic unitaries $U_{S_b}$ we will obtain $d+1$ bases in a full set of MUB. More explicitly, let $\ket{v}$ denote the standard basis vectors, and let us define
\be
\ket{b,v} = U_{S_b} \ket{v}
\ee
for all $b \in \Fd \cup \{\infty\}$ and $v\in \Fd$, then the claim is that $\ket{b,v}$ are $d(d+1)$ vectors in a complete set of MUBs, where $b$ labels the bases, and $v$ labels the vectors within each basis. To see why, let us first note that for any symplectic matrix
\be \label{def-S-MUBs}
S = \bmp \alpha & \beta \\ \gamma & \delta \emp
\hskip 15mm \det(S) = 1\ee
with $\beta \ne 0$, it directly follows from \cref{def-US-primepower} that
\be \abs{\bra{v} U_S \ket{v'}} = \frac{1}{\sqrt{d}} \ee
for all standard basis vectors $\ket{v}$ and $\ket{v'}$. One can verify that for any $b,b' \in \Fdinf$ and $b\ne b'$, the $\beta$-entry of $S_b^{-1} S_{b'}^{\phantom{1}}$ is non-zero, therefore
\be \begin{split}
\abs{\braket{b,v}{b',v'}} & = \abs{\bra{v}U_{S_b}^\dagger U_{S_{b'}} \ket{v'}} \\
& = \abs{\bra{v}U_{S_b^{-1}} U_{S_{b'}^{{\phantom{1}}}} \ket{v'}} \\
& = \abs{\bra{v}U_{S_b^{-1} S_{b'}^{{\phantom{1}}}} \ket{v'}} \\
& = \frac{1}{\sqrt{d}}
\end{split}\ee 
so that the two bases $b$ and $b'$ are indeed mutually unbiased.

We would like to point out that it also directly follows from the representation in \cref{def-US-primepower} that when $\beta = 0$, $U_S$ simply permutes vectors in the standard basis and adds phases to them. In general (for any value of $\beta$), for any symplectic $S$ in the form given by \cref{def-S-MUBs} one can explicitly work out the action of $U_S$ on the MUB vectors $\ket{b,v}$ to be \cite{Appleby2009}
\be \label{eq-USMUBs}
U_S \ket{b,v} \doteq \begin{cases}
\ket{\frac{\alpha b+\beta}{\gamma b+\delta}, \frac{v}{\gamma b+\delta}} & \text{if } b \ne \infty, \gamma b+\delta \ne 0 \\
\ket{\infty, -\gamma v} & \text{if } b \ne \infty, \gamma b+\delta = 0 \\
\ket{\frac{\alpha}{\gamma}, \frac{v}{\gamma}} & \text{if } b = \infty, \gamma \ne 0 \\
\ket{\infty, \delta v} & \text{if } b = \infty, \gamma = 0 \\
\end{cases}
\ee
where ``$\doteq$'' means ``equals up to a phase''. One now sees the reason behind the use of the symbol $\infty$: it allows us to summarize the permuting actions of $U_S$ on the bases by the M\"{o}bius transformation \cite{Appleby2008}
\be \label{eq-Mobius}
b \hskip3mm \rightarrow \hskip3mm \frac{\alpha b + \beta}{\gamma b + \delta}.
\ee
 
\begin{remark}
The faithful unitary representation we have is for symplectic matrices $S \in \SL(2,\Fd)$, which have $\det S = \alpha \delta - \beta \gamma = 1$. However, a general M\"{o}bius transformation only requires that $\alpha \delta - \beta \gamma \ne 0$, which is the requirement for the general linear group \gls{GL}. Note that if we scale $\alpha, \beta, \gamma$ and $\delta$ by a constant factor, the M\"{o}bius transformation in \cref{eq-Mobius} remains the same. M\"{o}bius transformations are therefore represented by the quotient group of GL, where any two elements $G$ and $G'$ in GL are considered equivalent if they are related by $G' = cG$ for some constant $c \in \Fd$. This is called the projective general linear group \gls{PGL}. In a similar manner, the special linear group $\SL(2,\Fd)$ gives rise to the projective special linear group \gls{PSL}, which is a proper subgroup of PGL. The orders of these groups are provided in \cref{tab-grouporders}.
\begin{table}[h]
\centering
\begin{tabular}{lll}
\rowcolor{gray!25} Notation & Name & Order \\ 
$\GL(2,\Fd)$ & general linear group & $d(d-1)(d^2-1)$\\ 
\gls{PGL} & projective general linear group & $d(d^2-1)$\\
\gls{SL} & special linear group & $d(d^2-1)$\\ 
\gls{PSL} & projective special linear group & $d(d^2-1)/2$\\ 
\end{tabular}
\vskip 4mm
\parbox{12cm}{\caption[The orders of GL and its subgroups]{The orders of the general linear group and its various subgroups.}
\label{tab-grouporders}}
\end{table}

\noindent We observe that if an element $G$ in $\GL(2,\Fd)$
\be \label{def-G-MUBs}
G = \bmp \alpha & \beta \\ \gamma & \delta \emp
\hskip 15mm \det(G) \ne 0\ee
has determinant $\det(G) = \Delta$ which is a quadratic residue in $\Fd$, meaning that $\Delta = x^2$ for some non-zero $x\in \Fd$, we can write $G$ as
\be \label{eq-GxS}
G = \bmp \alpha & \beta \\ \gamma & \delta \emp =
x \bmp \alpha x^{-1} & \beta x^{-1} \\ \gamma x^{-1} & \delta x^{-1} \emp = x S,
\ee
where $S$ is clearly an element of $\SL(2,\Fd)$. Since $G$ and $S$ in the above expression are equivalent, they correspond to the same element in PGL. Therefore, elements in GL whose determinants are quadratic residues do not add to PGL beyond the contribution from SL. But elements whose determinants are quadratic non-residues do. The message here is that symplectic unitaries permutes MUB bases according to M\"{o}bius transformations, but not all M\"{o}bius transformations can be realized by these permutations. We will later see that Galois-unitaries come in to help supply the missing transformations (all of them in certain cases, and some of them in other cases).
\end{remark}

\section{The Clifford group extended by g-unitaries}\label{sec:gunitaries}
We again want to remind that throughout this chapter, we exclusively use the term Clifford group to refer to the restricted Galoisian Clifford group \cite{Appleby2009}, as opposed to the ordinary version used in the previous chapter. The symplectic group $\SL(2,\fd{F}_d)$ is one key ingredient in the construction of the Clifford group. In \cref{sec:Clifford-app} we describe a faithful unitary representation of $\SL(2,\fd{F}_d)$. Here, the question of interest is: is it possible to extend this representation to also include all linear transformations $G$ in the discrete phase space that have determinant $\Delta=\det(G) \ne 1$? According to \cref{eq-OmegaG}, such a transformation scales the symplectic area by a factor of $\Delta$. In order for the group law in \cref{eq-grouplaw-p} to hold, the representation of $G$ should also transform $\omega \mapsto \omega^{\Delta}$, otherwise it will not be an automorphism of the Weyl-Heisenberg group. This is where Galois automorphisms come into the picture (see \cref{sub:Galois} for an introduction to Galois automorphisms), as they do precisely what we need: 
\be g_{\Delta}:\omega \mapsto \omega^{\Delta}.\ee

In the special case when $\Delta=\det(G) =-1$, $G$ is called an anti-symplectic matrix in $\ESL(2,\fd{F}_d)$, and it is represented by an anti-unitary transformation, i.e. an ordinary unitary transformation following complex conjugation. The Clifford group extended by these anti-unitaries is called the extended Clifford group, which was well studied in \cite{Appleby2005}.

We now consider the case when $G$ is an arbitrary element of $\GL(2,\Fd)$, i.e. a $2\times 2$ matrix with determinant $\Delta \ne 0$ (so that it is invertible) over the field $\fd{F}_d$. Although it is possible to analyze the most general case of odd prime power dimensions right away, we would like to start with the simpler case when $d=p$ is an odd prime to explain the core concepts first, then add the technical complications of the general case $d = p^n$ later.

\subsection{In odd prime dimensions}\label{sec:gunitaries-prime}
Let the dimension $d=p$ be an odd prime, and let $G$ be any element in $\GL(2,\Fd)$
\be
G = \bmp \alpha & \beta \\ \gamma & \delta \emp
\ee
where $\alpha, \beta , \gamma,\delta \in \F_d$ and $\Delta \equiv \det(G) \ne 0$. We can always decompose $G$ into
\be \label{eq-GSK}
G = SK_{\Delta},
\ee
where, for any $x\in \Fd$, we define $K_x \in \GL(2,\Fd)$ to be
\be\label{def-Kx}
K_x \equiv \bmp 1 & 0 \\ 0 & x \emp, \ee
and
\be
S = GK_{\Delta}^{-1} = \bmp \alpha &\beta \Delta^{-1} \\ \gamma & \delta \Delta^{-1} \emp
\ee
Note that $S$ has determinant $1$, so $S\in \SL(2,\Fd)$ and can be represented by the unitary $U_S$ given in \cref{def-US-prime}. The matrix $K_{\Delta}$ has determinant $\Delta$ and will be represented by the Galois automorphism $g_{\Delta}$:
\be
K_{\Delta} = \bmp 1 & 0 \\ 0 & \Delta \emp
\hskip 3mm \rightarrow \hskip 3mm 
g_{\Delta}:\omega \mapsto \omega^{\Delta}.
\ee
Therefore $G$ can now be represented by a Galois-unitary (or g-unitary for short) $U_G$:
\be \label{def-UG}
G = SK_{\Delta}\hskip 3mm \rightarrow \hskip 3mm 
U_G \equiv U_S g_{\Delta}.	
\ee
This concept of a g-unitary was introduced in \cite{Appleby2013} as a generalization to an anti-unitary transformation, which can be realized by first applying complex conjugation, and then applying a unitary transformation. Similarly, to realize a g-unitary, one first applies a Galois automorphism, and then performs a unitary transformation.

It is worth emphasizing that $U_G$ is not a unitary operator except when $\det(G)=1$. It is not even a linear operator in general, so it cannot be expressed in a matrix form. The action of $U_G$ on a vector in the Hilbert space is 
\be U_G\qs = U_S g_{\Delta}(\qs),\ee
where $g_{\Delta}(\qs)$ denotes the vector obtained by applying $g_{\Delta}$ to the components of $\qs$ in the standard basis. 

Furthermore, g-unitaries are only defined to act on vectors (or matrices) whose components belong to the cyclotomic field $\fd{Q}(\omega)$. Because of this restriction, we have to verify that these added Galois automorphisms can act on the whole Clifford group. Looking at the representation in \cref{def-US-prime} one might ask whether the overall factor $e^{i\phi}/\sqrt{p}$ is in the cyclotomic field. Let us recall the Gaussian sum \cite{Berndt1998}
\be \label{eq-Gaussiansum}
\sum_{x=0}^{p-1} \omega^{x^2}=
\sum_{x\in Q} \omega^x - \sum_{x\in N} \omega^x = \begin{cases}
\sqrt{p} & \text{if } p=4k+1 \\
i\sqrt{p} & \text{if } p=4k+3. \end{cases}
\ee
It follows that the factors $e^{i\phi}/\sqrt{p}$ in \cref{def-US-prime} belong to the cyclotomic field, and therefore so do the entries of the symplectic unitaries $U_S$. We are now allowed to use symplectic unitaries along with the Galois automorphisms of the cyclotomic field extension to represent $\GL(2,\Fd)$. This representation is faithful, as shown in the following lemma.
\begin{lemma}\label{lem-faithfulness}
Let the dimension $d=p$ be an odd prime and let $G_1$ and $G_2$ be any two elements of $\GL(2,\Fd)$. It then holds that
\be 
U_{G_1} U_{G_2} = U_{G_1 G_2}.
\ee
\end{lemma}
\begin{proof}
Explicitly, let
\be G_1 = \bmp \alpha_1 & \beta_1 \\ \gamma_1 & \delta_1 \emp
\hskip 15mm 
G_2 = \bmp \alpha_2 & \beta_2 \\ \gamma_2 & \delta_2 \emp,
\ee
and let us write them in the same form as \cref{eq-GSK}:
\be G_1 = S_1 K_1 \hskip 15mm G_2 = S_2 K_2. \ee
Note that $K_1$ and $K_2$ are short notations for $K_{\Delta_1}$ and $K_{\Delta_2}$.
Then we can write
\be G_1 G_2 = S_1 K_1 S_2 K_2 = S_1 (K_1 S_2 K_1^{-1}) K_1 K_2, \ee
where
\be K_1 S_2 K_1^{-1} = \bmp \alpha_2 & \beta_2 \Delta_1^{-1}\\ \gamma_2\Delta_1 & \delta_2 \emp, \ee
is symplectic and is therefore represented by a symplectic unitary $U_{K_1 S_2 K_1^{-1}}$. On the other hand, in the g-unitary representation we have
\be \label{eq-UG1UG2}
U_{G_1} U_{G_2} = U_{S_1} g_1 (U_{S_2} g_2) = U_{S_1} g_1(U_{S_2}) g_1 g_2, \ee
where $g_1$ and $g_2$ are short notations for $g_{\Delta_1}$ and $g_{\Delta_2}$.
From \cref{eq-Gaussiansum} it follows that
\be 
g_1(e^{i\phi}/\sqrt{p}) = 
\begin{cases} e^{i\phi}/\sqrt{p}  & \text{if } \Delta_1 \in {\bf Q} \\
-e^{i\phi}/\sqrt{p}  &\text{if } \Delta_1 \in {\bf N} 
\end{cases}
\ee
In addition to the fact that $l(-\beta \Delta_1^{-1}) = l(-\beta) l(\Delta_1^{-1})$, which equals to $ l(-\beta)$ if $\Delta_1 \in {\bf Q}$, and $-l(-\beta)$ if $\Delta_1 \in {\bf N}$, and in view of the representation given in \cref{def-US-prime}, we obtain
\be \label{eq-g1US2}
g_1(U_{S_2}) = U_{K_1 S_2 K_1^{-1}}.
\ee
Given that the representation of SL is faithful, \cref{eq-UG1UG2} can be rewritten as
\be U_{G_1} U_{G_2} =  U_{S_1} U_{K_1 S_2 K_1^{-1}} g_1 g_2 = U_{S_1 (K_1 S_2 K_1^{-1}) K_1 K_2} = U_{G_1 G_2}\ee
as desired.\\
\end{proof}

\begin{remark}
We want to remind that the action of g-unitaries is restricted to vectors in the Hilbert space whose components belong to the cyclotomic field $\fd{Q}(\omega)$. Although this set is dense within the Hilbert space, one should not be tempted to play the usual trick of taking limits because these transformations are not continuous, as will be demonstrated below. The good news, however, is that this restricted subset of the Hilbert space g-unitaries can act on includes all MUB vectors, since the components of the standard basis vectors are in the field, and since every other MUB vector can be obtained from the standard basis by a Clifford unitary, whose entries are also in the field. It also includes all vectors that can be obtained from the standard basis by applying a larger set of transformations known as the Clifford hierarchy, which is large enough for universal quantum computation \cite{Gottesman1999}.

The Hilbert space norm is preserved by g-unitaries only if it is rational, which need not always be the case. This means that g-unitaries are wildly discontinuous. As an example, let us consider the g-unitary consisting of only the Galois automorphism $g_2:\omega \mapsto \omega^2$ in dimension $d=5$, and its action on the (unnormalized) vector
\be
\qs = \bmp \omega^2 + \omega^3 \\ 1 \\ 1 \\ 1 \\ 1 \emp
\hskip 3mm \rightarrow \hskip 3mm 
g_2(\qs) = \bmp \omega^4 + \omega \\ 1 \\ 1 \\ 1 \\ 1 \emp.
\ee
Both of these are real vectors. So they can be approximated by rational vectors, which are left invariant by the g-unitary. However $\qs$ has clearly been moved quite a distance in the Hilbert space by $g_2$.

G-unitaries do, however, preserve a norm that is the product of the scalar product of two vectors with all its $d-2$ Galois conjugates \cite{Janusz1973}. In mathematics this is called the field norm. It is unknown whether the field norm has any physical meaning in this case.
\end{remark}

\subsection{In odd prime power dimensions}\label{sec:gunitaries-primepower}

Now that the concept of g-unitaries has been explained, we move on to the general case of odd prime power dimensions $d=p^n$. In this case, complications come from the fact that the finite field $\fd{F}_d$ of order $d$ no longer contains only ordinary integers like the prime field $\fd{Z}_p$. Materials on finite fields (see \cref{sub:finitefields}) and on the Galoisian Weyl-Heisenberg and Clifford groups in odd prime power dimensions (see \cref{sub:Clifford-primepower}) will be assumed.

We start with the first complication that the definition of a g-unitary $U_G$ in \cref{def-UG} no longer makes sense in general. This is because one cannot define the Galois automorphism $g_x: \omega \mapsto \omega^x$ for all elements $x \in \fd{F}_d$: one can only raise $\omega$ to a power which is an integer, not an abstract element of a finite field. In order to still use \cref{def-UG}, we now have to impose a restriction on $G$, namely its determinant $\Delta$ must belong to the ground field $\fd{F}_p$ (i.e. integers mod $p$). All such $G$ form a subgroup that we will denote by $\GL_p(2,\Fd)$.
\begin{definition}
\gls{GLp} is defined to be the subgroup of $\GL(2,\Fd)$ consisting of $2 \times 2$ matrices whose (non-zero) determinants are in the ground field $\Fp$.
\end{definition}

We can now safely use \cref{def-UG} to define $U_G$ for any $G\in$ $\GL_p(2,\Fd)$. Hence the case of prime power dimensions (where the power $n >1$) differs from the prime dimensional case in that the g-extended Clifford group only includes a proper subgroup of GL. Another difference is that we should now use the more general formula \cref{def-US-primepower} for the unitary representation $U_S$ of a symplectic $S$:
\be
S = \bmp \alpha & \beta \\ \gamma & \delta \emp 
\hskip 3mm \rightarrow \hskip 3mm 
U_S =
\begin{cases} l(\alpha ) \sum_{x\in\fd{F}_d}\omega^{\tr(\alpha \gamma x^2/2)}\ket{\alpha x}\bra{x} &\text{if } \beta = 0 \\
h(\beta)\sum_{x,y \in \fd{F}_d} \omega^{\tr(\frac{\delta x^2 - 2xy + \alpha y^2}{2\beta})} \ket{x}\bra{y} 
&\text{if }\beta \neq 0,
\end{cases}
\ee
where we have defined $h(x) = (-i)^{-n(p+3)/2}l(-x)/\sqrt{d}$.

The main result of this section is the faithfulness of the above representation described by the following theorem.
\begin{theorem}\label{thm-faithfulness}
Let the dimension $d=p^n$ be an odd prime power, where $n$ is odd, and let $G_1$ and $G_2$ be any two elements of $\GL_p(2,\Fd)$. Then it holds that
\be 
U_{G_1} U_{G_2} = U_{G_1 G_2}.
\ee
\end{theorem}
\begin{proof}
The proof is similar to the proof of \cref{lem-faithfulness}, so we will not repeat the set-up here. The only thing we need to do is to re-verify \cref{eq-g1US2}:
\be \label{eq-g1US2-2}
g_1(U_{S_2}) = U_{K_1 S_2 K_1^{-1}}. \ee

First, recall the basic fact that $\Fp = \{x\in\Fd: x^p = x\}$. Let $\theta$ be a primitive element of $\Fd$. One can then show that
\be \theta_p \equiv \theta^{1+p+\cdots+p^{n-1}}\ee
belongs to $\Fp$ by verifying that $\theta_p^p = \theta_p$, and that $\theta_p$ is a primitive element of $\Fp$. We can then write $\Delta_1 = \theta_p^u$, and $g_1 = g_{\theta_p}^u$, for some integer $0\le u \le p-2$. Note that $\Delta_1$ is a quadratic residue in $\Fp$ if and only if $u$ is even.

Since $h(\beta)^2$ is rational while $h(\beta)$ is irrational, and since $g_{\theta_p}$ generates the Galois group, we must have
\be g_{\theta_p}\left(h(\beta)\right)= - h(\beta), \ee
which implies
\be g_{\Delta_1}\left(h(\beta)\right)= (-1)^u h(\beta). \ee
As $n$ is odd, we can use Lemma 1 in \cite{Appleby2009} to write
\be g_{\Delta_1}\left(h(\beta)\right) = h(\Delta_1\beta). \ee
Since $l(\alpha)$ is either 1, -1, or 0, and is therefore rational, it is invariant under $g_{\Delta_1}$. Thus
\be
g_{\Delta_1}U_{S_2} =
\begin{cases} l(\alpha ) \sum_{x\in\fd{F}_d}\omega^{\tr(\Delta_1\alpha \gamma x^2/2)}\ket{\alpha x}\bra{x} &\text{if } \beta = 0 \\
h(\Delta_1\beta)\sum_{x,y \in \fd{F}_d} \omega^{\tr(\frac{\Delta_1(\delta x^2 - 2xy + \alpha y^2)}{2\beta})} \ket{x}\bra{y} 
&\text{if }\beta \neq 0,
\end{cases}
\ee
which validates \cref{eq-g1US2-2}.\\
\end{proof}

\begin{remark}
It can be seen from the proof above that when $n$ is even, we have $U_{G_1} U_{G_2} = \pm U_{G_1 G_2}$. The representation in this case is therefore ``close to faithful'' in a sense.
\end{remark}

\section{Arithmetic of g-unitaries}\label{sec:arithmetic}

This section provides some basic arithmetic of g-unitaries. The derivations are straightforward, but they might help readers better understand how to handle these non-linear g-unitary operators. In any case, it would be useful to have a list of formulae that we can readily use in later calculations. Throughout this section we will always use the notation $U_G$ to refer to the g-unitary representing an element $G \in \GL_p(2,\Fd)$
\be
G = \bmp \alpha & \beta \\ \gamma & \delta \emp
\hskip 10mm \alpha, \beta , \gamma,\delta \in \F_d
\hskip 10mm \Delta \equiv \det(G) \in \fd{F}_p \backslash \{0\},
\ee
and decompose $U_G$ into a Galois automorphism $g_{\Delta}$ followed by a symplectic unitary $U_S$
\be
U_G = U_S g_{\Delta},
\ee
where $S=GK_{\Delta^{-1}}$ is an element of $\SL(2,\Fd)$, with $K_{\Delta}$ being a matrix of determinant $\Delta$ as defined in (\ref{def-Kx}). Where there is no ambiguity, we will drop the subscript $\Delta$ and write $g$ for the Galois automorphism for short. It should also be noted that if we omit the parentheses specifying what $g$ acts on, it should be understood that $g$ acts on everything to its right in the expression, for example:
\be
g_1 A g_2 B = g_1 ( A g_2 (B)) = g_1(A) g_1 g_2 (B).
\ee

\subsection{Action on vectors and matrices}
We know how a Galois automorphism $g_{\Delta}$ transforms a number in the cyclotomic field $\fd{Q}(\omega)$: it replaces every $p$-th root of unity $\omega$ by $\omega^{\Delta}$ in that number's decomposition into powers of $\omega$. For a matrix (or a vector), we require all its entries (or components) in the standard basis to be in the cyclotomic field in order to define the action of a Galois automorphism on it. When this condition is met, the action is defined entry-wise:
\be
(g(A))_{j,k} = g(A_{j,k}).
\ee
Accordingly, the action of a g-unitary on a matrix (or a vector) $A$ is:
\be
U_G A  = U_S g(A).
\ee

\subsection{Composition and power}
Let $U_{G_1}$ and $U_{G_2}$ be two g-unitaries. We want to find out the resulting action of applying them one after another. Since these are operators, it is the safest practice to apply them to an arbitrary matrix $A$
\be 
\begin{split}
U_{G_1} U_{G_2} A & =  U_{S_1} g_1 U_{S_2} g_2 A \\
& = U_{S_1} g_1 ( U_{S_2} g_2 (A)) \\
& = U_{S_1} g_1(U_{S_2}) g_1 (g_2 (A)),
\end{split}
\ee
and remove $A$ in the end to get
\be 
U_{G_1} U_{G_2} = U_{S_1} g_1(U_{S_2}) g_1 g_2.
\ee
Similarly, the composition of more than two g-unitaries is given by
\be
U_{G_1} U_{G_2} \cdots U_{G_k} = U_{S_1}g_1(U_{S_2}) g_1 g_2(U_{S_3}) \cdots g_1 g_2...g_k.
\ee
Setting all g-unitaries identical in the previous equation results in the power formula
\be
U_G^k = U_S g(U_S)\cdots g^{k-1}(U_S) g^k.
\ee
If the dimension $d$ is a prime number, we have (by Fermat's little theorem) 
\be
\Delta^{d-1} = 1 \text{  (mod $d$)}.
\ee
This means $g_{\Delta}^{d-1}:\omega \mapsto \omega^{\Delta^{d-1}} = \omega$ is the identity mapping, and therefore
\be
U_G^{d-1} = U_S g(U_S)\cdots g^{d-2}(U_S)
\ee
is an ordinary unitary.

\subsection{The inverse}
The inverse of a g-unitary $U_G = U_S g$ is given by
\be \label{eq-UGinverse}
U_G^{-1} = g^{-1}U_S^{-1} = g^{-1}(U_S^{-1})g^{-1},
\ee
where the inverse of a Galois automorphism $g_{\Delta}$ is given by
\be \label{eq-ginverse}
g_{\Delta}^{-1} = g_{\Delta^{-1}},
\ee
so that $g_{\Delta}g_{\Delta}^{-1} = g_{\Delta}^{-1}g_{\Delta}$ is the identity mapping.
One can prove (\ref{eq-UGinverse}) by verifying that $U_G^{-1} U_G = U_G U_G^{-1} = \eye$. 

When the power $n$ is odd, the representation of $U_G$ is faithful (see \cref{thm-faithfulness}) and another expression for the inverse of $U_G$ is
\be \label{eq-UGinverse2}
U_G^{-1} = U_{G^{-1}},
\ee
because $U_G U_{G^{-1}} = U_{G G^{-1}} = \eye$.

\subsection{Conjugate transposition and the adjoint}
Since the action of a Galois automorphism $g$ on a vector or matrix $A$ is defined entry-wise, it commutes with the transposition on $A$. Meanwhile, the Galois group for the cyclotomic extension is abelian, therefore $g$ also commutes with all other Galois automorphisms in the group including the complex conjugation. This means
\be \label{eq-gdagger}
(g(A))^{\dagger} = g(A^{\dagger}).
\ee
The conjugate transposition of $U_G A$, where $A$ is a matrix or a vector over the cyclotomic field $\fd{Q}(\omega)$, can be calculated as follows:
\be\label{eq-UGconjtrans} 
\begin{split}
(U_G A)^{\dagger} &= (U_S g(A))^{\dagger}\\
&= (g(A))^{\dagger}U_S^{\dagger} \\
&= g(A^{\dagger})U_S^{\dagger},
\end{split}
\ee
or in Dirac notation, when applied to state vectors:
\be\label{eq-UGconjtrans2} 
(U_G \qs)^{\dagger} = g(\bra{\psi}) U_S^{\dagger}.
\ee
\begin{definition} As in \cite{Appleby2013}, we define the adjoint of a g-unitary $U_G$ to be the unique operator $U_G^{\dagger}$ such that
\be\label{def-UGdagger}
\langle U_G x, y \rangle = g(\langle x, U_G^{\dagger} y \rangle) \hskip 10mm \forall x,y \in \fd{Q}(\omega)^d.
\ee
\end{definition}
This is a natural generalization from the case of anti-linear operators \cite{Appleby2013}, where their adjoints are defined so that  $\langle L x, y \rangle = (\langle x, L^{\dagger} y \rangle)^*$, with $^*$ denoting complex conjugation. It turns out that the adjoint of a g-unitary $U_G$ is exactly its inverse, since
\be
\begin{split}
g \langle x,U_G^{-1} y \rangle &= g \langle x, g^{-1}U_S^{\dagger}y \rangle \\
&= g \left(\bra{x}\right) g (g^{-1}U_S^{\dagger}\ket{y}) \\
&= g \left(\bra{x}\right) U_S^{\dagger} \ket{y} \\
&= \langle U_G x, y \rangle.
\end{split}
\ee
If $d=p^n$ and $n$ is odd, we further have
\be\label{eq-UGdagger}
U_G^{\dagger} = U_G^{-1} = U_{G^{-1}}.
\ee

\subsection{Conjugate action on matrices and displacement operators}
Let $A$ be an arbitrary $d \times d$ matrix over the cyclotomic field $\fd{Q}(\omega)$. Its conjugation by a g-unitary $U_G$ is given by 
\be
U_G A U_G^{-1} = U_S g(A) U_S^{-1},
\ee
which can be verified straightforwardly.

Next, we are going to calculate the conjugation of a displacement operator $\Dp$ by $U_G$. We will show that
\be
U_G  \Dp U_G^{-1} = D_{G\p}.
\ee
We first need to know how a Galois automorphism acts on $\Dp$. From the explicit form of the shift and the phase operator, one can see that
\be
g_{\Delta}(X) = X \hskip 15mm g_{\Delta}(Z) = Z^{\Delta},
\ee
which implies
\be
\begin{split}
g_{\Delta}(\Dp) &= g(\omega^{\frac{p_1 p_2}{2}})g(X^{p_1})g(Z^{p_2}) \\
& = \omega^\frac{\Delta p_1 p_2}{2}X^{p_1}Z^{\Delta p_2}\\
& = D_{K_{\Delta} \p}.
\end{split}
\ee
Therefore
\be
\begin{split}
U_G \Dp U_G^{-1} &= U_S g \Dp g^{-1} U_S^{-1}\\
& = U_S g(\Dp) g (g^{-1} U_S^{-1})\\
& = U_S D_{K_{\Delta}\p} U_S^{-1}\\
& = D_{SK_{\Delta}\p}\\
& = D_{G\p}.
\end{split}
\ee

Conjugations of phase point operators $A_{\p}$ used in the discrete Wigner function can be similarly calculated. Recall their definition
\be \label{def-Ap}
A_{\p} = D_{\p} A_0 D_{\p}^{\dagger},
\ee
where
\be \label{def-A0}
A_0 = \frac{1}{d}\sum_{\p} D_{\p}.
\ee
It is clear that $A_0$ is left invariant by the conjugation:
\be U_G A_0 U_G^{-1} = \frac{1}{d}\sum_{\p} D_{G\p} = A_0. \ee
Other phase point operators $A_{\p}$ are transformed into
\be \label{eq-UGAp} \begin{split}
U_G A_{\p} U_G^{-1} &=  (U_G D_{\p} U_G^{-1})( U_G A_0 U_G^{-1})( U_G D_{\p}^{\dagger} U_G^{-1} )\\
& = D_{G\p} A_0 D_{G\p}^{\dagger} \\
& = A_{G\p}.
\end{split} \ee

\section{Geometric interpretation}\label{sec:geometric}

\subsection{Complementarity polytopes}\label{sub:polytopes}
In this section we would like to introduce a class of geometrical objects closely related to \MUBs ~that can help explain the role of g-unitaries from a geometrical perspective. These are called complementarity polytopes and they have been well studied in \cite{Bengtsson2005}. The space under consideration here is called the Bloch space, which will be defined shortly.

We start from the set of $d \times d$ Hermitian matrices of unit trace. This set includes all quantum states represented by density matrices, and it also contains matrices that are not quantum states, namely those that are not positive semidefinite. For any element $H$ in the set, we define the Hermitian traceless matrix $\dot{H}$ by
\be \label{def-dot}
\dot{H} = H - \eye/d.
\ee
We obtain the set of all traceless Hermitian matrices. This new set is clearly closed under addition and multiplication by real scalars. It therefore forms a real vector space, with the zero element (the origin) corresponding to the maximally mixed state $\eye/d$, as can be seen from \cref{def-dot}.

\begin{definition} The aforementioned vector space will be referred to as the Bloch space. A point (which can also be considered as a vector) $\dot{H}$ in the Bloch space corresponds to a Hermitian operator $H$ of unit trace on the Hilbert space via equation \cref{def-dot}. The density operators on the Hilbert space form a convex body in the Bloch space, which will be called the Bloch body.
\end{definition}

Considered as a real vector space, the Bloch space has dimensionality equal to the number of real parameters needed to describe a $d$ by $d$ traceless Hermitian matrix, which can be worked out to be $d^2-1$. The Bloch body, restricted by positive semidefiniteness, is also $(d^2-1)$-dimensional. For example, for the qubit case ($d=2$), the Bloch body of quantum states is a 3-dimensional ball, which is commonly called the Bloch ball. In higher dimensions, the Bloch body is no longer a ball.

Next, we use the standard Hilbert-Schmidt inner product to define a dot product for two vectors $\dot{H}_1$ and $\dot{H}_2$ in the Bloch space by
\be \label{def-dotproduct}
\dot{H}_1 \cdot \dot{H}_2 \equiv \frac{1}{2}\Tr(\dot{H}_1 \dot{H}_2) = \frac{1}{2}\left( \Tr(H_1 H_2)-\frac{1}{d}\right).
\ee
The Bloch space therefore can be thought of as a $(d^2-1)$-dimensional Euclidean space, with the Euclidean distance between any two points $\dot{H}_1$ and $\dot{H}_2$ given by
\be \label{def-distance}
D(\dot{H}_1,\dot{H}_2) =   \sqrt{\frac{1}{2}\Tr\left(\dot{H}_1-\dot{H}_2\right)^2} = \sqrt{\frac{1}{2}\Tr\left(H_1-H_2\right)^2} \hskip 2mm .
\ee

\begin{remark} Pure quantum states are the extreme points of the Bloch body. They all lie on the surface of a sphere of radius $\sqrt{(d-1)/2d}$ centered at the origin. However, not every point on this sphere corresponds to a quantum state, except when $d=2$, in which case the sphere is called the Bloch sphere and it consists entirely of pure states, while the interior of the Bloch ball consists of all mixed states.

We also want to note that if $H_1$ and $H_2$ are density matrices of two MUB vectors in two distinct bases, then the vectors $\dot{H}_1$ and $\dot{H}_2$ in the Bloch space are orthogonal to each other, as their dot product can be seen to vanish.
\end{remark}

	In the Bloch space, one can construct a regular $(d^2-1)$-dimensional simplex, consisting of $d^2$ vertices labeled by $\dot{A}_{\bu}$, where the index $\bu = (u_1,u_2) \in \fd{F}_d^2$ can take $d^2$ values. Since the dihedral angle of a regular $n$-simplex is $\cos^{-1}(n^{-1})$ (see \cite{Krasnodebski1971} for the first general proof, or \cite{Parks2002} for a more elementary calculation), the vertices $\dot{A}_{\bu}$ can be chosen so that
\be
\dot{A}_{\bu}\cdot\dot{A}_{\bv} = \begin{cases}
\frac{d^2-1}{2d} &\text{if } \bu = \bv \\
\frac{-1}{2d} &\text{if } \bu \ne \bv, \end{cases}
\ee
or equivalently in the Hilbert space
\be \label{eq-trAuAv}
\Tr(A_{\bu}A_{\bv}) = \begin{cases}
~d &\text{if } \bu = \bv \\
~0 &\text{if } \bu \ne \bv. \end{cases}
\ee
Eventually, $A_{\bu}$ will be identified with Wootters' phase point operators for odd prime dimensions \cite{Wootters1987}, or kernel operators \cite{Klimov2005} and displaced parity operators \cite{Vourdas2005} in the more general case of odd prime power dimensions. But for now we simply use them to define a regular simplex in the Bloch space. We center the simplex at the origin by imposing the condition
\be \label{eq-sumA}
\sum_{\bu}\dot{A}_{\bu} = 0 \hskip 5mm \Leftrightarrow \hskip 5mm 
\sum_{\bu}A_{\bu} = d\eye.\ee
Each facet (a hyperplane containing $d^2-1$ vertices) of the simplex that does not contain vertex $\dot{A}_{\bu}$ for some $\bu$ consists of all points $\dot{M}$ such that $\Tr(A_{\bu}M)=0$.

What we have so far is a regular simplex centered at the origin created from $d^2$ vertices in the Bloch space. We are now going to impose a combinatorial structure on it, using what is called a finite affine plane. From now on, we will assume that the dimension $d = p^n$ is a prime power, because that is when affine planes are guaranteed to exist.

\begin{definition}A finite affine plane of order $d$ consists of $d^2$ points grouped into subsets that are called lines. The grouping is done in such a way that any two points belong to a unique line, and for every point not belonging to a line there exists a parallel line (two lines are called parallel if they are disjoint subsets of points) containing that point.
\end{definition}

The existence or non-existence of affine planes in general is an open problem. However, there are some cases for which it has been solved. Bruck-Ryser theorem states that if $d=1$ or 2 mod 4, and $d$ is not the sum of two squares, then affine planes of order $d$ cannot exist \cite{Bruck1949}. The existence for $d=10$ has also been ruled out \cite{Lam1989,Lam1991}. Here, we are interested in the case $d=p^n$ is a prime power, when it is known that affine planes exist and that they have the following properties \cite{Bennett1995}:
\begin{enumerate}
\item Each line contains exactly $d$ points.
\item Each point belongs to exactly $d+1$ lines.
\item There are $d+1$ sets of $d$ parallel lines, which totals to $d(d+1)$ lines.
\end{enumerate}

The sets of parallel lines in the last property mentioned above will be called pencils of lines, or just pencils for short. An easy way to visualize a finite affine plane is by associating it with the vector space $\fd{F}_d^2$ (see \cref{fig-affine9} for an example when $d=3$). Although this coordinatization is not required for the construction of affine planes in general, we choose to stick with it here.
\begin{figure}[h]
\centering
\includegraphics[scale = 0.3]{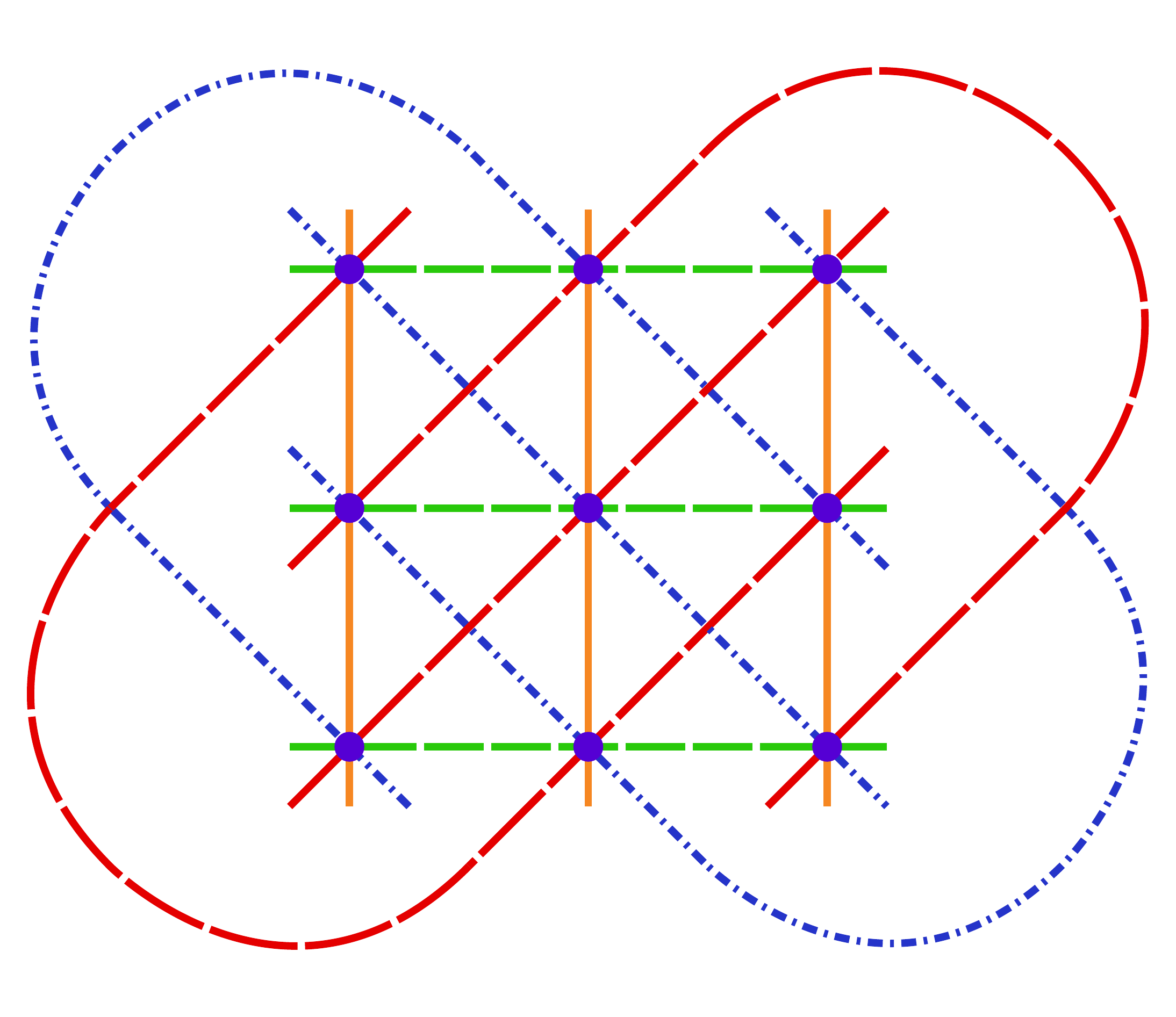}
\vskip 2mm
\parbox{12cm}{\caption[An affine plane of order 3]{Example of an affine plane of order 3 with 9 points and 12 lines. The points are positioned by their coordinates in $\fd{Z}_3^2$. The lines are grouped into 4 pencils, each containing 3 parallel lines marked by the same color. The orange pencil (solid lines) is labeled by $\infty$ since they have a ``slope'' of $\infty$.}
\label{fig-affine9}}
\end{figure}
Particularly, we will use the index $b \in \{\fd{F}_d \cup \infty\}$ to label the $d+1$ pencils, where the symbol $\infty$ refers to the pencil of ``vertical'' lines. Lines in the $b$-th pencil are labeled by the index $v \in \Fd$, and are denoted by $l_v^b$ (the reason behind this labeling is that the line $l_v^b$ will later be made to correspond to the $v$-th MUB vector in the $b$-th basis). The affine plane's points are labeled by $\bu = (u_1,u_2) \in \fd{F}_d^2$.

We associate each line in the affine plane with an operator $P_v^{(b)}$ defined to be the average of all phase point operators $A_{\bu}$ associated with points on that line:
\be
P_v^{(b)} =\frac{1}{d} \sum_{\bu \in l_v^b} A_{\bu}.
\ee
One can verify that $P_v^{(b)}$ is a Hermitian matrix of unit trace, so it corresponds to a point in the Bloch space. In view of \cref{eq-sumA} and the fact that there are $d+1$ lines intersecting at each point, we can invert the above equation to get
\be \label{def-A}
A_{\bu} = \sum_{l_v^b \ni \bu} P_v^{(b)} - \eye,
\ee
which contains a summation of all operators $P_v^{(b)}$ representing lines going through the point $\bu$ represented by $A_{\bu}$.
\begin{figure}[h]
\centering
\includegraphics[scale = 0.3]{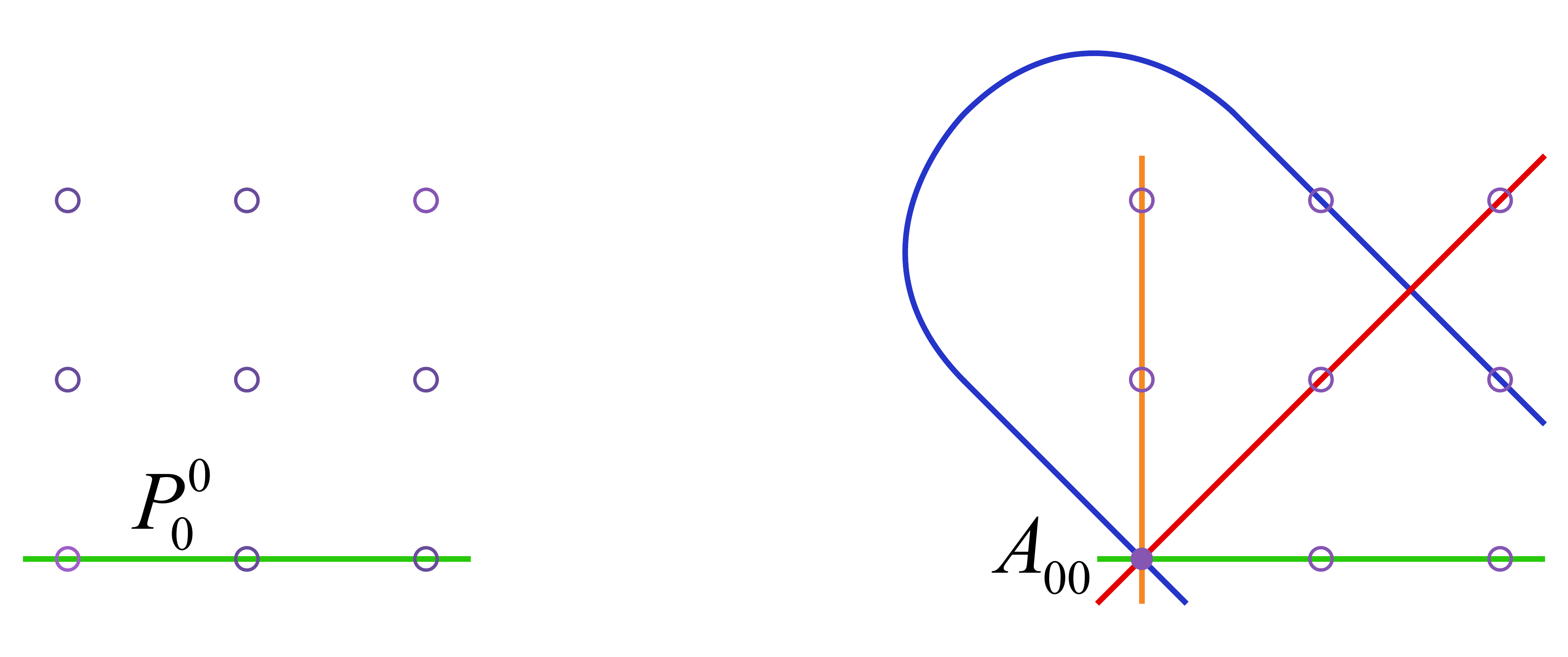}
\vskip 2mm
\parbox{12cm}{\caption[Phase point operators and line operators on an affine plane for $d=3$]{Illustrations of line operators and phase point operators for $d=3$. To the left, the operator $P_0^{(0)}$ corresponding to the line is given by the equation $P_0^{(0)} = (A_{00} + A_{10} + A_{20})/3$. To the right, the intersection point of four lines corresponds to the phase point operator $A_{00} = P_0^{(0)} + P_0^{(1)} + P_0^{(2)} + P_0^{(\infty)} -\eye$. Note that both the lines and the points in the affine plane correspond to points in the Bloch space.}
\label{fig-affine9-2}}
\end{figure}
Using the properties of the affine plane, we can then derive
\be \label{eq-TrP1P2}
\Tr\left(P_{v_1}^{(b_1)} P_{v_2}^{(b_2)}\right) = \begin{cases}
1 & \text{if $b_1 = b_2$ and $v_1 = v_2$  (identical lines)} \\
0 & \text{if $b_1 = b_2$ and $v_1 \ne v_2$  (parallel lines, no common point)} \\
\frac{1}{d} & \text{if $b_1 \ne b_2$ (intersecting lines, one common point).} \\
\end{cases}
\ee
This means that any pencil of $d$ parallel lines (in the affine plane) forms a regular simplex spanning a $(d-1)$-dimensional hyperplane in the Bloch space. There are $d+1$ of these hyperplanes. They are orthogonal to each other and they span the whole Bloch space. 

If we take the convex hull of $d(d+1)$ vertices $\dot{P}_v^{(b)}$ in the Bloch space, we obtain what is called a complementarity polytope \cite{Bengtsson2005}. One can construct complementarity polytopes in all Euclidean spaces of dimension $d^2-1$ for an arbitrary $d$ by splitting the space into $d+1$ totally orthogonal $(d-1)$-dimensional subspaces, placing a regular simplex centered at the origin in each subspace, and taking the convex hull. The special advantage when $d$ is a prime power is that we are able to make use of the combinatorics of an affine plane to inscribe it in a regular simplex. In fact, the phase point operator simplex (whose vertices are the $\dot{A}_{\bu}$) is not the only simplex we can inscribe the complementarity polytope into. It is just one out of $d^{d-1}$ possibilities. Let us consider a vector $\vec{v} = (v_0,v_1,...,v_{d-1},v_{\infty})$ with $d+1$ components $v_b$, which are elements of $\fd{F}_d$. There are $d^{d+1}$ such vectors. For each $\vec{v}$ we define a generalized phase point operator
\be \label{def-generalizedA}
A_{\vec{v}} \equiv \sum_{b} P_{v_b}^{(b)} - \eye.
\ee
The difference between this equation and \cref{def-A} is that here no assumption is made about the lines labeled by $b$ and $v_b$, whereas the lines in \cref{def-A} are required to intersect at the point $\bu$. If we do this for each $\vec{v}$, we obtain $d^{d+1}$ generalized phase point operators $A_{\vec{v}}$. They clearly have unit trace, and for any vectors $\vec{v}$ and $\vec{v}'$ they satisfy
\be
\Tr\left(A_{\vec{v}} A_{\vec{v}'}\right) = \sum_b \Tr\left(P_{v_b}^{(b)} P_{v'_b}^{(b)}\right) -1.
\ee
The $d^{d+1}$ vectors $\vec{v}$ can be grouped into $d^{d-1}$ groups, each containing $d^2$ vectors, in such a way that any two vectors $\vec{v}$ and $\vec{v}'$ within a group agree at exactly one component (this is a non-trivial combinatorial problem, see section 4 in \cite{Bengtsson2005M} for how to do it). One can see from the previous equation that for any $\vec{v}$ and $\vec{v}'$ belonging to the same group
\be
\Tr(A_{\vec{v}}A_{\vec{v}'}) = \begin{cases}
~d &\text{if } \vec{v} = \vec{v}' \\
~0 &\text{if }  \vec{v} \ne \vec{v}'. \end{cases}
\ee
So the points $\dot{A}_{\vec{v}}$ in each group form the vertices of a regular $(d^2-1)$-simplex, much like the phase point operators $\dot{A}_{\bu}$ do in \cref{eq-trAuAv}, resulting in $d^{d-1}$ simplices. Moreover, we have
\be
\Tr\left(A_{\vec{v}} P_v^{(b)}\right) = \begin{cases}
1 &\text{if the $b$-th component of $\vec{v}$ is $v$}\\
0 &\text{otherwise.}\end{cases}
\ee
This means for each $A_{\vec{v}}$ and two associated parallel hyperplanes defined by the two equations $\Tr\left(A_{\vec{v}} H\right) = 0$ and $\Tr\left(A_{\vec{v}} H\right) =1$, every vertex $\dot{P}_v^{(b)}$ of the complementarity polytope must lie on one of the two hyperplanes. The complementarity polytope is therefore confined with between all such pairs of hyperplanes, hence circumscribed by each of the $d^{d-1}$ simplices. We have an illustration for $d=2$ (see \cref{fig-Bloch}), as the Bloch space in this case is 3-dimensional.
\begin{figure}[h]
\centering
\includegraphics[scale = 1]{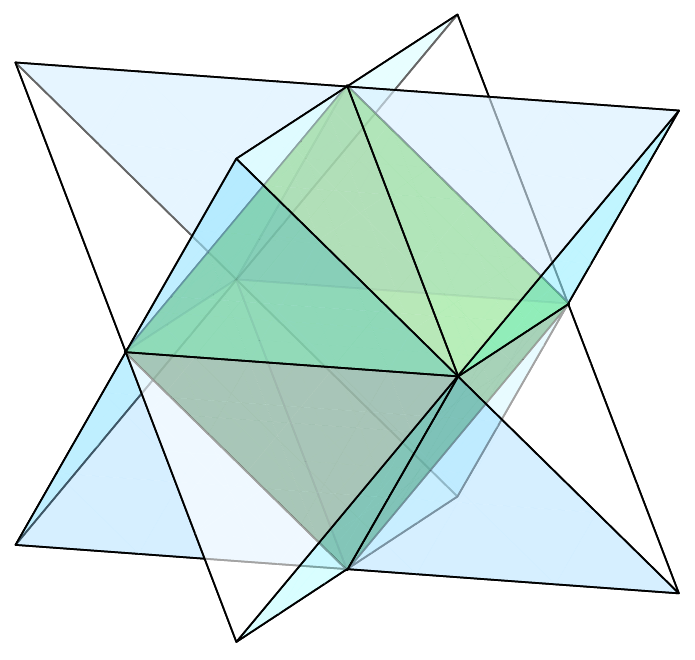}
\vskip 4mm
\parbox{12cm}{\caption[The complementarity polytope for $d=2$]{When $d=2$, the Bloch space has $2^2-1=3$ dimensions and the Bloch body consisting of all quantum states is a 3-dimensional ball. The complementarity polytope (green) is inscribed in $2^{2-1} = 2$ simplices. It is a regular octahedron with $2\cdot 3$ vertices and $2^3$ facets.}
\label{fig-Bloch}}
\end{figure}

\subsection{The symmetry group of the complementarity polytope}\label{sub:symmetrygroup}
Recall from the geometrical meaning of \cref{eq-TrP1P2} that a complementarity polytope consists of $d+1$ simplices, each of which spans a $(d-1)$-dimensional hyperplane, and that these hyperplanes are all orthogonal to each other. The symmetry group of the complementarity polytope is therefore 
\be \label{eq-symgrouphuge}
 S_{d+1} \times S_d \times S_d \times \cdots \times S_d \subset O(d^2-1),\ee
where $S_d$ is the group of all permutations of the vertices of a $(d-1)$-dimensional simplex, and $S_{d+1}$ is the group of all permutations of $d+1$ orthogonal $(d-1)$-dimensional planes. However, we want to impose an extra condition on the complementarity polytope, namely that its vertices must correspond to pure quantum states. This means the polytope is inscribed in the convex Bloch body of quantum states. This might sound like a hard task, but when $d$ is a prime power, a full set of MUBs ~exists and this could be done by letting $P_v^{(b)}$ be the density matrix of the standard MUB vector $r$ in the MUB basis $z$:
\be \label{eq-PMUB}
P_v^{(b)} = \ketbra{e_v^{(b)}}.
\ee
From Wigner's theorem \cite{Wigner1931}, we know the symmetry group of the Bloch body of quantum states, ignoring anti-unitaries for the moment, is
\be \label{eq-symgroupunitary}
U(d)/U(1) \subset SO(d^2-1).
\ee
Since Clifford unitaries simply transform one MUB vector into another, we deduce that the intersection of the two symmetry groups in \cref{eq-symgrouphuge} and \cref{eq-symgroupunitary} contains the Clifford group (ignoring overall phases), which, when $d$ is a prime power, can be written as
\be \label{eq-symgroupClifford}
\SL(2,\Fd) \ltimes \Fd^2 .
\ee
If we include anti-unitaries, this will lead to the extended Clifford group \cite{Appleby2009}
\be \label{eq-symgroupextendedClifford}
\ESL(2,\Fd) \ltimes \Fd^2 .
\ee

On the other hand, from the combinatorics of an affine plane, we know that the complementarity polytope can be inscribed into a regular simplex, whose symmetry group is $S_{d^2}$. The intersection of the symmetry group of the polytope and that of the simplex can be identified by noticing that it must take points to points and lines to lines on the affine plane. Its elements are therefore affine transformations, and the group is isomorphic to 
\be \label{eq-symgroupaffine}
\GL(2,\Fd)\ltimes \Fd^2. 
\ee
This is not quite the same as the result in \cref{eq-symgroupextendedClifford} that we obtain by considering symmetries of the complementarity polytope that also preserve the inscribed body of quantum states. When $d$ is a prime ($d>3$), it is g-unitaries that provide all of the extra symmetries. When $d$ is a prime power, g-unitaries provide only some of the extra transformations, namely those belonging to the subgroup $\GL_p(2,\Fd)$ of $\GL(2,\Fd)$, as a g-unitary $U_G$ is only defined for $G$ with $\det(G) \in \fd{F}_p$.

Thus, we come to the conclusion on the geometrical interpretation of g-unitaries: when their action is restricted to the standard MUB vectors, they are simply rotations in the Bloch space, just like ordinary unitaries. However, unlike unitary operators, which are rotations on the whole Bloch body, first of all g-unitaries do not apply to all quantum states. Then, even on the domain that they do apply, namely $\fd{Q}(\omega)^d$, they are not rotations on the whole domain due to their wildly discontinuous nature. This interpretation therefore has a very limited scope.\\

\begin{remark} Further to the discussion at the end of \cref{sec:MUBs}, we want to note that in odd prime dimensions the projective group PGL is a subgroup of the group $S_{d+1}$ in \cref{eq-symgrouphuge}, which permutes the $(d-1)$-dimensional hyperplanes corresponding to the bases of the MUB. Let us recall that elements of GL whose determinant is a quadratic residue do not add anything to PGL beyond the contribution from SL, as seen from \cref{eq-GxS}. When $d = 3$ mod $4$, element $-1$ of the field $\Fd$ is a quadratic non-residue, so the full set of projective transformations can be obtained from the extended symplectic group represented by unitary and anti-unitary operators. When $d = 1$ mod $4$, $-1$ is a quadratic residue, and we need to include general g-unitaries to obtain the full set of projective transformations. In prime power dimension $d=p^n$, we only obtain all the projective transformations from g-unitaries when $n$ is odd.

Lastly, we want to note that when the vertices of the complementarity polytope are associated with the standard MUB vectors as in \cref{eq-PMUB}, there is a special choice for the phase point operators that takes a rather simple form. When $d$ is odd, one of them is the parity operator
\be \label{def-parity}
A_0 = \frac{1}{d}\sum_{\bv \in \Fd^2} D_{\bv},
\ee
and the others are simply obtained by acting on the parity operator with the WH group's displacement operators
\be
A_{\bu} = D_{\bu} A_0 D_{\bu}^{-1} = \sum_{\bv \in \Fd^2} \omega^{\Omega(\bu,\bv)} D_{\bv}.
\ee

\end{remark}

\section{Simulating g-unitaries using unitaries}\label{sec:simulation}
Quantum mechanics teaches us that physical transformations must be unitary: a non-unitary operator cannot be implemented in any physical system. But that is not to say that it cannot be simulated in a physical system. Simulations of unphysical transformations can be useful in fundamental studies in physics, and have in fact been experimentally realized, for example to study Majorana's equation where anti-unitaries are involved \cite{Casanova2011,Zhang2014}. The key technique used in these papers is to separate the real and the imaginary parts of a quantum state and embed them in a larger Hilbert space. It is worth going through their simulation of an anti-unitary for a qubit, as it will help make the idea transparent. Then, we will propose a scheme for simulating g-unitaries in any odd prime power dimension.

Let $\bar{U}$ be an anti-unitary, which can always be written as
\be\bar{U} = U K,\ee
where $U$ is an ordinary $2 \times 2$ unitary and $K$ stands for complex conjugation. Let
\be \qs = \bmp a_1 + i b_1 \\ a_2 + i b_2 \emp \ee
be a quantum state expressed as a vector in the standard basis, where $a_1,a_2, b_1, b_2$ are real numbers. The action of $\bar{U}$ on $\qs$ is then
\be \bar{U} \qs = U K \qs = U \bmp a_1 - i b_1 \\ a_2 - i b_2 \emp  
= U \qs_{\text{re}} -iU\qs_{\text{im}}. \ee
If we embed $\qs$ into a Hilbert space twice as large using the following mapping $\mathcal{T}$
\be
\bmp a_1 + i b_1 \\ a_2 + i b_2 \emp
\hskip 3mm \xrightarrow{\mathcal{T}} \hskip 3mm 
\bmp a_1\\ a_2\\ b_1 \\ b_2 \emp,
\ee
applying the unitary operator $(\sigma_z \otimes U)$, where $\sigma_z$ is a Pauli matrix, to obtain
\be
\left(
\begin{array}{cccc}
\off & \off & \off & \off \\
& \twotofour{$U$} & & \twotofour{\off$0$} \\
& &  & \\
& \twotofour{$0$} & & \twotofour{$-U$}
\end{array}
\right)
\left(
\begin{array}{c}
a_1 \\
a_2 \\
b_1 \\
b_2
\end{array} 
\right)
=
\left(
\begin{array}{c}
\multirow{2}{*}{\off$U\qs_{\text{re}}$} \\
\relax \\
\multirow{2}{*}{$-U\qs_{\text{im}}$} \\
\relax
\end{array} 
\right),
\ee
and finally applying the inverse of the embedding mapping to go back to the original Hilbert space
\be
\left(
\begin{array}{c}
\multirow{2}{*}{\off$U\qs_{\text{re}}$} \\
\relax \\
\multirow{2}{*}{$-U\qs_{\text{im}}$} \\
\relax
\end{array} 
\right)
\hskip 3mm \xrightarrow{\mathcal{T}^{-1}} \hskip 3mm 
U \qs_{\text{re}} -iU\qs_{\text{im}},\ee
we will get the same result as if we have used $\bar{U}$ to act on $\qs$. Therefore the unphysical anti-unitary $\bar{U}$ can be effectively simulated by the unitary operator $(\sigma_z \otimes U)$ in a larger embedding space.

Let us now make the generalization to g-unitaries. Let $d=p^n$ be an odd prime power dimension. Let $G$ be an element of $\GL_p(2,\Fd)$ and let $U_G$ be its corresponding g-unitary written in the usual form
\be
U_G = U_S g_{\Delta},
\ee
where $U_S$ is unitary, and $\Delta = \det(G) \in \Fp$. The Galois automorphism $g_{\Delta}$ can be thought of as a permutation on the set $\{1,\omega,...,\omega^{p-1} \}$ leaving 1 invariant and mapping $\omega^k$ to $\omega^{\Delta k}$. We will denote this $p \times p$ permutation matrix by $\sigma_{\Delta}$. For any vector $\qs \in \fd{Q}(\omega)^d$, we can embed it into a $dp$-dimensional vector space over the rational field $\fd{Q}$ using the embedding mapping $\mathcal{T}$ defined as follows:
\be
\qs = \bmp q_0^{(1)} + q_1^{(1)}\omega + ... + q_{p-1}^{(1)}\omega^{p-1} \\
q_0^{(2)} + q_1^{(2)}\omega + ... + q_{p-1}^{(2)}\omega^{p-1} \\
\vdots \\
q_0^{(d)} + q_1^{(d)}\omega + ... + q_{p-1}^{(d)}\omega^{p-1}\emp
\hskip 3mm \longrightarrow \hskip 3mm 
\mathcal{T}\left(\qs\right) = \bmp q_0^{(1)} \\ \vdots \\q_0^{(d)}\\ \vdots \\ q_{p-1}^{(1)} \\ \vdots \\ q_{p-1}^{(d)} \emp,
\ee
where $q_j^{(k)}$ are rational numbers and $q_{p-1}^{(1)} = q_{p-1}^{(2)} = ... = q_{p-1}^{(d)} = 0$ (we leave them in the expression as we need to keep track of $\omega^{p-1}$ terms).

It can be seen that by applying the unitary operator $(\sigma_{\Delta} \otimes U_S)$ on $\mathcal{T}\left(\qs\right)$ and then applying $\mathcal{T}^{-1}$ to get back to the $d$-dimensional vector space $\fd{Q}(\omega)^d$, we obtain
\be \mathcal{T}^{-1}\big(  \left(\sigma_{\Delta} \otimes U_S\right) \mathcal{T}\left(\qs\right) \big) = U_S g_{\Delta}(\qs).\ee
This means that the action of $U_G$ can be simulated by the unitary $(\sigma_{\Delta} \otimes U_S)$ in the embedding space using the embedding mapping $\mathcal{T}$ as defined. 

\begin{remark}
We have shown that in principle it is possible to simulate a g-unitary by a physical unitary operator in a larger Hilbert space. In practice, however, this is extremely difficult to implement because the embedding mapping $\mathcal{T}$ is highly discontinuous: given 2 complex numbers that are very close to each other, the rational coefficients in their cyclotomic expansions can be vastly different.
\end{remark}

\section{The MUB-cycling problem}\label{sec:MUBcyclers}

\begin{definition}
Given a full set of \MUB, a MUB-cycler is an operator whose repeated actions on any single basis generate all other bases in succession.
\end{definition}
\begin{note}
Although technically a MUB-cycler is an operator acting on MUB bases in the Hilbert space, we also call an element $G \in \GL(2,\fd{F}_d)$ a MUB-cycler if its representation $U_G$ is a MUB-cycler.
\end{note}
In even prime power dimensions, unitary MUB-cyclers have been constructed \cite{Wootters2007, Kern2010}. In odd prime power dimensions it has been shown that there is no Clifford unitary MUB-cycler \cite{Appleby2009}. However, if the dimension equals 3 mod 4, there are MUB-cycling anti-unitaries \cite{Appleby2009}. Now that we have a notion of Galois-unitaries, we want to see whether they enable us to solve the MUB cycling problem in cases where ordinary unitary and anti-unitary operators fail. The quick answer is yes: in odd prime power dimension $d = p^n$ where the exponent $n$ is odd, there exist g-unitaries that cycle through the MUBs. In this section we will provide a construction for all such operators. For even $n$ we will disprove their existence. These results will come clear from the theorems to follow.

\subsection{Suborder and 3 types of GL elements}

We first want to introduce the notion of the suborder of a $\GL$ element. There are $d+1$ bases in a full set of MUBs in a prime power dimension $d$. In order for a projective permutation of the bases to cycle through all of them, we need an element of $\PGL$ (see \cref{sec:MUBs}) with order $d+1$, or equivalently an element of $\GL$ with an ``effective order" of $d+1$, by which we mean that we only care about its permutation on the bases and neglect its action on individual vectors within each basis. To explain this more precisely let us consider an element $G \in \GL(2,\fd{F}_d)$ and its $m$-th power expressed as:
\begin{equation}
G = \begin{pmatrix}\alpha & \beta \\ \gamma &\delta \end{pmatrix}, \text{ }
G^m = \begin{pmatrix}\alpha_m & \beta_m \\ \gamma_m &\delta_m \end{pmatrix}.
\end{equation}
It follows from the expression for the M\"{o}bius transformation (\ref{eq-Mobius}) that $U_{G^m}$ takes the basis labelled by $b=0$ to basis $b' = \beta_m/\delta_m$ if $\delta_m \ne 0$, and to basis $b'=\infty$ if $\delta_m = 0$. So if $\beta_m =0$, basis $b=0$ will be brought back to itself if repeatedly acted on by $U_G$ for $m$ times. We can now define the suborder of $G$ as follows.

\begin{definition} The suborder $m$ of $G$ is the smallest positive integer for which $\beta_m= 0$.
\end{definition}

\begin{remark} The suborder $m$ of $G$ need not be equal to the order of $G$ because although $U_{G^m}$ brings a MUB basis back to itself, it might permute the vectors within the basis. In general, $m$ is a factor of the order of $G$ (hence the name suborder). Following Lemma \ref{lem-Am}, we will see that the smallest positive integer $m$ for which $G^m$ is proportional to the $2 \times 2$ identity matrix is an equivalent definition of the suborder.\\
\end{remark}

What we can say about the suborder of a $\GL$ element turns out to depend crucially on the nature of its eigenvalues. Let $G$ be an element of $\GL(2,\fd{F}_d)$ with trace $t$ and determinant $\Delta \ne 0$. The eigenvalues of $G$ are roots of the characteristic polynomial $x^2-tx+\Delta = 0$ and are  given by
\begin{equation}
\lambda_{\pm} = (t \pm \sqrt{t^2-4\Delta})/2 \ .
\end{equation}
If $t^2-4\Delta$ is zero or a quadratic residue, i.e. it has a non-zero square root in $\F_d$, then $\lambda_{\pm}$ belong to the field $\F_d$. Otherwise, the eigenvalues do not belong to $\F_d$, but they are still well-defined and they can be included in the extension field $\F_{d^2}$. To deal with these cases, it is convenient to classify $\GL$ elements into three types, as summarized in \cref{tab-GLtypes}.
\begin{table}[h]
\centering
\begin{tabular}{clc}
\rowcolor{gray!25} Type & Definition in terms of $t$ and $\Delta$ & Equivalent definition\\ 
1 & $t^2-4\Delta$ is a quadratic residue  & $\lambda_{\pm} \in \F_d, \lambda_{+} \ne \lambda_{-}$\\ 
2 & $t^2-4\Delta$ is a quadratic non-residue & $\lambda_{\pm} \notin \F_d, \lambda_{+} \ne \lambda_{-}$ \\ 
3 & $t^2-4\Delta = 0$ & $\lambda_{\pm} =t/2 \in \F_d$ \\ 
\end{tabular}
\vskip 4mm
\parbox{12cm}{\caption[Three types of $\GL$ elements]{A classification of $\GL$ elements into three types, among which only type 2 can include MUB-cyclers, as will be seen in the next section.}
\label{tab-GLtypes}}
\end{table}
\subsection{Constructing MUB-cyclers}

Throughout this section we assume that the dimension $d$ is a prime power of the form $d=p^n$, where $p$ is an odd prime number.

\begin{lemma}[Cayley-Hamilton theorem \cite{Hamilton1853} for $2\times2$ matrices]\label{lem-A2}
If $A$ is a $2\times2$ matrix of trace $t$ and determinant $\Delta$, then
\begin{equation}
A^2 = tA - \Delta I, 
\end{equation}
where $I$ is the $2 \times 2$ identity matrix.
\end{lemma}
\begin{proof}
One can explicitly calculate $A^2$ to verify that the lemma is true.
\end{proof}

\begin{lemma}\label{lem-Am}
If $A$ is a $2\times2$ matrix with trace $t$ and determinant $\Delta$, then it holds, for any integer $m \ge 1$, that
\begin{equation}\label{eq-Am}
A^m = s_m A - s_{m-1} \Delta I,
\end{equation}
where the sequence $\{s_m\}$ is defined by the recurrence relation
\begin{equation}\label{eq-smr}
s_{m+1} = t s_m - \Delta s_{m-1},
\end{equation}
with $s_0=0$ and $s_1=1$. Equivalently, $s_m$ can be calculated by
\begin{equation}\label{eq-sm}
s_m = \begin{cases} (\lambda_+^m - \lambda_-^m)/(\lambda_+-\lambda_-) \qquad & \text{if $\lambda_+ \ne \lambda_-$} \\ 
m\lambda_+^{m-1} \qquad & \text{if $\lambda_+ = \lambda_-$} \end{cases}
\end{equation}
where $\lambda_{\pm}$ are roots of the characteristic polynomial $x^2 - t x + \Delta$.
\end{lemma}

\begin{proof}
We will prove the lemma by induction. Let us first note that (\ref{eq-sm}) is equivalent to 
\begin{equation}\label{eq-sm2}
s_m = \sum_{i=0}^{m-1} \lambda_+^{m-1-i} \lambda_-^i \ .
\end{equation}
By definition, $s_2 = t s_1 - \Delta s_0 = t = \lambda_+ + \lambda_-$, therefore (\ref{eq-sm2}) holds for $m=1$ and $m=2$. Suppose (\ref{eq-sm2}) holds for $m=1,2,...,$ up to $m=k$, then
\begin{equation}
\begin{split}
s_{k+1} = t s_k - \Delta s_{k-1} & = (\lambda_+ + \lambda_-) \sum_{i=0}^{k-1} \lambda_+^{k-1-i}\lambda_-^i 
- \lambda_+ \lambda_- \sum_{i=0}^{k-2} \lambda_+^{k-2-i}\lambda_-^i \\
& = \sum_{i=0}^{k-1} \lambda_+^{k-i}\lambda_-^i  
+ \sum_{i=0}^{k-1} \lambda_+^{k-1-i}\lambda_-^{i+1}
- \sum_{i=0}^{k-2} \lambda_+^{k-1-i}\lambda_-^{i+1}  \\
& = \sum_{i=0}^{k-1} \lambda_+^{k-i}\lambda_-^i 
+  \left( \sum_{i=1}^{k} \lambda_+^{k-i}\lambda_-^{i}
- \sum_{i=1}^{k-1} \lambda_+^{k-i}\lambda_-^{i} \right) \\
& = \sum_{i=0}^{k-1} \lambda_+^{k-i}\lambda_-^i + \lambda_-^k 
= \sum_{i=0}^{k} \lambda_+^{k-i}\lambda_-^i \ ,  
\end{split}
\end{equation}
which implies that it also holds for $m=k+1$, and consequently, for all $m \ge 1$.

Equation (\ref{eq-Am}) obviously holds for $m=1$. Lemma \ref{lem-A2} implies that it also holds for $m=2$. Suppose it holds for $m=1,2,...,$ up to $m=k$, then
\begin{equation}
\begin{split}
A^{k+1} = A^k A & = (s_k A - s_{k-1} \Delta I) A \\
& = s_k A^2 - s_{k-1} \Delta A\\
& = s_k ( t A - \Delta I) - s_{k-1} \Delta A\\
& = (s_k t - s_{k-1} \Delta) A - s_k \Delta I \\
& = s_{k+1} A - s_k \Delta I \ ,
\end{split}
\end{equation}
which implies that it also holds for $m=k+1$, and consequently, for all $m \ge 1$.
\end{proof}

\begin{remark}
If $A$ takes the form $A = \begin{pmatrix}\alpha & \beta \\ \gamma &\delta \end{pmatrix}$, we can explicitly rewrite Eq. (\ref{eq-Am}) as 
\begin{equation}
A^m = \begin{pmatrix}
s_m\alpha-s_{m-1}\Delta & s_m\beta \\
s_m\gamma & s_m\delta - s_{m-1}\Delta
\end{pmatrix},	
\end{equation}
from which it can be seen that if $\beta \ne 0$ (for a non-zero determinant, we can always force $\beta \ne 0$ using the canonical form in (\ref{def-canonical}) near the end of this section) the suborder of $A$ is the smallest positive integer $m$ for which $s_m=0$. For such $m$, $A^m$ is proportional to the identity matrix.\\
\end{remark}

\begin{theorem}\label{thm-GLtypes}
Let $G$ be an element of $\GL(2,\F_d)$ with determinant $\Delta$. 
\begin{enumerate}
\item If $G$ is of type 1, then $G$ has suborder of at most $d-1$.
\item If $G$ is of type 2, then $G$ has suborder of at most $d+1$ and satisfies
\begin{equation}\label{eq-Gd1}
G^{d+1} = \Delta I \text{\hskip 6mm ($I$ is the $2 \times 2$ identity matrix).}
\end{equation}
\item If $G$ is of type 3, then $G$ has suborder of at most $d$.\end{enumerate}
\end{theorem}

\begin{proof}
Let $\lambda_{\pm}$ be the eigenvalues of $G$, based upon which we define the sequence $\{s_m\}$ just as in Lemma \ref{lem-Am}. Although $\lambda_{\pm}$ might not be in the field $\F_d$, the sequence $\{s_m\}$ always is, as seen from the recursive definition in (\ref{eq-smr}). Lemma \ref{lem-Am} implies that if $s_m = 0$ for some $m$, then $G^m = -s_{m-1}\Delta I$, and therefore the suborder of $G$ is at most $m$. Let us now consider specific cases. Facts about finite fields (see Section \ref{sub:finitefields}) will be used implicitly.
\begin{enumerate}
\item If $G$ is of type 1, then the eigenvalues $\lambda_{\pm}$ are in $\F_d$, so
\begin{equation}
\lambda_{+}^{d-1} = \lambda_{-}^{d-1} = 1 \ ,
\end{equation}
\begin{equation}
s_{d-1} = (\lambda_{+}^{d-1}-\lambda_{-}^{d-1})/(\lambda_{+}-\lambda_{-}) = 0 \ ,
\end{equation}
and therefore $G$ has suborder of at most $d-1$.
\item If $G$ is of type 2, to include $\lambda_{\pm}$ we create an extension field $\F_{d^2}$ from the base field $\F_d$ and the generator $j \equiv \sqrt{t^2 - 4 \Delta}$. Since $(j^d)^2=(t^2-4\Delta)^d = t^2-4\Delta=j^2 $, we have $j^d = \pm j$. Because $j$ is not in the field $\F_d$ we cannot have $j^d = j$, and it therefore must be the case that $j^d = -j$. As $d$ is odd we have:
\begin{equation}\label{eq-lambdad}
\lambda_{\pm}^d 
= \frac{\left(t \pm j \right)^d}{2^d} 
= \frac{t \pm j^d}{2}
=\frac{t \mp j}{2} = \lambda_{\mp} \ .
\end{equation}
We then use (\ref{eq-sm}) to derive
\begin{equation}
s_{d}= (\lambda_{+}^{d}-\lambda_{-}^{d})/(\lambda_{+}-\lambda_{-}) = -1\ ,
\end{equation}
\begin{equation}
s_{d+1}= \frac{\lambda_{+}^{d+1}-\lambda_{-}^{d+1}}{\lambda_{+}-\lambda_{-}} 
= \frac{\lambda_{+}\lambda_{-} -\lambda_{-}\lambda_{+}}{\lambda_{+}-\lambda_{-}} 
= 0 \ ,
\end{equation}
and therefore
\begin{equation}
G^{d+1} = s_{d+1}G - s_{d}\Delta I = \Delta I \ .
\end{equation}
It follows that $G$ has suborder of at most $d+1$.
\item
If $G$ is of type 3, then $\lambda_{\pm} = t/2$. It follows from (\ref{eq-sm}) that  $s_d = d\lambda_{+}^{d-1} = 0$, so $G$ has suborder of at most $d$.
\end{enumerate}
\end{proof}

\begin{lemma}\label{lem-eta}
Let $G \in \GL_p(2,\F_d)$, i.e. an element $\GL(2,\F_d)$ whose determinant $\Delta$ is in the prime field $\F_p$. Let $\bar{\theta}$ be a primitive element of $\F_{d^2}$ (therefore $\theta = \bar{\theta}^{d+1}$ is a primitive element of $\F_d$). Note that $(d-1)/(p-1)$ is an integer, so we can define $\eta \in \F_{d^2}$ as
\begin{equation}\label{def-eta}
\eta \equiv \bar{\theta}^{(d-1)/(p-1)} \ .
\end{equation}
Then $G$ is of type 2 if and only if it has eigenvalues $\eta^r$ and $\eta^{dr}$, for some integer $r$ in the range $0 < r < (p-1)(d+1)$ that is not a multiple of $(d+1)/2$.

\end{lemma}

\begin{proof}
Assume that $G$ is of type 2, and let $\lambda_{\pm}\notin \fd{F}_d$ be its eigenvalues. Following (\ref{eq-lambdad}), we have $\lambda_{\pm}^d = \lambda_{\mp}$, so we may write
\begin{equation}
\lambda_{+} = \bar{\theta}^k \hskip 15mm \lambda_{-} = \bar{\theta}^{dk}
\end{equation}
for some integer $1 \le k \le d^2-1$. The assumption that $\lambda_{\pm}$ are not elements of $\F_d$ implies that $k$ is not a multiple of $d+1$. The fact that $\Delta \in \F_p$ means
\begin{equation}
\bar{\theta}^{pk(d+1)} = \Delta^p = \Delta = \bar{\theta}^{k(d+1)} \ ,
\end{equation}
or equivalently
\begin{equation}\label{eq-thetak}
\bar{\theta}^{k(p-1)(d+1)} = 1 \ ,
\end{equation}
implying that $(d-1) \mid k(p-1)$. Let $r = k(p-1)/(d-1)$, then $r$ is an integer in the range $0 \le r \le (p-1)(d+1)$. The eigenvalues can then be re-written as
\begin{equation}
\lambda_{+} = \eta^r \hskip 15mm \lambda_{-} = \eta^{dr} \ .
\end{equation}
The requirement that $\lambda_{+} \notin \F_d$ means 
\begin{equation}
\eta^{r} = \bar{\theta}^{r(d-1)/(p-1)} \notin \F_d \ ,
\end{equation} 
which is true if and only if $r(d-1)/(p-1)$ is not a multiple of $(d+1)$, which in turn is equivalent to
$r$ not being a multiple of $(d+1)/2$ because $\gcd((d-1)/(p+1),d+1)=2$.

Conversely, if $G$ has eigenvalues of the form $\lambda_{+} = \eta^r$ and $\lambda_{-}=\eta^{dr}$, where $r$ is not a multiple of $(d+1)/2$, then $\lambda_{\pm}$ are not in the field $\F_d$, and $G$ is therefore of type 2. One can further verify that its trace is in $\F_d$ and its determinant is in $\F_p$ by defining 
\begin{equation}
t \equiv \eta^r + \eta^{dr} \hskip 15mm \Delta \equiv \eta^{(d+1)r}
\end{equation}
and using the facts $\eta^{d^2}=\eta$ and $\eta^{(d+1)p}=\eta^{d+1}$ to check that
\begin{equation}
t^d = t \hskip 15mm \Delta^p = \Delta \ .
\end{equation}
\end{proof}
\begin{remark}
With Lemma \ref{lem-eta}, all type-2 elements of $\GL_p(2,\Fd)$ are now characterized by an integer $r$, via their eigenvalues. In the next theorem, we will pin down exactly which values of $r$ correspond to MUB-cyclers when they exist.\\
\end{remark}
\begin{theorem}\label{thm-MUBcycler}
Let $G \in \GL_p(2,\F_d)$ be of type 2 and let the integer $r$ be as in the statement of Lemma \ref{lem-eta}.
\begin{enumerate}
\item When $n$ is even, $G$ has suborder of at most $(d+1)/2$.
\item When $n$ is odd, $G$ has suborder $d+1$ if and only if $\gcd(r,d+1)=1$.
\end{enumerate}
\end{theorem}

\begin{proof}
Let $\lambda_{\pm}$ be the eigenvalues of $G$ and the sequence $s_m$ be as defined in Lemma \ref{lem-Am}. We recall that the suborder of $G$ is the smallest positive integer $m$ for which $s_m= 0$, which in this case is equivalent to $\lambda_{+}^m = \lambda_{-}^m$, as $G$ is of type 2 and its eigenvalues are distinct.
\begin{enumerate}
\item When $n$ is even, $(d-1)/(p-1) = 1+p+\dotsc+p^{n-1}$ is an even integer, so $(d-1)/2$ is a multiple of $(p-1)$. It then follows from (\ref{eq-thetak}) that
\begin{equation}
\bar{\theta}^{k(d-1)(d+1)/2}=1 \ ,
\end{equation}
which implies $\lambda_{+}^{(d+1)/2} = \lambda_{-}^{(d+1)/2}$, or $s_{(d+1)/2}=0$. Therefore $G$ has suborder of at most $(d+1)/2$ and cannot be a MUB-cycler.

\item When $n$ is odd, $(d-1)/(p-1)$ is an odd integer. It follows from this, and the fact that $\gcd(d+1,d-1)=2$, that $(d-1)/(p-1)$ is co-prime to $d+1$. We have $\lambda_{+}^m = \lambda_{-}^m$ if and only if $\eta^{m(d-1)r}=1$, which in turn is true if and only if $mr(d-1)/(p-1)$ is a multiple of $d+1$. Therefore $G$ has suborder $d+1$ if and only if $r$ is co-prime to $d+1$.
\end{enumerate}
\end{proof}
In summary, in this section we have proved the non-existence of MUB-cyclers when the exponent $n$ is even. When $n$ is odd, we have identified all MUB-cycling elements in $\GL(2,\Fd)$ according to the characteristics of their eigenvalues. Lastly, we want to provide an explicit form for these MUB-cyclers. The proof in the Appendix of \cite{Appleby2008} can be extended to show that for any element in $G \in \GL(2,\F_d)$ with trace $t$ and determinant $\Delta$, where $t^2-4\Delta \ne 0$, there exists $S \in \SL(2,\F_d)$ such that
\begin{equation}\label{def-canonical}
G = S G_c S^{-1} \hskip 15mm
G_c = \begin{pmatrix}
0 & -\Delta \\
1 & t
\end{pmatrix},\end{equation}
where we call $G_c$ the canonical form of $G$. Therefore, an element of $\GL(2,\F_d)$ is a MUB-cycler if and only if it is conjugate to $G_0^r$ where
\begin{equation}
G_0 = \begin{pmatrix}
0 & -\eta^{(d+1)}\\
1 & \eta + \eta^{d}
\end{pmatrix} \ ,
\label{def-G0}
\end{equation}
$\eta$ is defined as in (\ref{def-eta}), and $r$ is an integer co-prime to $d+1$.  Note that the order of $G_0$ is $(p-1)(d+1)$ because this is the smallest integer $r$ such that $\eta^r$ and $\eta^{dr}$,  the eigenvalues of $G^r_0$, are both equal to $1$.\\
\begin{remark}
It follows that anti-symplectic MUB-cyclers exist if and only if the dimension $d=3$ (mod 4), a fact already shown in ~\cite{Appleby2009}. This is because $G_0^r$ is anti-symplectic if and only if $\eta^{r(d+1)}=-1$, which is true if and only if $r$ is an odd multiple of $(p-1)/2$.  If $d=1$ (mod 4) then $(p-1)/2$ is even, so no multiple of $(p-1)/2$ is co-prime to $d+1$. But if $d=3$ (mod 4) one can see that $(p-1)/2$ is co-prime to $d+1$, implying that $G_0^{r(p-1)/2}$ is an anti-symplectic MUB-cycler for every $r$ co-prime to $d+1$.
\end{remark}

\section{Eigenvectors of MUB cyclers}\label{sec:eigenvectors}

One can always find eigenvalues and eigenvectors of an ordinary unitary operator by diagonalizing it. However when it comes to g-unitaries, the situation is tricker. When dealing with g-unitaries, one has to be extra careful because much of our intuition about ordinary unitaries can fail for g-unitaries. For example, a scalar multiplication of an eigenvector of a g-unitary can change its eigenvalues, resulting in the possibility of a g-unitary having infinitely many eigenvalues (see the example below). Or in other cases, a g-unitary might not have any eigenvector at all (see the example below). We will start this section with an example of anti-unitaries, and then proceed to the analysis of the eigenvectors of a special kind of g-unitaries, namely the MUB-cyclers.

\begin{example}
Let $U_A$ be an anti-unitary over the complex field $\fd{C}$, which can be expressed as $U_A = UK$, where $U$ is a unitary and $K$ denotes complex conjugation. We notice that $U_A^2 = UKUK = U\bar{U}$, where $\bar{U}$ denotes the complex conjugate of $U$, is a unitary. Let $\qsp$ be an eigenvector of $U_A$ with the eigenvalue $\lambda$:
\be U_A \qsp = \lambda \qsp \ . \ee
It then follows that
\be U_A^2 \qsp = U K \lambda \qsp = \lambda^* U K \qsp = \abs{\lambda}^2 \qsp \ . \ee
Since $U_A^2$ is unitary, $\abs{\lambda}^2$ must be of modulus 1, and therefore $\lambda = e^{i \theta}$ is a phase. There are two things that follow from this. First of all, let $e^{i\phi}$ be any phase, then 
\be U_A (e^{i\phi} \qsp) = e^{-i\phi}\lambda \qsp = e^{i(\theta-2\phi)}e^{i\phi} \qsp  \ee
so $e^{i\phi}\qsp$ is an eigenvector of $U_A$ with eigenvalue $ e^{i(\theta-2\phi)}$. This means that $U_A$ has a continuum of eigenvalues, whose eigenvectors only differ by an overall phase, and that we can ensure the eigenvalue is 1 by adjusting the phase. Secondly, since $\abs{\lambda}^2$ is real and positive, we must have  $\abs{\lambda}^2 = 1$. Therefore $\qsp$ is an eigenvector of the unitary $U_A^2$ with unit eigenvalue. If $U_A^2$ does not have any eigenvalue equal to 1, then $U_A$ cannot have any eigenvector at all. For a concrete example in $\fd{C}^2$ consider the anti-unitary $U_A = UK$ where
\be U = \frac{1}{\sqrt{2}}\bmp1 & 1 \\ -1 & 1\emp \ , \quad U_A^2 = U \bar{U} = \bmp 0 & 1 \\ -1 & 0\emp \ . \ee
The eigenvalues of $U_A^2$ are $\pm i$, therefore $U_A$ has no eigenvector.
\end{example}

Wigner has characterized the eigenvectors of anti-unitaries \cite{Wigner1960}. It is not straightforward to generalize his results to g-unitaries. However, by restricting ourselves to a special kind of g-unitaries, namely those that have the MUB-cycling property, we are able to provide a complete characterization of their eigenvalues and eigenvectors. The results are summarized in Theorem \ref{thm-eigenvectors}, which states that MUB-cycling g-unitaries always have eigenvectors, which are unique up to multiplication by a scalar, and that we can always find an eigenvector with unit eigenvalue.

Let us first set up some notations and definitions. For the rest of this section, we will always assume that the dimension $d = p^n$ is an odd prime power where the exponent $n$ is odd. $G$ is a fixed element of $\GL_p(2,\fd{F}_d)$ with eigenvalues $\eta^r$ and $\eta^{rd}$ (as in Lemma \ref{lem-eta}), where $r$ is co-prime to $d+1$ so that $G$ is a MUB-cycler (by Theorem \ref{thm-MUBcycler}). We will use $t=\eta^r+\eta^{rd}$ and $\Delta = \eta^{r(d+1)}$ to denote the trace and determinant of $G$, respectively.
\begin{definition}
If we define the multiplicative order of $G$ to be the smallest positive integer $m$ for which $\Delta^m = 1$, then it follows from (\ref{def-eta}) that $\bar{\theta}^{mr(d+1)(d-1)/(p-1)}=1$ (where $\bar{\theta}$ is a primitive element of $\fd{F}_{d^2}$), which is true if and only if $mr$ is a multiple of $p-1$. Since $r$ is odd because it is co-prime to $d+1$, and $p-1$ is even, $m$ must be even. We will therefore use $2m_0$ to denote the multiplicative order of $G$.
\end{definition}
\begin{remark}
It then follows that $\Delta^{m_0} = \pm 1$. Since $2m_0$ is the smallest positive integer for which $\Delta^{2m_0} = 1$, we cannot have $\Delta^{m_0} = 1$, and therefore it must be the case that
\be \label{eq-deltam0} \Delta^{m_0} = - 1 \ . \ee
This implies that $G^{2m_0}$ is a symplectic matrix and $G^{m_0}$ is an anti-symplectic matrix, and correspondingly, $U_G^{2m_0}=U_{G^{2m_0}}$ is unitary and $U_G^{m_0}=U_{G^{m_0}}$ is anti-unitary.
\end{remark}
\begin{definition}
If $\{ \ket{x} \}$ denotes the standard basis, then the parity operator $A$ is defined as
\be \label{def-parity} A = \sum_x \ket{-x}\bra{x} \ .\ee
Alternatively, $A$ can also be defined from the unitary representation of the (unique) element of order 2 in the symplectic group:
\be \label{def-parity2} A = (-1)^{(p-1)/2}U_P \qquad \text{where} \qquad P = \bmp -1 & 0 \\ 0 & -1 \emp \ . \ee
\end{definition}

\begin{theorem}\label{thm-eigenvectors}
Let $d=p^n$ be an odd prime power, where the exponent $n$ is odd, and let $G \in \GL_p(2,\fd{F}_d)$ be a MUB cycler. Let $\fd{Q}(\omega)^d$ be the subspace of the Hilbert space consisting of all vectors whose components (in the standard basis) belong to the cyclotomic field $\fd{Q}(\omega)$.
\begin{enumerate}
\item There exists a non-zero $\qs \in \fd{Q}(\omega)^d$ such that $U_G \qs = \qs$.
\item $\ket{\phi} \in \fd{Q}(\omega)^d$ is an eigenvector of $U_G$ if and only if $\ket{\phi} = \mu \qs$ for some $\mu \in \fd{Q}(\omega)$.
\item The eigenspace of $U_G^{2m_0}$ with eigenvalue 1 is one-dimensional and spanned by $\qs$.
\item $\qsp$ is an eigenvector of the parity operator with eigenvalue $(-1)^{(p-1)/2}$.
\end{enumerate}
\end{theorem}
\begin{proof}
The theorem in an immediate consequence of the following lemmas.
\end{proof}

\begin{lemma}\label{lem-UG2m0}
With notations and definitions as above, suppose $\qsp \in \fd{Q}(\omega)^d$ is an eigenvector of $U_G$ with eigenvalue $\lambda \in \fd{Q}(\omega)$, i.e.
\be\label{eq-UGphi}
U_G \qsp = \lambda \qsp \ ,
\ee
then it holds that
\be
U^{2m_0}_G \qsp = \qsp \ .
\ee
\end{lemma}
\begin{proof}
Applying $U_G$ to (\ref{eq-UGphi}) repeatedly we will obtain
\be \label{eq-UGm0} U_G^{m_0} \qsp = \kappa \qsp \ , \ee
where
\be \kappa = \lambda g_{\Delta}(\lambda) \dots g_{\Delta}^{m_0-1}(\lambda) \ . \ee
Then we apply the anti-unitary $U_G^{m_0}$ to (\ref{eq-UGm0}) to obtain
\be U_G^{2m_0} \qsp = \kappa^* \kappa \qsp \ , \ee
where $\kappa^*$ is the complex conjugate of $\kappa$ and $U_G^{2m_0}$ is unitary. As $n$ is odd, by Theorem \ref{thm-faithfulness} the representation of $\GL_p(2,\F_d)$ is faithful, and therefore 
\be U_G^{2m_0(d+1)} = U_{G^{2m_0(d+1)}} = \eye \ . \ee 
It then follows that $(\kappa^* \kappa)^{d+1} = 1$, which implies $\kappa^* \kappa = 1$ because it has to be a real positive number. Hence, we conclude that $\qsp$ is an eigenvector of $U_G^{2m_0}$ with unit eigenvalue.
\end{proof}

\begin{lemma}\label{lem-dimS1}
With notations and definitions as above, let $\mathcal{S}_{1}$ be the eigenspace of $U_G^{2m_0}$ corresponding to eigenvalue $1$, then $\mathcal{S}_{1}$ is one-dimensional.
\end{lemma}
\begin{proof}
To determine the dimensionality of $\mathcal{S}_{1}$, we will calculate the trace of the projection operator onto that subspace. We first note the following linear algebraic fact. If $U$ is a unitary of order $k$, and $\tau$ is an eigenvalue of $U$ ($\tau$ therefore has to be a $k$-th root of unity), then the projection operator onto the eigenspace corresponding to $\tau$ is given by
\be P_{\tau} = (\eye + \tau^{-1}U + \dotsc + \tau^{-(k-1)}U^{k-1})/k  \ . \ee
One can prove this by verifying that $P_{\tau}^2 = P_{\tau}$ so that $P_{\tau}$ is a projection operator, and that $P_{\tau}$ projects any vector into an eigenvector with eigenvalue $\tau$ and leaves all eigenvectors with eigenvalue $\tau$ invariant. Particularly in our case, the projection operator onto $\mathcal{S}_1$ is given by
\be\label{def-P1} P_1 = \frac{1}{d+1} \sum_{u=0}^d U_{G^{2m_0}}^u \ , \ee
and the dimensionality of $\mathcal{S}_1$ is therefore
\be \label{eq-dimS1} \dim \mathcal{S}_1  =   \Tr(P_1) =  \frac{1}{d+1} \sum_{u=0}^d \Tr \left(U_{G^{2m_0}}^u \right)\ . \ee
To calculate these traces we make use of a result from Theorem 5 in Ref.~\cite{Appleby2009}, applicable to any symplectic $S \in \SL(2,\fd{F}_d)$, namely
\be \Tr (U_S) = l(t-2) \ , \ee
where $t = \Tr(S)$ and $l(x)$ is the Legendre symbol defined in (\ref{def-Legendre}). Since the eigenvalues of $G$ are $\eta^r$ and $\eta^{dr}$, we have
\be 
\begin{split}
\Tr(G^{2m_0u}) & = \eta^{2rm_0 u} + \eta^{2drm_0 u} \\
& = \eta^{2rm_0 u} + \eta^{-2rm_0 u} \\
& = (\eta^{rm_0 u} - \eta^{-rm_0 u})^2 + 2 \ ,
\end{split} 
\ee where in the second step we use the fact that $\Delta^{2m_0} = \eta^{2rm_0}\eta^{2drm_0} = 1$. Note that $\Tr(G^{2m_0u})=2 \ $ if and only if $\eta^{2rm_0u}=1$, or equivalently, $2rm_0u$ is a multiple of $(p-1)(d+1)$. Since $2m_0$ is the order of $\Delta$, which is an element of a group of order $p-1$,  $2m_0$ must be a factor of $(p-1)$. Therefore $ru$ is a multiple of $(d+1)(p-1)/2m_0$, hence a multiple of $d+1$. Taking into account the fact that $r$ is co-prime to $d+1$, we deduce that $u$ must be a multiple of $d+1$ and therefore must be zero, since $0\le u \le d$. Therefore, for $1 \le u \le d$, we have $\Tr(G^{2m_0u})\ne2$, and
\be \Tr \left(U_{G^{2m_0}}^u \right)\ = \Tr \left(U_{G^{2m_0u}} \right) = 
l\left((\eta^{rm_0 u} - \eta^{-rm_0 u})^2\right)  \ , \ee
which equals $1$ if $\eta^{rm_0 u} - \eta^{-rm_0 u} \in \fd{F}_d$, and $-1$ otherwise. To determine this, we notice
\be \begin{split}
(\eta^{rm_0 u} - \eta^{-rm_0 u})^d &= \eta^{drm_0 u} - \eta^{-drm_0 u} \\
&= (-1)^{u+1} (\eta^{rm_0 u} - \eta^{-rm_0 u}) \ , 
\end{split} \ee
where in the last step we make use of the fact that $\eta^{drm_0} =- \eta^{-rm_0}$ since $\Delta^{m_0} = -1$ according to (\ref{eq-deltam0}). Hence, $\eta^{rm_0 u} - \eta^{-rm_0 u} \in \fd{F}_d$ if and only if $u$ is odd, and therefore
\be \Tr \left(U_{G^{2m_0}}^u \right) = (-1)^{u+1} \qquad 1 \le u \le d \ . \ee
For $u=0$, clearly $\Tr (U^u_{G^{2m_0}}) = \Tr (\eye) = d$. We now evaluate (\ref{eq-dimS1}) to conclude the proof:
\be \dim \mathcal{S}_1 = \frac{1}{d+1} \left( d + \sum_{u=1}^d (-1)^{u+1} \right) = 1 \ . \ee
\end{proof}
\begin{lemma}\label{lem-UGeigenvector}
Every MUB-cycling $U_G$ has a non-zero eigenvector $\qs \in \fd{Q}(\omega)^d$ with unit eigenvalue, i.e.
\be U_G \qs = \qs \ . \ee
\end{lemma}
\begin{proof}
From Lemma \ref{lem-dimS1} we know that the unitary $U_{G^{2m_0}}$ has exactly one eigenvalue equal to 1, implying $\det(U_{G^{2m_0}} - \eye) =0$, which in turns means the system of linear equations $(U_{G^{2m_0}} - \eye) X = 0$ has a non-trivial solution. Since the matrix elements of $U_{G^{2m_0}}$ are in the cyclotomic field $\fd{Q}(\omega)$, there exists a non-zero vector $\qsp \in \fd{Q}(\omega)^d$ so that $U_{G^{2m_0}}\qsp = \qsp$. Since
\be U_{G^{2m_0}}U_{G} \qsp = U_G U_{G^{2m_0}} \qsp = U_G \qsp \ee
and since $\mathcal{S}_1$ is one-dimensional, we must have
\be U_G \qsp = \lambda \qsp \ , \qquad \lambda \in \fd{Q}(\omega) \ . \ee
Repeatedly applying $U_G$ to this equation and recalling the fact $U_G^{2m_0} = \eye$, we see that $\lambda$ has to satisfy
\be \lambda g_{\Delta}(\lambda)\dotsc g_{\Delta}^{2m_0-1}(\lambda) = 1 \ .\ee
By Theorem \ref{thm-Hilbert90} (a variant of Hilbert's Theorem 90), there exists $\mu \in \fd{Q}(\omega)$ such that
\be \lambda = \mu /g_{\Delta}(\mu) \ . \ee
If we define $\qs = \mu \qsp$, it immediately follows that $U_G \qs  = \qs$ as desired.
\end{proof}
\begin{lemma}\label{lem-parity} Any eigenvector of a MUB-cycler $U_G$ is also an eigenvector of the parity operator $A$ defined in (\ref{def-parity2}), with eigenvalue $(-1)^{(p-1)/2}$.
\end{lemma}
\begin{proof}
Appealing again to the fact that $\Delta^{m_0}=-1$, and the result from Theorem \ref{thm-GLtypes} that $G^{d+1} = \Delta I$, we have $G^{(d+1)m_0} = P$ (where $P = -I$), and therefore
\be \begin{split} A &= (-1)^{(p-1)/2}U_P \\
&= (-1)^{(p-1)/2} U_{G^{(d+1)m_0}} \\
& = (-1)^{(p-1)/2} (U_G^{2m_0})^{(d+1)/2} 
\end{split} \ee
Let $\qsp$ be an eigenvector of $U_G$, then by Lemma \ref{lem-UG2m0} we have $U_G^{2m_0}\qsp = \qsp$. It follows that $A \qsp = (-1)^{(p-1)/2} \qsp$, which concludes the proof of the lemma.\\
\end{proof}
\begin{remark}
Since every MUB-cycling g-unitary has exactly one eigenvector (up to a scalar multiplication), $\qsp$ is an eigenvector of $U_{G_0^r}$ if and only if it is an eigenvector of $U_{G_0}$, where $G_0$ is defined in (\ref{def-G0}). As every MUB-cycler is conjugate to $G_0^r$ (for some $r$ co-prime to $d+1$), it then follows that the eigenvectors of all MUB-cycling g-unitaries form a single orbit under the extended Clifford group.

\end{remark}

\section{MUB-balanced states}\label{sec:MUBbalanced}

In section \ref{sec:eigenvectors} we showed that when the dimension $d$ is an odd power of an odd prime, every MUB-cycling g-unitary has an eigenvector, which is unique up to a scalar multiplication. Additionally, if $d = 3$ mod $4$, as we will show in this section, these eigenvectors have an extra property: they are MUB-balanced states. The concept of MUB-balanced states was recently introduced by Amburg \etal \cite{Amburg2014}. Rotationally invariant states previously constructed by Sussman and Wootters in even prime power dimensions \cite{Wootters2007, Sussman2007} also have this property. These states all belong to a larger class of quantum states called \MUSs ~\cite{Wootters2007, Sussman2007,Amburg2014,Appleby2014S}. In the case $d = 1$ mod $4$, our numerical calculations in low dimensions (up to $d=31$) show that the eigenvectors of MUB-cyclers are neither MUB-balanced states nor \MUS.

Given a full set of MUBs, a MUB-balanced state is one whose measurement outcome probabilities with respect to every basis are the same up to a permutation. Let $\qs$ be a normalized state, let $\ket{b,v}$ denote the MUB vectors, where $b\in \{0,1,...,d-1,\infty\}$ labels the bases, and $v \in \{0,1,...,d-1\}$ labels the vectors in a basis, and let 
\be p_{b,v} = \abs{\braket{\psi}{b,v}}^2 \ee
be the measurement probabilities. Then $\qs$ is a MUB-balanced state if and only if for each basis $b$, there exists a permutation $\sigma$ such that
\be p_{b,v} = p_{0,\sigma(v)} \ee
for all $v$. It follows from an argument in \cite{Wootters2007} that MUB-balanced states have to be \MUS. For completeness, it is worth providing a sketch of this argument. Let 
\be H_b = -\log_2 \left( \sum_{v} p_{b,v}^2 \right) \ee
be the quadratic R\'enyi entropy in basis $b$. One can show that the total entropy $T = \sum_b H_b$ satisfies the inequality
\be\label{ineq-renyi}
T \ge (d+1) \log_2\left( \frac{d+1}{2}\right)
\ee
which turns out to be saturated if and only if for all $b$
\be \sum_v p_{b,v}^2 = \frac{2}{d+1}.\ee
States that saturate the bound in \cref{ineq-renyi} are called minimum uncertainty states (\MUS). For a MUB-balanced state, $\sum_v p_{b,v}^2$ is independent of $b$. Together with the fact that
\be \sum_{b,v} p_{b,v}^2 = 2, \ee
it clearly follows that a MUB-balanced state is consequently a \MUS. MUB-balancedness is therefore a stricter condition than  minimum uncertainty.

The main result of this section is stated in the following theorem.

\begin{theorem}\label{thm-MUBbalanced}
Let the dimension $d=p^n$ be a prime power satisfying $d = 3 \mod 4$ (the exponent $n$ therefore has to be odd), and let $G \in \GL_p(2,\Fd)$ be a MUB-cycler. Let $\qsp$ be a normalized eigenvector of $U_G^{2m_0}$ with eigenvalue 1 as defined in \cref{thm-eigenvectors}, then $\qsp$ is MUB-balanced.
\end{theorem}
\begin{remark}
The g-unitary $U_G$ plays a crucial part in the following proof. However, for the practical purpose of calculating $\qsp$, it suffices to work with the ordinary unitary $U_G^{2m_0}$. Moreover, since the eigenstates of MUB-cycling g-unitaries form a single orbit of the extended Clifford group, we only need to prove the theorem for the case of MUB-cycling anti-unitaries. We will proceed with the general case of an arbitrary g-unitary nevertheless. This way one will see why it does not work for the case $d = 1 \mod 4$.
\end{remark}
\begin{proof}
Recall from \cref{lem-eta} and \cref{thm-MUBcycler} that $\Delta =\det(G)$ can be written as
\be \Delta = \theta^{r(d-1)/(p-1)}, \ee
where $\theta$ is a primitive element of $\F_d$, and $\gcd(r,d+1)=1$. By the assumption that $n$ is odd, it follows that $r(d-1)/(p-1)$ is odd, so $\Delta$ is a quadratic non-residue. By Lemma 1 of reference \cite{Appleby2009}, when $d = 3\mod 4$ we also have -1 being a quadratic non-residue. Therefore $-\Delta$ is a quadratic residue, i.e. there exists $x \in \F_d$ such that $x^2 = -\Delta$. If we let $F$ be arbitrary
\be F = \bmp \alpha & \beta \\ \gamma & \delta \emp \ee
and define a symplectic
\be S = \bmp x^{-1} & 0 \\ 0 & x \emp, \ee
then it can be verified straightforwardly that
\be SK_{-1}FK_{-1}S^{-1} =  \bmp \alpha & \Delta^{-1}\beta \\ \Delta \gamma & \delta \emp= K_{\Delta}FK_{\Delta}^{-1}.
\ee
Note that $K_{-1}$ is represented by complex conjugation, so in the g-unitary representation we have
\be \label{eq-gUF}
U_S U_F^* U_S^{-1} = g_{\Delta}(U_F).
\ee
From \cref{thm-eigenvectors} we know that there exists $\qs$ such that $U_G\qs = \qs$. We can write
\be \ketbra{\psi} = \lambda P_1 \ee
for some constant $\lambda$, where $P_1$ is the projection operator defined in \cref{def-P1}. Noting that \cref{eq-gUF} can be applied to $P_1$ and letting $v$ be arbitrary, we now can calculate
\be \begin{split}
g_{\Delta}(p_{0,v}) = g_{\Delta}\left(\abs{\braket{0,v}{\psi}}^2\right) 
& = g_{\Delta}\left(\braket{0,v}{\psi}\braket{\psi}{0,v}\right) \\
& = g_{\Delta}\left(\bra{0,v}\lambda P_1\ket{0,v}\right) \\
& = g_{\Delta}(\lambda) g_{\Delta}(\bra{0,v}) U_S P_1^*U_S^{-1} g_{\Delta}(\ket{0,v}) \\
& = \frac{g_{\Delta}(\lambda)}{\lambda^*} g_{\Delta}(\bra{0,v}) U_S \ket{\psi}^*\bra{\psi}^*U_S^{-1} g_{\Delta}(\ket{0,v}) \\
& = \frac{g_{\Delta}(\lambda)}{\lambda^*} \abs{\bra{\psi}^*U_S^{-1} g_{\Delta}(\ket{0,v})}^2 \\
& = \frac{g_{\Delta}(\lambda)}{\lambda^*} \abs{\braket{\psi}{0,xv}}^2 \\
& = \mu p_{0,xv},
\end{split} \ee
where $\mu = g_{\Delta}(\lambda)/\lambda^*$ is a constant, and in the penultimate step we have used \cref{eq-USMUBs}. Repeating this formula $k$ times for an arbitrary integer $k$, we obtain
\be g_{\Delta}^k (p_{0,v}) = \mu_k~ p_{0,x^kv} \ee
where $\mu_k =g_{\Delta}^{k-1}(\mu)g_{\Delta}^{k-2}(\mu)...\mu$ is independent of $v$.

Since $U_G$ is a cycling g-unitary, for every basis $b$ there exists an integer $k$ such that
\be \ket{b,v} \doteq U_G^k \ket{0,\sigma(v)} \ee
for all $v$. In the above expression ``$\doteq$'' means ``equal up to a phase'' and $\sigma$ is a permutation dependent only on $b$. It follows that
\be \begin{split}
\braket{\psi}{b,v} &\doteq \bra{\psi}U_G^k\ket{0,\sigma(v)}\\
& = g_{\Delta}^k\left(\braket{U_G^{-k}\psi}{0,\sigma(v)}\right)\\
& = g_{\Delta}^k\left(\braket{\psi}{0,\sigma(v)}\right).
\end{split} \ee
Consequently,
\be p_{b,v} = g_{\Delta}^k\left(p_{0,\sigma(v)}\right) = \mu_k p_{0,x^k\sigma(v)}. \ee
Since $p_{b,v}$ are probabilities, we must have
\be \sum_v p_{b,v} = \sum_v p_{0,x^k\sigma(v)} = 1. \ee
This implies $\mu_k = 1$, which in turn implies that the normalized state 
\be \qsp = \frac{\qs}{\sqrt{\abs{\braket{\psi}{\psi}}}} \ee
is MUB-balanced.\\
\end{proof}

\begin{remark}
The MUB-balanced states established in this theorem are identical to those constructed by Amburg {\it et al} using a completely different method \cite{Amburg2014}. In their paper, the orbit of states is generated from the state corresponding to the discrete Wigner function
\be W_{\p} = \frac{1}{d(d+1)}\left(1-d\delta_{\p,\mathbf{0}} +\sum_{x\in \fd{F}_d^*}l(x^2+1)\omega^{\tr(xp_1^2 +x p_2^2)}\right). \ee
Let us define a specific element in $\GL(2,\Fd)$
\be G = \bmp \alpha & \beta \\ -\beta & \alpha \emp, \ee
with
\be \alpha = (\eta + \eta^d)/2 \hskip 15mm \beta = i_M(\eta - \eta^d)/2, \ee
where $i_M = \eta^{(p-1)(d+1)/4}$ is a modular analogue of $i$ (notice that $i_M^2=-1$), and $\eta$ is defined in Lemma \ref{lem-eta}. Using basic finite-field facts (see Section \ref{sub:finitefields}), one can check that
\be
\alpha^d =\alpha\hskip 15mm \beta^d = \beta,
\ee
so $\alpha$ and $\beta$ belong to the field $\fd{F}_d$. Thus $G$ is indeed an element of $\GL(2,\Fd)$. Furthermore, $G$ has trace
\be \Tr(G) = \eta + \eta^d \ee
and determinant
\be \Delta = \det (G) = \eta^{d+1}, \ee
and therefore, by \cref{thm-MUBcycler}, it is a MUB-cycler. Since
\be (\alpha p_1 + \beta p_2)^2 + (-\beta p_1 + \alpha p_2)^2 = \Delta(p_1^2 + p_2^2), \ee
it clearly follows that
\be \label{eq-WGp}
W_{G\p} = g_{\Delta}(W_{\p}). 
\ee
Let $\rho$ be the density matrix corresponding to $W_{\p}$ given by:
\be \rho = \sum_{\p} W_{\p} A_{\p}, \ee
where the phase point operators $A_{\p}$ are defined in \cref{def-Ap}. It follows from \cref{eq-WGp} and \cref{eq-UGAp} that
\be U_G \rho U_G^{-1} =  \rho. \ee
In view of a result from reference \cite{Amburg2014} that the state corresponding to $W_{\p}$ is a pure state, one can now conclude that this state is an eigenvector of the MUB-cycling g-unitary $U_G$. The Wigner functions of these MUB-balanced states for dimension 7 and 11 are visualized in \cref{fig-balancedstates}.
\begin{figure}[h]
\centering
\includegraphics[scale=1]{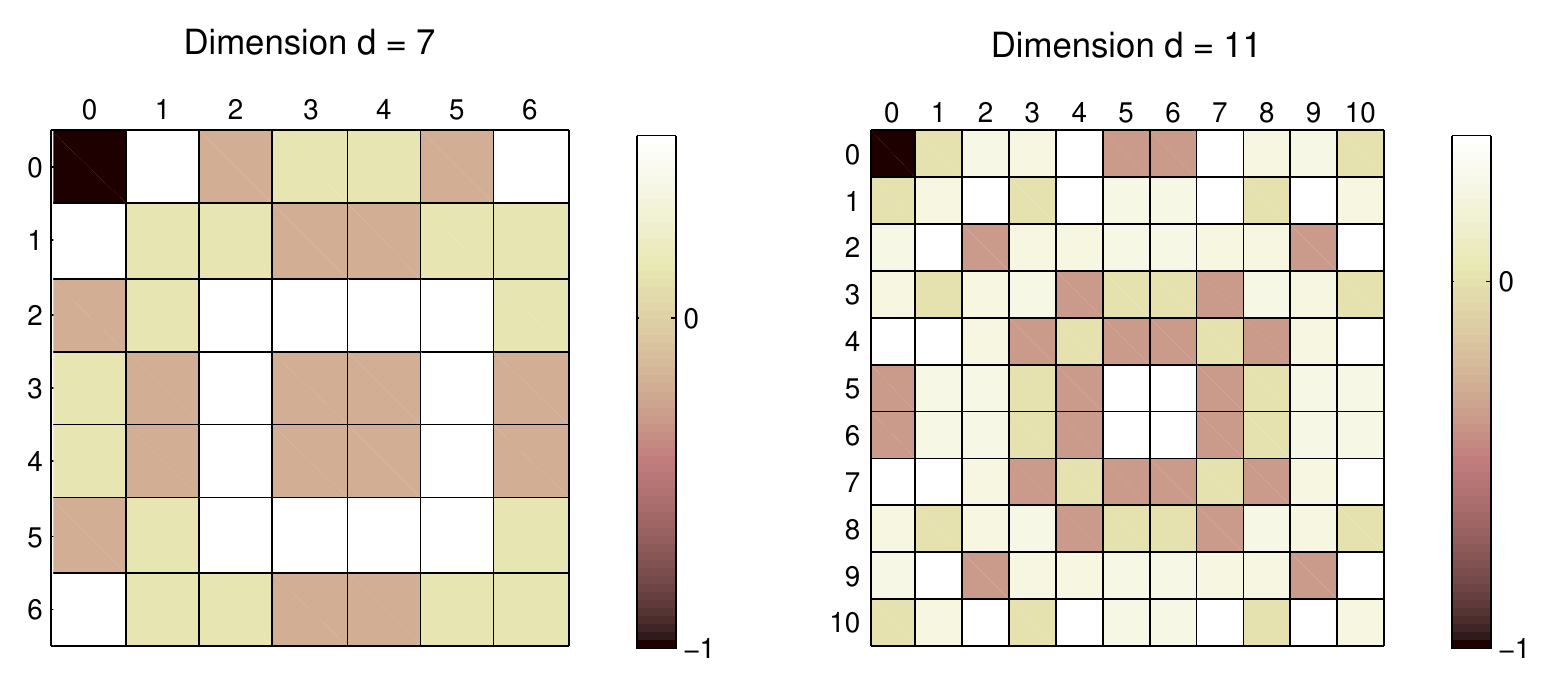}
\caption[Wigner function of MUB-balanced states in $d=7$ and 11]{Color plots of the scaled Wigner function ($dW_{\p}$) corresponding to $\rho$ in dimension $d=7$ and $11$. The axes of the discrete phase space are labeled by the two components of $\p=(p_1,p_2)$. The rotational invariance property is manifested via concentric ``circles'' on the discrete plane.}
\label{fig-balancedstates}
\end{figure}

\end{remark}

One might want to ask how many MUB-balanced states arising from MUB-cycling g-unitaries there are. Let us consider the MUB-cycler $U_{G_0}$ where $G_0$ is the matrix with multiplicative order $2m_0 = p-1$ defined in \cref{def-G0}. Let $\qs$ be the corresponding MUB-balanced state. By \cref{thm-eigenvectors}, we know $\qs$ is the unique eigenvector of the unitary $U_{G_0}^{p-1}$ with eigenvalue 1. Therefore it will be left invariant by any element $V$ of the extended Clifford group satisfying
\be V U_{G_0}^{p-1} V^{-1} = U_{G_0}^{(p-1)k} \ee
for some integer $k$. It turns out that the only possible values for $k$ are $\pm 1$ and
\be
V = U_{G_0^{m(p-1)/2}}
\ee
when $k=1$, or 
\be
V = U_{G_0^{m(p-1)/2}F}
\ee
when $k=-1$, where 
\be 
F = \bmp 0 & i_M\eta^{(d+1)/2} \\ i_M\eta^{-(d+1)/2} & 0 \emp
\ee
is a symplectic matrix satisfying $FG_0^{p-1}F^{-1} = G_0^{-(p-1)}$ and $m$ can take any integer value from 0 to $2d+1$. Hence, there are $4d+4$ possibilities for $V$.
\begin{conjecture}\label{conj-MUBbalanced}
We conjecture that there are no other unitaries or anti-unitaries in the extended Clifford group that leave $\qs$ invariant. As the order of the extended Clifford group is $2d^3(d^2-1)$ \cite{Gehles2002}, the number of MUB-balanced states would therefore be $d^3(d-1)/2$.
\end{conjecture}
\begin{table}[h]
\centering
\begin{tabular}{cccc}
\rowcolor{gray!25} Dimension &  ord($\mathcal{E}_d$) & MUB-cyclers & MUB-balanced\\ 
$d=7$ &  32,928 & 504 & 1,029\\ 
$d=11$ &  319,440 & 2,200 & 6,655\\ 
$d=19$ &  4,938,480 & 24,624 & 61,731\\ 
\end{tabular}
\vskip 4mm
\parbox{10.5cm}{
\caption[Number of MUB-balanced states in $d=7,11,$ and 19]{The order of the extended Clifford group, the number of MUB-cyclers and the number of distinct MUB-balanced states found in dimensions 7, 11 and 19.}
\label{tab-MUBbalanced}}
\end{table}

We have checked the conjecture in detail for dimensions $d = 7,11,19$. The numbers reported in \cref{tab-MUBbalanced} agree with our prediction. Originally, we expected that our construction technique would yield many new MUB-balanced states in addition to those found by Amburg {\it et al}. Our results, however, suggest that Amburg {\it et al} have indeed constructed the entire set. If this is true, it means that MUB-balanced states form a highly distinguished geometrical structure. They seem to be even rarer than SICs: in most of the dimensions that have been analyzed, there exist more than one orbits of SICs, whereas MUB-balanced states seem to come only in a single orbit. They ``have no right to exist'', as Amburg {\it et al} has put it. Of course this should not be taken literally, as we know they do exist after all. Our work provides one way of explaining their right of existence, by unveiling a new underlying g-unitary symmetry.

\chapter{Summary and Outlook}\label{chap:summary}

\section{Summary of main results}
We chose to study SICs as we believed such a symmetric structure in the Hilbert space could reveal deep insights into quantum theory. One of our results reveals the geometric significance of SICs on the cone of non-negative operators: SICs are the closest to being orthogonal bases. Explicitly, they form the only sets of $d^2$ normalized positive semidefinite operators that minimize a class orthogonality measures $K_t$ defined to be
\be \label{def-Kt}
K_t = \sum_{i\ne j} \left(\Tr (A_i A_j)\right)^t \hskip 10mm t\in\fd{R},~t > 1.
\ee

Studying SICs naturally led us to a study of the Weyl-Heisenberg symmetry - the group symmetry that finite dimensional quantum mechanics was built upon in the very early years. While studying orbits of quantum states under the WH group, we observed that under certain conditions, among $d^2$ vectors in the orbits one can find sets of $d$ vectors that are linearly dependent. We proved in \cref{thm-linde} that if the initial vector belongs to certain eigenspaces (depending on the dimension $d$) of the Zauner unitary, and if a set of $d$ vectors consists of certain combinations of triplets and singlets, linear dependency will occur. In dimensions $d=3$ this explains the 3 linear dependencies that arise from arbitrary SIC fiducial vectors in the known SIC family. Interestingly, there are special SICs in this family that give rise to 9 linear dependencies. This fact is connected to the Hesse configuration in the theory of elliptic curves. We performed an exhaustive numerical search for linear dependencies in dimensions $d=4$ to 9. The only other case where we observed extra linear dependencies from a SIC fiducial (compared to an arbitrary initial vector from the same eigenspace of $\UZ$) is when the SIC fiducial lies in $\mcH_{\eta^2}$ in dimension $d=8$. In dimension $d=6$ and $d=9$ we analyzed in detail the relations among normal vectors of hyperplanes spanned by linearly dependent sets. Besides some orthogonality relations, we observed that some of these normal vectors always formed 2-dimensional and 3-dimensional SICs, even though the initial vector is not a SIC fiducial. We were able to provide an analytical explanation for the observed 2-dimensional small SICs in $d=6$.

SICs are also the motivation for the construction of g-unitary operators. In \cref{chap:gunitary} we describe a toy model for SICs, in which g-unitaries were constructed from the simple Galois groups of cyclotomic field extensions, and we studied their actions on MUBs. We have gone through a large amount of technical details, which are summarized in \cref{thm-faithfulness,thm-GLtypes,thm-MUBcycler,thm-eigenvectors,thm-MUBbalanced}. However, the picture can be better summarized in words.

G-unitaries are constructed to generalize the notion of anti-unitaries. However, their action is restricted only to vectors whose components belong to some special number field, which is the cyclotomic field in our case. In this case, g-unitaries play role in the description of MUBs in odd prime power dimensions. Their actions on the MUB vectors can be interpreted as rotations in Bloch space, just as any ordinary unitary operator.

We studied projective transformations that permute mutually unbiased bases. In odd prime power dimension $d=p^n$ where $n$ is odd, we showed that there are transformations that cycle through the full set of $d+1$ bases and that can be realized by g-unitaries. If $d=3$ mod 4, these transformations can be effected by anti-unitaries, but if $d=1$ mod 4, we need to appeal to g-unitaries. We studied the eigenvectors of these MUB-cycling g-unitaries and showed every MUB-cycling g-unitary always had a unique eigenvector (up to a scalar multiplication). Furthermore, we can always choose a scalar so that this eigenvector is invariant under the g-unitary, in which case we proved that it is also invariant under the parity operator.

When $d=3$ mod 4, we proved that the invariant eigenvectors of MUB-cycling g-unitaries are MUB-balanced states. Our construction can be considered as a supplement to work done by Amburg \etal \cite{Amburg2014} in two ways. First, while the construction of Amburg \etal was done using Wigner functions, our construction was done directly in the Hilbert space. Secondly, our technique provides a connection to the original construction in even prime power by Wootters and Sussman \cite{Wootters2007} in which MUB-balanced states were found to be eigenvectors of MUB-cycling unitaries. In our case, because of the lack of cycling unitaries in odd prime power dimensions, we instead used cycling g-unitaries. Originally, we expected that our construction would yield more MUB-balanced states, since we were able to expose the underlying symmetry and found all such states with this symmetry. However, the results suggest that these states all lie on a single orbit the the extended Clifford group, which is remarkable since even rare quantum states like SICs come in more than one orbits in many dimensions that have been analyzed.

G-unitary symmetry for the simple case of cyclotomic field extensions has provided us a thorough understanding of the ``right to exist'' of a distinguished class of quantum states, namely MUB-balanced states. Extending our analysis to g-unitaries that are applicable to SIC states or some other special quantum states is a non-trivial task. However, we hope that our results can be considered as a first useful step in the direction of solving the SIC problem, as we believe that symmetries play very important roles in solving hard problems in physics.

\section{List of open problems}

The following list brings up research questions or problems we have not been able to answer that are worth further investigation.

\begin{enumerate}
\item As discussed in \cref{sub:numerical}, there are more linearly dependent sets observed numerically than what can be accounted for by \cref{thm-linde}. An explanation for this is still left open.
\item Small 3-dimensional SICs arising from normal vectors of the hyperplanes spanned by linearly dependent sets in $d=9$ are observed, but have not yet been fully understood.
\item An exhaustive search for small SICs in $d=12$ has not been done. It is also an open question if small SICs can be found in other dimensions rather than $d=6$ and 9.
\item In the study of g-unitaries, we only found eigenvectors for a special class of them, namely those that have the MUB-cycling property. It is an open question whether other g-unitaries also have eigenvectors, and how to find them.
\item We have counted the number of MUB-balanced states for small dimensions in \cref{tab-MUBbalanced}. It remains as a conjecture that the number of MUB-balanced states in every odd prime power dimensions equal to 3 mod 4 is $d^3(d-1)/2$.

\end{enumerate}

\appendix
\chapter{Appendices}\label{chap:app}

\section{Field theory}\label{sec:fieldtheory}

The theory of fields is a major branch in mathematics that has been studied extensively. In the scope of the thesis we use a number of elementary facts in field theory. We devote this section as an introduction to fields, field extensions, Galois automorphisms, and finite fields, for readers with little or no background in field theory. Besides finite fields, which play an important role in the theory of Mutually Unbiased Bases, we also discuss infinite fields, especially cyclotomic fields, which are crucial for our construction of g-unitaries in \cref{chap:gunitary}. With the purpose of helping the readers quickly grasp the relevant key concepts, we avoid unnecessarily technical definitions and derivations as much as we can. Rigorous treatments of the subjects can be found in textbooks on fields and Galois theory \cite{Stewart1972, Roman2006, Milne2003, Lidl1997}.\\

\begin{definition} A field $\fd{F}$ is a set together with two operations addition and multiplication (denoted by $+$ and $\cdot$) such that for all $a,b,c \in \fd{F}$ the following axioms hold.
\begin{enumerate}
\item Closure:  $a+b$ and $a \cdot b$ are in $\fd{F}$.
\item Associativity: $a+(b+c) = (a+b)+c$ and $a \cdot (b \cdot c) = (a \cdot b) \cdot c$.
\item Commutativity: $a+b = b+a$ and $a\cdot b = b \cdot a$.
\item Identity elements: there exists an additive identity $0$ such that $a + 0 = a$ for all $a \in \fd{F}$, and there exists a multiplicative identity $1$ such that $a \cdot 1 = a$ for all $a \in \fd{F}$. 
\item Inverses: for every $a \in \fd{F}$ there exists an element $-a$ such that $a+ (-a) = 0$, and for every non-zero $a \in \fd{F}$ there exists an element $a^{-1}$ such that $a \cdot a^{-1} = 1$. These imply the existence of subtraction and division operations.
\item Distributivity: $a \cdot (b+c) = (a \cdot b) + (a \cdot c)$.
\end{enumerate}
\end{definition}
\begin{remark}
A field $\fd{F}$ consists of two abelian groups ($\fd{F}$ under addition, and $\fd{F}\setminus\{0\}$ under multiplication), whose operations are compatible in the sense of the distributivity law.
\end{remark}

\begin{example} Common examples include the field of all rational numbers $\fd{Q}$ and field of all real numbers $\fd{R}$ under ordinary addition and multiplication. There also exist finite fields, i.e. fields with a finite number of elements. For example, if $p$ is a prime number, then the set of integers $\fd{Z}_p=\{0,1,..,p-1\}$ with addition and multiplication modulo $p$ forms a field, called a prime field. To see why multiplicative inverses exist in $\fd{Z}_p$, take any non-zero element $x \in \fd{Z}_p$ and consider $p-1$ products $xk$ where $k=1,2,...,p-1$. As $p$ is prime we cannot have $x(k_1-k_2) = 0$ (mod $p$) for $k_1 \ne k_2$ (mod $p$). Therefore these products are distinct and they must take all $p-1$ non-zero values in $\fd{Z}_p$ including 1.
\end{example}
\subsection{Field extensions}

\begin{example} The concept of field extensions is best illustrated by the example of the construction of $\fd{C}$, the complex field that quantum physicists are very familiar with. We start with the real field $\fd{R}$ and an observation that the polynomial  $x^2 + 1$ is irreducible over $\fd{R}$. If we define an imaginary number $i$ by the property $i^2 + 1 = 0$, then it does not belong to $\fd{R}$, and we can extend $\fd{R}$ to include $i$ by defining the complex field $\fd{C}$ as the set of all numbers of the form:
\begin{equation}\label{def-C} \fd{C} \equiv \{a + i b: a,b \in \fd{R}\}.\end{equation}
The addition rule in the new field $\fd{C}$ can be straightforwardly defined as
\begin{equation} (a_1+ib_1) + (a_2+ib_2) = (a_1+a_2) + i(b_1+b_2).\end{equation}
For the multiplication rule we use the defining property of $i$, namely $i^2 = -1$, to derive
\begin{equation}
\begin{split}
(a_1+ib_1)(a_2+ib_2) &= a_1 a_2 + i(a_1 b_2 + a_2 b_1) + i^2 b_1 b_2 \\
&= (a_1 a_2 -b_1 b_2) + i(a_1 b_2 + a_2 b_1).
\end{split}
\end{equation}
One can think of a complex number $(a+ib)$ as a 2-component vector $(a,b)$ in a real vector space. The dimensionality of this vector space is equal to the degree of the polynomial defining $i$, namely 2. The construction of $\fd{C}$ in (\ref{def-C}) can be generalized as follows.
\end{example}

Given a field $\fd{F}$ and a number $h \notin \fd{F}$, let $\fd{E}$ be the smallest field containing both $h$ and $\fd{F}$, denoted by $\fd{E} \equiv \fd{F}(h)$. $\fd{F}$ is called the ground field, $\fd{E}$ is called the extended field or the extension field, the field extension (not to be confused with the extension field) is denoted by $\fd{E}/\fd{F}$ (reads as $\fd{E}$ over $\fd{F}$), and $h$ is called a field generator. Field generators need not be unique. For example, for the field extension $\fd{C}/\fd{R}$ instead of $i$ we could also use $-i$ as a generator.

\begin{definition} Assume that $h$ is algebraic over $\fd{F}$, meaning that it is a root of a polynomial with coefficients in $\fd{F}$. Among all such polynomials that admit $h$ as a root, let
\begin{equation}\label{def-minimalpolynomial}
P(x) = x^n + c_{n-1}x^{n-1} + ... + c_1 x + c_0 \qquad (c_i \in \fd{F})
\end{equation}
be the one with the lowest degree having the leading coefficient equal to one ($P(x)$ is called the minimal polynomial of $h$ over $\fd{F}$). The extended field $\fd{E}$ can then be constructed as
\begin{equation}\label{def-extensionfield}
\fd{E} \equiv \fd{F}(h) = \{f_0 + f_1 h + ...+f_{n-1} h^{n-1}: f_i \in \fd{F}\} \ .
\end{equation}
\end{definition}

\begin{remark}
We do not  include higher powers of $h$ in (\ref{def-extensionfield}) because they can be reduced to powers smaller than $n$ by the property of $h$ being a root of a polynomial of degree $n$. One can see that $\fd{E}$ is closed under addition, subtraction, and multiplication, and it can be shown to be also closed under division \cite{Roman2006}. $\fd{E}$ can be regarded as an $n$-dimensional vector space over $\fd{F}$, and $n$ is called the degree of the field extension $\fd{E}/\fd{F}$.
\end{remark}

\subsection{Galois automorphisms}\label{sub:Galois}

\begin{definition} Given a field extension $\fd{E}/\fd{F}$ ($\fd{E}$ is an extension of $\fd{F}$), a Galois automorphism of the extension $\fd{E}/\fd{F}$ is defined as an automorphism of $\fd{E}$ that fixes elements in $\fd{F}$. In other words, it is a bijective mapping $g: \fd{E} \to \fd{E}$ that has the following properties.
\begin{enumerate}
\item $g(e_1 + e_2) = g(e_1) + g(e_2)$ for all  $e_1,e_2 \in \fd{E}$.
\item $g(e_1  e_2) = g(e_1)  g(e_2)$ for all  $e_1,e_2 \in \fd{E}$.
\item $g(f) = f$ for all  $f \in \fd{F}$.
\end{enumerate}
\end{definition}
\begin{remark} It follows from the definition that if a Galois automorphism $g$ takes an element $x \in \fd{E}$ to the ground field $\fd{F}$, then $x$ has to belong to the ground field itself, or else $g$ fails to be a bijective mapping.
\end{remark}

\begin{definition} The Galois automorphisms form a group called the Galois group of the extension $\fd{E}/\fd{F}$, denoted by $\Gal(\fd{E}/\fd{F})$.
\end{definition}
One property of the Galois group is that its order is less than or equal to the degree of the field extension. To see why we first note that if $\fd{E} = \fd{F}(h)$ and $g \in \Gal(\fd{E}/\fd{F})$, then the value of $g(h)$ determines the value of $g(e)$ for every element $e \in \fd{E}$, as can be clearly seen from the defining properties of Galois automorphisms and from the construction of $\fd{E}$ in (\ref{def-extensionfield}). In other words, a Galois automorphism $g$ is completely specified by its action on a generator. Secondly, we notice that when we apply $g$ to both sides of the equation $P(h) = 0$, where $P(x)$ is the minimal polynomial of $h$ with the form given in (\ref{def-minimalpolynomial}), $g$ does not act on the polynomial coefficients because they are in the ground field $\fd{F}$, therefore
\begin{equation}
(g(h))^n + a_{n-1}(g(h))^{n-1} + ... + a_1 g(h) + a_0 = 0,
\end{equation}
which implies that $g(h)$ is also a root of $P(x)$. Since $P(x)$ has degree $n$, the number of values of $g(h)$ is at most $n$, and hence, so is the number of Galois automorphisms. When the Galois group has the same order as the degree of the extension, the extension is called a Galois extension. This has an important mathematical implication, namely the fundamental theorem of Galois theory. Although this theorem is not invoked in the thesis, we do want to note that all the field extensions used in the thesis are in fact Galois extensions.

\begin{example} Let us go back to the example of the extension of the real field $\fd{R}$ to the complex field $\fd{C}$. If $g: \fd{C} \to \fd{C}$ is a Galois automorphism of the extension $\fd{C}/\fd{R}$, then it must satisfy
\begin{equation}
g(i)g(i) = g(i^2) = g(-1) = -1 \ ,
\end{equation}
which implies either $g(i)=i$ or $g(i)=-i$. If $g(i) = i$, then for any $a,b \in \fd{R}$ we have $g(a+ib) = g(a) +g(i)g(b) = a + ib$, meaning that $g$ is the identity mapping. If $g(i) = -i$, then $g(a+ib) = a-ib$, so $g$ is complex conjugation. The Galois group for the extension $\fd{C}/\fd{R}$ therefore consists of only two elements: the identity mapping and complex conjugation. It is a Galois extension because the group has order 2, which is same as the degree of the extension.
\end{example}
When viewed as functions from $\fd{E}$ to $\fd{F}$, Galois automorphisms are linearly independent, as shown in the following theorem.
\begin{theorem}[Dedekind]\label{thm-Dedekind} Given an extension $\fd{E}/\fd{F}$, let $\{g_i\}_{i=1}^n$ be its Galois group then $g_i$'s are linearly independent functions from $\fd{E}$ to $\fd{F}$, meaning that for $a_1, a_2,...,a_n \in \fd{E}$,
\begin{equation}
a_1 g_1(x) + a_2 g_2(x) + \dotsc + a_n g_n(x) = 0 \quad \text{ for all } x \in \fd{E}
\end{equation}
if and only if each and every $a_i=0$. It then follows that any non-zero linear combination of Galois automorphisms is a non-zero mapping.\\
\begin{remark}
Dedekind's theorem in fact holds for a larger class of functions called characters. Here we restrict ourselves to Galois automorphisms, but the proof is the identical.
\end{remark}
\begin{proof}
Suppose we have a zero function
\begin{equation}\label{eq-dedekind}
a_1 g_1 + a_2 g_2 + \dotsc + a_k g_k= 0 \ 
\end{equation}
for some $a_1, a_2,...,a_k \in \fd{E}$. We will prove by induction that all the coefficients $a_i$ must be zero. As none of the Galois automorphisms is zero, the statement is clearly true for $k=1$. For $k > 1$, because $g_1 \ne g_k$, we can find an element $x_0 \in \fd{E}$ such that $g_1(x_0) \ne g_k(x_0)$. Multiplying equation (\ref{eq-dedekind}) by $g_k(x_0)$ we get
\begin{equation}\label{eq-dedekind2}
a_1 g_k(x_0) g_1 + \dotsc + a_k g_k(x_0) g_k = 0\ .
\end{equation}
Evaluate (\ref{eq-dedekind}) at $x_0x$ we get
\begin{equation}
a_1 g_1(x_0) g_1(x) + \dotsc + a_k g_k(x_0) g_k(x) = 0 \ ,
\end{equation}
which, as $x$ can take any value in $\fd{E}$, implies
\begin{equation}
a_1 g_1(x_0) g_1 + \dotsc + a_k g_k(x_0) g_k = 0 \ .
\end{equation}
Subtracting (\ref{eq-dedekind2}) from the above equation we obtain
\begin{equation}
a_1 \left[g_1(x_0) -g_k(x_0)\right]g_1 + \dotsc + a_{k-1} \left[g_{k-1}(x_0) -g_k(x_0)\right] g_{k-1}= 0 \ ,
\end{equation}
which, by the induction hypothesis, implies $a_1 = 0$. Removing the $a_1$ term in (\ref{eq-dedekind}) and appealing to the induction hypothesis again, we then deduce $a_2 = \dotsc = a_k = 0$.
\end{proof}

\end{theorem}

\subsection{Cyclotomic fields}

\begin{definition} A cyclotomic field \gls{Qw} is an extension field generated from the rational field $\fd{Q}$ and a primitive $N$-th root of unity $\omega=e^{2 \pi i/N}$. Although cyclotomic fields can be defined for any $N$, we will restrict ourselves to the case when $N=p$ is a prime number. In such a case, the minimal polynomial of $\omega$ over $\fd{Q}$ is
\begin{equation}
P(x) = 1 + x + x^2+ ... + x^{p-1} \ ,
\end{equation}
and the elements of the cyclotomic field $\fd{Q}(\omega)$ are of the form
\begin{equation}
\fd{Q}(\omega) = \{q_0 + q_1 \omega + ... + q_{p-2}  \omega^{p-2}: q_i \in \fd{Q}\} \ .
\end{equation}
\end{definition}

The extension $\fd{Q}(\omega)/\fd{Q}$ is of degree $p-1$. Let $g: \fd{Q}(\omega) \to \fd{Q}(\omega)$ be a Galois automorphism of the field extension. We then have the identity
\begin{equation}
(g(\omega))^p = g(\omega^p) = g(1) = 1 \ ,
\end{equation}
which implies that $g(\omega)= \omega^k$ for some integer $k$ in the range $1\le k \le p-1$ (the value of $g(\omega)$ cannot be 1 because $\omega$ does not belong to the ground field). If we specifically denote the Galois automorphism that maps $\ \omega \mapsto \omega^k$ by \gls{gk}, then $\{g_k\}_{k=1}^{p-1}$ forms the Galois group of order $p-1$ of this field extension. Note that $g_1$ is the identity mapping, and $g_{p-1}$ is complex conjugation. A general element $g_k$ can be thought of as a generalization of complex conjugation, and it is often called a Galois conjugation. Galois conjugations are the building blocks for our later construction of g-unitaries. For now we want to mention one more property of theirs, as stated in the following theorem.\\
\begin{theorem}[a variant of Hilbert's Theorem 90 \cite{Hilbert1998}]\label{thm-Hilbert90} Let $g \in \Gal(\fd{Q}(\omega)/\fd{Q})$ be a Galois automorphism of order $m$, meaning that $m$ is the smallest positive integer for which $g^m$ is the identity mapping. If $\lambda \in \fd{Q}(\omega)$ satisfies
\be \lambda g(\lambda) ... g^{m-1}(\lambda) = 1 \ , \ee 
then there exists $\mu \in \fd{Q}(\omega)$ such that 
\be \label{eq-Hilbert90} \lambda = \mu / g(\mu) \ . \ee
\end{theorem}
\begin{proof}
Let us consider a mapping $T:\fd{Q}(\omega) \to \fd{Q}(\omega)$ defined as
\be
T(x) = x + a_1 g(x) + a_2 g^2(x) + ...  + a_{m-1}g^{m-1}(x) \ ,
\ee
where $a_i = \lambda g(\lambda) ... g^{i-1}(\lambda)$. Note that 
\be 
\lambda g(a_i) = a_{i+1} \qquad \text {for }i = 1,...,m-2 
\ee
and
\be
\lambda g(a_{m-1}) = \lambda g(\lambda)...g^{m-1}(\lambda) = 1 \ , \ee
therefore we have
\be \lambda g(T(x)) = T(x) \ . \ee
From Theorem \ref{thm-Dedekind} we know that $T(x)$ is not a zero mapping, so there exists $x_0 \in \fd{Q}(\omega)$ such that $T(x_0) \ne 0$. If we define $\mu \equiv T(x_0)$, then (\ref{eq-Hilbert90}) is immediately obtained.
\end{proof}

\subsection{Finite fields}\label{sub:finitefields}

We have mostly dealt with fields that are infinite in size so far. Finite fields also come into play in the thesis, especially in the theory of Mutually Unbiased Bases (MUB) and in technical manipulations for MUB cyclers and their eigenvectors.

\begin{definition} A finite field (somewhat confusingly also called a Galois field) is a field that has a finite number of elements, called its order. 
\end{definition}
The prime field $\fd{Z}_p$ for any prime number $p$ is an example of a finite field, as previously mentioned. There are also finite fields of other orders. However, it is well known (dating back to Galois) that finite fields only exist for which the order is a prime power $p^n$, and that for every prime power there exists a unique (up to an isomorphism) field of this order. Thus, we can refer to a finite field only by its order, and we shall denote the finite field of order $d$ by \gls{Fd}, where $d$ must be a prime power for $\fd{F}_d$ to exist. We will not provide the proof here, but will instead give a concrete example of how to generate a larger finite extension field from a prime field $\fd{F}_p$.

\begin{example} Consider the finite field $\fd{F}_2 = \{0,1\}$ of order 2, and let $P(x) = x^2 + x + 1$. We see that $P(0) = P(1) = 1$, so $P(x)$ does not have a root in $\fd{F}_2$. We then define $\lambda$ to be a root of $P(x)$ and define
\begin{equation}
\fd{F}_2(\lambda)= \{a + \lambda b: a,b \in \fd{F}_2\}.
\end{equation}

One can see that $\fd{F}_4 \equiv \fd{F}_2(\lambda) = \{0,1,\lambda, \lambda + 1 \}$ has 4 elements and its addition and multiplication tables can be calculated using the identity $\lambda^2 + \lambda + 1 = 0$ as shown in \cref{tab-F4}. In general, if we start from an irreducible polynomial of degree $r$ in $\fd{F}_p$, we will be able to extend the field to $\fd{F}_{p^r}$.\\

\begin{table}[h]
\centering
\begin{minipage}[b]{0.45\linewidth}\centering
\begin{tabular}{M{0.85cm} M{0.85cm} M{0.85cm} M{0.85cm} M{0.85cm}}
\cellcolor{gray!65}\boldmath$\textcolor{white}{+}$ & \cellcolor{gray!25}$0$ & \cellcolor{gray!25}$1$ & \cellcolor{gray!25}$\lambda$ & \cellcolor{gray!25}$\lambda +1$ \\ 
\cellcolor{gray!25}$0$ & $0$ & $1$ & $\lambda$ & $\lambda +1$ \\ 
\cellcolor{gray!25}$1$ & $1$ & $0$ & $\lambda+1$ & $\lambda$ \\ 
\cellcolor{gray!25}$\lambda$ & $\lambda$ & $\lambda + 1$ & $0$ & $1$ \\ 
\cellcolor{gray!25}$\lambda+1$ & $\lambda + 1$ & $\lambda$ & $1$ & $0$ \\ 
\end{tabular}
\end{minipage}
\hspace{0.45cm}
\begin{minipage}[b]{0.45\linewidth}\centering
\begin{tabular}{M{0.85cm} M{0.85cm} M{0.85cm} M{0.85cm} M{0.85cm}}
\cellcolor{gray!65}\boldmath$\textcolor{white}{\cdot}$ &  \cellcolor{gray!25}$0$ &  \cellcolor{gray!25}$1$ &  \cellcolor{gray!25}$\lambda$ &  \cellcolor{gray!25}$\lambda +1$ \\
\cellcolor{gray!25}$0$ & $0$ & $0$ & $0$ & $0$ \\
\cellcolor{gray!25}$1$ & $0$ & $1$ & $\lambda$ & $\lambda+1$ \\
\cellcolor{gray!25}$\lambda$ & $0$ & $\lambda$ & $\lambda+1$ & $1$ \\
\cellcolor{gray!25}$\lambda+1$ & $0$ & $\lambda+1$ & $1$ & $\lambda$ \\
\end{tabular}
\end{minipage}
\vskip 2mm
\caption[Addition and multiplication tables for $\fd{F}_4$]{Addition and multiplication tables for the finite field $\fd{F}_4$.}
\label{tab-F4}
\end{table}

\end{example}

We would like to mention a few basic properties of finite fields \cite{Lidl1997} that are frequently used in Section \ref{sec:MUBcyclers} and Section \ref{sec:eigenvectors}. First of all, every finite field admits a primitive element $\theta$ (there could be more than one) such that every non-zero element $\lambda$ in the field can be written as $\lambda = \theta^k$ for some non-negative integer $k$ smaller than the order of the field. Secondly, for $\fd{F}_p$ and its extension field $\fd{F}_{d}$, where $d = p^n$ is a prime power, the following statements hold:
\begin{enumerate}
\item $a^d=a  \hskip 7mm \forall a \in \fd{F}_d$.
\item $(a+b)^p = a^p + b^p  \hskip 7mm \forall a,b \in \fd{F}_{d}$.
\item $\forall a \in \fd{F}_{d}$, $a^p=a$ if and only if $a \in \fd{F}_p$.
\end{enumerate}

\subsection{Field trace}

\begin{definition} Given a Galois extension $\fd{E}/\fd{F}$, the field trace of an element $e\in \fd{E}$ denoted by \gls{trace}$(e)$ (we use the lower case to denote the field trace to distinguish it from the matrix trace) is defined as the sum of all Galois conjugates of $e$:
\begin{equation}\label{def-tr}
\tr(e) = \sum_{g \in \Gal{(\fd{E}/ \fd{F})}} g(e).
\end{equation}
\end{definition}
One notices that $\tr(e)$ is left invariant by every Galois automorphism in the Galois group, which implies that $\tr(e)$ belongs to the ground field $\fd{F}$. Thus, field trace is a mapping from $\fd{E}$ to $\fd{F}$. Its following properties can be straightforwardly verified.
\begin{enumerate}
\item $\tr(e_1 + e_2) = \tr(e_1) + \tr(e_2)$ for all $e_1,e_2 \in \fd{E}$.
\item $\tr(f e) = f \tr(e)$ for all $e \in \fd{E}$ and $f \in \fd{F}$.\\
\end{enumerate}

\begin{example} For the $\fd{C}/\fd{R}$ extension, the Galois group only has 2 elements: the identity and complex conjugation. Therefore taking the field trace of a complex number,
\begin{equation}
\tr(c) = c + c^*,
\end{equation}
is just the same as taking its real part (modulo a factor of 2).
\end{example}

For the extension $\fd{F}_d / \fd{F}_p$, where $d = p^n$ is a prime power (this will be the context in which we use field trace), there is a concrete formula for the field trace given by \cite{Lidl1997}
\begin{equation} \label{def-tr-primepower}
\tr(x) = x + x^{p} + x^{p^2} + ... + x^{p^{n-1}} \quad \forall x \in \fd{F}_d \ .
\end{equation}

\section{Finite-field construction of Clifford unitaries}\label{sec:Clifford-app}
The Clifford group is introduced in \cref{sub:Clifford}, in which we describe a construction of Clifford unitaries in a $d$-dimensional Hilbert space for a general $d$. We refer to those as the ordinary Clifford group. When $d$ is a prime or a prime power, a finite field of order $d$ exists (see \cref{sub:finitefields}), providing us an aid to express the Weyl-Heisenberg group and Clifford unitaries in a slightly different way that is more convenient for certain purposes, such as in the constructions of MUBs (in \cref{sec:MUBs}) and Galois-unitaries (in \cref{sec:gunitaries}). These variants of the WH and Clifford groups are referred to as Galoisian variants. Galoisian Clifford groups in turn come in two versions: a full and a restricted one, where the latter is a subgroup of the former \cite{Appleby2009}. Here we are only interested in the restricted version. To better illustrate how finite fields come into play, we will start with the simpler case of prime dimensions before going into the generalized case of prime power dimensions.\\

\begin{note}This section is a review of some results already fully described by Appleby in \cite{Appleby2009}. There are other constructions for unitary representations of $\SL(2,\fd{F}_d)$ (for example by Chau \cite{Chau2005}). We choose to use the one in \cite{Appleby2009} only for convenience.
\end{note}
\subsection{In odd prime dimensions $d=p$}\label{sub:Clifford-prime}
We start with the case of the Hilbert space's dimension $d=p$ being an odd prime. Let $\omega$ be a $p$-th primitive root of unity
\be \omega = e^{2\pi i /p}, \ee
and let $\{\ket{x}:x =0,1,...,p-1\}$ be the standard basis, where the labels on the states are elements of the finite field $\fd{Z}_p$, i.e. they are integers modulo $p$. 
\begin{definition}We define the shift operator \gls{X} and the phase operator \gls{Z} by their action on the basis states in the usual way:
\be X\ket{x} = \ket{x+1} \hskip 15mm Z\ket{x} = \omega\ket{x}. \ee
\end{definition}
\begin{definition}
The Weyl-Heisenberg displacement operators \gls{Du} are defined to be
\be \label{def-Du-prime}
D_{\bu} = \omega^{u_1 u_2/2}X^{u_1}Z^{u_2} \hskip 15mm \bu=\bmp u_1 \\ u_2 \emp, \ee
where the two components $u_1$ and $u_2$ are elements of $\fd{Z}_p$ and $1/2$ denotes the inverse of 2 in that field.
\end{definition}
\begin{definition}
For any two vectors $\bu$ and $\bv$ in the discrete phase space $\fd{Z}_p^2$, the symplectic form \gls{Omega} is defined as
\be \Omega(\bu,\bv) = u_2 v_1 - u_1 v_2. \ee
\end{definition}
The phases $\omega^{u_1u_2/2}$ in \cref{def-Du-prime} allow us to write the group law in the form
\be \label{eq-grouplaw-p}
D_{\bu}D_{\bv} = \omega^{\Omega(\bu,\bv)}D_{\bu+\bv}.\ee

The symplectic form \gls{Omega} has a geometrical meaning: if we consider $\bu$ and $\bv$ as vectors on the real plane $\fd{R}^2$, whose components happen to be integers modulo $p$, then \gls{Omega} is the area of the parallelogram spanned by $\bu$ and $\bv$. Therefore it is sometimes also called the symplectic area. Let $G$ be a linear transformation on the phase space, i.e. $G$ is a $2\times2$ matrix whose entries are in $\fd{Z}_p$, and let
\be \Delta = \det(G).\ee
It can be seen that under this linear transformation, the symplectic form is scaled by a factor of $\Delta$:
\be \label{eq-OmegaG}
\Omega(G\bu,G\bv) = \Delta \Omega(\bu,\bv). 
\ee

Next, we want to define unitaries $U_S$, which are called symplectic unitaries for a reason that will become clear in a moment, so that their action on the Weyl-Heisenberg displacement operators takes the form
\be U_S D_{\bu} U_S^{-1} = D_{S\bu}.\ee
Two conditions are required for this \cite{Appleby2009}. The first condition is that $S$ has to be a linear transformation. The second one comes from the group law \cref{eq-grouplaw-p}, from which it can be seen that $S$ has to preserve the symplectic form. In view of \cref{eq-OmegaG}, we deduce that $S$ must have determinant one and must therefore belong to the symplectic group $\SL(2,\fd{Z}_p)$. Their faithful (as opposed to only projective) unitary representation is given by \cite{Appleby2009}
\be \label{def-US-prime}
S = \bmp \alpha & \beta \\ \gamma & \delta \emp 
\hskip 3mm \rightarrow \hskip 3mm 
U_S =
\begin{cases} l(\alpha ) \sum_x\omega^{\alpha \gamma x^2/2}\ket{\alpha x}\bra{x} &\text{if } \beta = 0 \\
\frac{e^{i\phi}}{\sqrt{d}}\sum_{x,y} \omega^{\frac{\delta x^2 - 2xy + \alpha y^2}{2\beta}} \ket{x}\bra{y} 
&\text{if }\beta \neq 0, 
\end{cases}
\ee
where $\det(S) = \alpha \delta - \beta \gamma = 1$ mod $d$,
\be \label{def-Legendre}
l(x) = 
\begin{cases}\phantom{-}1 &\text{ if } x \in {\bf Q}\\
 -1 &\text{ if } x \in {\bf N}\\
\phantom{-}0 &\text{ if } x = 0, \end{cases}
\ee
and
\be
e^{i\phi} =
\begin{cases}
(-1)^k l(-\beta) & \text{if } d = 4k+1\\
i(-1)^{k+1} l(-\beta) & \text{if } d = 4k+3.
\end{cases}
\ee
In number theory, \gls{Legendre} is known as the Legendre symbol. \gls{Q} is the set of quadratic residues, i.e. elements of $\fd{Z}_p$ that can be written as the square of another non-zero element, and \gls{N} is the set of quadratic non-residues. 

The Clifford group then consists of all possible products of symplectic unitaries and displacement operators $U_S D_\bu$. In other words, it is isomorphic to the semidirect product of the symplectic group and the Weyl-Heisenberg group
\be
\SL(2,\fd{Z}_p) \ltimes \fd{Z}_p^2 \ .
\ee

\subsection{In odd prime power dimensions $d=p^n$}\label{sub:Clifford-primepower}
When the dimension $d = p^n$ is an odd prime power, we use the elements of the field $\fd{F}_d$ to label the elements of the Weyl-Heisenberg group just like in the odd prime case. However, $\Fd$ in general contains abstract elements which are not ordinary numbers, so the definitions in the previous section have to be slightly adjusted (since, for example in the definition of $D_\bu$ in \cref{def-Du-prime}, one cannot raise $\omega$, $X$, and $Z$  to the power of an abstract field element).\\
\begin{definition}
With $\omega=e^{2\pi i/p}$ still being a primitive $p$-th root of unity, we now define the shift and the phase operators as follows:
\be X_u\ket{x} = \ket{x+u} \hskip 15mm Z_u\ket{x} = \omega^{\tr(xu)}\ket{x}, \ee
where $x,u \in \fd{F}_d$ and the field theoretic trace $\tr(x)$ is a mapping from $\fd{F}_d$ to the ground field $\fd{F}_p$ defined in \cref{def-tr-primepower} and specifically denoted in the lower case to be distinguished from the trace of linear operators.\\
\end{definition}
\begin{definition} The Weyl-Heisenberg displacement operators are defined to be
\be
D_{\bu} = \omega^{\tr(u_1 u_2 /2)}X_{u_1}Z_{u_2},
\ee
where the two components $u_1$ and $u_2$ of $\bu$ are elements in $\fd{F}_d$.
\end{definition}

The WH group constructed this way is isomorphic to the direct product of $n$ copies of the WH group defined in the previous section for prime dimension $p$ \cite{Appleby2009}. The Clifford group, just like in the previous case, is the semidirect product of the symplectic group $\SL(2,\fd{F}_d)$ and the WH group. A faithful unitary representation of $\SL(2,\fd{F}_d)$ is given by \cite{Appleby2009}
\be \label{def-US-primepower}
S = \bmp \alpha & \beta \\ \gamma & \delta \emp 
\hskip 3mm \rightarrow \hskip 3mm 
U_S =
\begin{cases} l(\alpha ) \sum_{x\in\fd{F}_d}\omega^{\tr(\alpha \gamma x^2/2)}\ket{\alpha x}\bra{x} &\text{if } \beta = 0 \\
\frac{e^{i\phi}}{\sqrt{d}}\sum_{x,y \in \fd{F}_d} \omega^{\tr(\frac{\delta x^2 - 2xy + \alpha y^2}{2\beta})} \ket{x}\bra{y} 
&\text{if }\beta \neq 0, 
\end{cases}
\ee
where 
\be e^{i\phi} = (-i)^{-n(p+3)/2}l(-\beta) \ee
and \gls{Legendre} is again the Legendre symbol defined in \cref{def-Legendre}, but now for the field $\fd{F}_d$. One can  verify that this representation reduces to \cref{def-US-prime} when $d=p$.

\cleardoublepage 	
\phantomsection  
\renewcommand*{\bibname}{References}
\addcontentsline{toc}{chapter}{\textbf{References}}
\printbibliography

\glsaddallunused 

\end{document}